\documentclass{article}

\usepackage[english]{babel}
\usepackage{amsmath}
\usepackage{enumerate}

\usepackage{thm-restate}
\usepackage{amsthm} 
\usepackage[a4paper,top=3cm,bottom=3cm,left=3cm,right=3cm,marginparwidth=1.75cm]{geometry}

\usepackage[numbers]{natbib}
\usepackage{pifont}

\usepackage{amsmath}

\usepackage{amsfonts}
\usepackage{amssymb}
\usepackage{graphicx}
\usepackage{adjustbox}
\usepackage{diagbox}
\usepackage{booktabs}

\usepackage[colorlinks=true, allcolors=blue]{hyperref}
\usepackage{mathtools}

\usepackage{float}
\usepackage{xcolor}
\newtheorem{theorem}{Theorem}[section]
\newtheorem{lemma}[theorem]{Lemma}

\newtheorem{proposition}[theorem]{Proposition}

\newtheorem{corollary}[theorem]{Corollary}

\newtheorem{remark}[theorem]{Remark}
\newtheorem{example}[theorem]{Example}

\newtheorem{definition}[theorem]{Definition}

\numberwithin{equation}{section}

\def\E{{\mathbb{E}}}
\newcommand{\R}{{\mathbb R}}
\newcommand{\Mid}{{\ \Big|\ }}

\definecolor{blue0}{RGB}{0,77,153} 
\definecolor{red0}{RGB}{179,0,77} 
\definecolor{green0}{RGB}{134,219,76} 
\definecolor{gray0}{RGB}{84,97,110}

\usepackage{mathtools}
\usepackage{authblk}
\usepackage{comment}
\usepackage{bm}
\usepackage{xurl}
\mathtoolsset{showonlyrefs}
\hypersetup{pdftitle={Fourier-Laplace transforms in polynomial Ornstein-Uhlenbeck volatility models},pdfauthor={Eduardo Abi Jaber, Xiaoyuan (Shaun) Li, Xuyang Lin}}

\title{Fourier-Laplace transforms in polynomial Ornstein-Uhlenbeck volatility models}
\author[1]{Eduardo Abi Jaber\thanks{eduardo.abi-jaber@polytechnique.edu}}
\author[2]{Shaun (Xiaoyuan) Li\thanks{shaunlinz02@gmail.com}}
\author[3]{Xuyang Lin\thanks{xuyang.lin@nyu.edu}}
\affil[1]{CMAP, École Polytechnique, Palaiseau, France}
\affil[2]{Morgan Stanley, Paris, France}
\affil[3]{Department of Finance and Risk Engineering, New York University, Brooklyn, USA}

\begin{document}
\maketitle

\begin{abstract}
We consider the Fourier-Laplace transforms of a broad class of polynomial Ornstein-Uhlenbeck (OU) volatility models, including the well-known Stein-Stein, Sch\"obel-Zhu, one-factor Bergomi, and the recently introduced quintic OU models motivated by the SPX-VIX joint calibration problem. We show the connection between the joint Fourier-Laplace functional of the log-price and the integrated variance, and the solution of an infinite-dimensional system of Riccati equations. Next, under some non-vanishing conditions on the Fourier-Laplace transforms, we establish an existence result for such Riccati equations and show how the Fourier-Laplace functional can be seen as the limit of a sequence of entire functions. On the practical side, we develop a numerical scheme to solve the stiff infinite-dimensional system of Riccati equations and demonstrate the efficiency and accuracy of the scheme for pricing SPX options and volatility swaps using Fourier and Laplace inversions, with specific examples of the quintic OU and the one-factor Bergomi models and their calibration to real market data.
\end{abstract}

\begin{description}
\item[Mathematics Subject Classification:] 60E10, 60G07, 60G20, 91G20, 91G60.
\item[JEL Classification:] C63, G10, G13.
\item[Keywords:] Stochastic volatility, Derivative pricing, Fourier methods, Riccati equations, SPX-VIX calibration.
\end{description}

\section{Introduction}
\label{intro}
Fourier inversion techniques hold a pivotal role in stochastic volatility modelling, particularly in the context of option pricing and hedging (\citet*{andersen2000jump,carr1999option,eberlein2010analysis,fang2009novel,lewis2001simple}, \citet*[Chaps.~8 and 10]{lipton2001mathematical}). They offer the dual advantage of significantly reducing computational time while maintaining a remarkable degree of accuracy, especially when compared to standard Monte Carlo methods. When it comes to model calibration, the applicability of Fourier methods is paramount, as thousands of derivatives across various maturities and strikes need to be evaluated simultaneously in real time.

Despite their numerical advantage, Fourier techniques have traditionally been confined to specific continuous stochastic volatility models where the characteristic function of the log-price is known in (semi-)closed form. These models usually exhibit Markovian affine structures in the sense of \citet*{duffie2003affine} such as the renowned Heston model by \citet{heston1993closed}, the Stein--Stein model by \citet{stein1991stock}, and the Schöbel--Zhu model by \citet{schobel1999stochastic}, together with some of their non-Markovian Volterra counterparts (\citet*{abi2019lifting,abi2022characteristic,abi2019affine,cuchiero2020generalized,el2019characteristic,gatheral2019affine}). The key ingredient in all these models is to compute the characteristic function of the log-price by solving a specific system of deterministic Riccati equations.

More recently, Fourier techniques have also found applications in models of signature volatility in \citet*{abi2024signature} and \citet*{cuchiero2023joint}, where the volatility process is modelled as a linear functional of the path-signature of semimartingales (e.g.~a Brownian motion). In such models, certain characteristic functionals have been related to non-standard infinite-dimensional systems of Riccati ordinary differential equations (ODEs). However, there exists no general theory regarding the existence of solutions for such equations, except for the specific result in \citet*[Proposition~5.8]{cuchiero2023signature} which provides the existence of a solution to the infinite-dimensional ODE system fonew
r the case of the characteristic function of powers of a single Brownian motion modulo a non-vanishing condition on the characteristic function. The primary challenge comes from the intricate questions on analyticity of the logarithm of the characteristic function. Furthermore, numerically solving these equations poses challenges due to their stiffness and complexity. This forms our primary motivation: to establish theoretical results within a specific framework for a larger class of Riccati equations related to integrated quantities of power series of Ornstein--Uhlenbeck processes and to develop more suitable numerical schemes.

We demonstrate that Fourier techniques can be effectively extended to a broad class of flexible models previously considered infeasible as they fall beyond the conventional class of affine diffusions, including the celebrated one-factor Bergomi model (\citet{bergomi2005smile,dupire1992arbitrage}) that has been shown to fit well to the SPX smiles, and the recently introduced quintic Ornstein--Uhlenbeck (OU) model of \citet*{jaber2022quintic}, which has demonstrated remarkable capabilities in fitting jointly the SPX--VIX volatility surface for maturities between one week to three months, supported by extensive empirical studies on more than 10 years of data in \citet*{abi2024volatility,abi2022joint}.

We refer to the class of models covered in this paper as the class of polynomial OU models, where the volatility of the log-price, denoted by $\sigma$, is defined as a power series of an Ornstein--Uhlenbeck process. Our three main theoretical results regarding the joint Fourier--Laplace functional of the log-price and integrated variance process are:
\begin{enumerate}[(i)]
\item A verification result on the expression of the joint Fourier--Laplace functional in terms of a solution to an infinite-dimensional system of Riccati ODEs in Theorem~\ref{thm1.1}.
\item The existence of a solution to such a system of ODEs modulo a non-vanishing condition on the joint characteristic functional in Theorem~\ref{thm1.2}.
\item A limiting characterisation of the joint Fourier--Laplace functional via a sequence of entire functions in Theorem~\ref{T:Discrete}.
\end{enumerate}
On the practical side, we devise a numerical scheme to solve this infinite-dimensional (stiff) system of Riccati equations, making the application of Fourier pricing for models such as the one-factor Bergomi and the quintic OU models usable in practice. We demonstrate the efficiency and accuracy of the scheme for pricing SPX options and volatility swaps using Fourier and Laplace inversions. We also successfully calibrate using Fourier techniques the quintic OU and the one-factor Bergomi models on real market data to highlight the stability and robustness of our numerical method across a wide range of realistic parameters values.

The paper is organised as follows. In Section~\ref{intro_poly_ou_model}, we introduce the class of polynomial OU volatility models. In Section~\ref{S3}, we present our three main theoretical results. In Section~\ref{S:numerical}, we design a numerical scheme for solving the infinite-dimensional Riccati ODE system, test our scheme on pricing SPX vanilla options and volatility swaps for the one-factor Bergomi and quintic OU volatility models for various maturities, and calibrate both models using real market data. Proofs of the results discussed in Section~\ref{S3} are collected in Section~\ref{section5} onwards.

\textbf{Notations.} $\mathbb{N}_0$ denotes the set of nonnegative integers, so that $\mathbb{N}=\{1,2,\ldots\}$. For a power series $q(x)$, we denote by $q_k$ the coefficients of $x^k$, i.e., $q(x)=\sum_{k=0}^{\infty}q_kx^k$. We also denote by $|q|$ the \textbf{absolute power series} of $q$, i.e., the corresponding $|q|(x)=\sum_{k=0}^{\infty}|q_k|x^k$. We remark that $|q|_k=|q_k|$. It is well known that if $p$ and $q$ have an infinite radius of convergence, i.e., $\sum_{k=0}^{\infty}|p_k||x|^k<\infty$ for all $x\in\mathbb{R}$, then $p(x)q(x)=\sum_{k=0}^{\infty}(p*q)_kx^k$ also has an infinite radius of convergence, with the notation $(p*q)_k:=\sum_{\ell=0}^k p_\ell q_{k-\ell}$ for the convolution. Unless stated otherwise, we assume that all power series in this paper have an infinite radius of convergence. Moreover, we specify that the term \textit{continuity} means \textit{joint continuity} when applied to a function taking $(t,x)$ as variables. For example, a function $f:[0,T]\times\mathbb{R}\to\mathbb{C}$ is $C^\infty$ in $x$ if all its space derivatives in $x$ exist and are continuous in $(t,x)$. We also adopt for the partial derivatives the notations $f_t=\partial f/\partial t$, $f_x=\partial f/\partial x$, $f_{xx}=\partial^2f/\partial x^2$ when they exist. We say that a function $f:\mathbb{R}\to\mathbb{C}$ is (real-)entire on $\mathbb{R}$ if it is equal to a power series $q$ (with an infinite radius of convergence), that is, for all $x\in\mathbb{R}$, $f(x)=\sum_{k=0}^{\infty}q_kx^k$. We say that it is (complex-)entire when $\mathbb{R}$ is replaced by $\mathbb{C}$. Clearly, any real-entire function admits a complex-entire extension. When there is no ambiguity, we simply use the word entire. We say that a function does not vanish if it is non-zero at each point of its domain.

\section{Polynomial Ornstein--Uhlenbeck volatility models}
\label{intro_poly_ou_model}
We consider the class of polynomial Ornstein--Uhlenbeck (OU) models for a (discounted) stock price $S$ with the stochastic volatility process $\sigma$ expressed as a power series of an OU process $X$, i.e.,
\begin{equation}\label{polynomial_model}
\begin{aligned}
dS_t&=S_t\sigma_t\,dB_t,\qquad S_0>0,\\
\sigma_t&=g_0(t)p(X_t),\qquad p(x)=\sum_{k=0}^{\infty}p_kx^k,\\
dX_t&=(a+bX_t)\,dt+c\,dW_t,\qquad X_0\in\mathbb{R},
\end{aligned}
\end{equation}
with $a,b,c\in\mathbb{R}$, $c\neq0$, $p_k\in\mathbb{R}$ for all $k\in\mathbb{N}_0$, $B=\rho W+\sqrt{1-\rho^2}\,W^\perp$ and $\rho\in[-1,1]$. Here $(W,W^\perp)$ is a two-dimensional Brownian motion on a risk-neutral filtered probability space $(\Omega,\mathcal{F},(\mathcal{F}_t)_{t\geq0},\mathbb{Q})$ satisfying the usual conditions. The deterministic bounded input curve $g_0:[0,T]\to\mathbb{R}_+$ allows the model to match certain term structures of volatility observed on the market, e.g.~the forward variance curve. Indeed, the instantaneous forward variance at time $t$ is $\xi_0(t):=\mathbb{E}[\sigma_t^2\mid\mathcal{F}_0]$ and the forward variance curve is $t\mapsto\mathbb{E}[\int_0^t\sigma_s^2\,ds]$. So the model can exactly reproduce the forward variance curve if we set $g_0(t)=\sqrt{\xi_0(t)/\mathbb{E}[p^2(X_t)]}$, because then clearly
\begin{align}\label{eq:fwdvarcalib}
\mathbb{E}\left[\int_0^t\sigma_s^2\,ds\right]=\int_0^t\xi_0(s)\,ds,\qquad t\geq0.
\end{align}
The real-valued coefficients $(p_k)_{k\geq0}$ are such that the power series $p$ has an infinite radius of convergence, and $\int_0^T\mathbb{E}[p^2(X_s)]\,ds<\infty$ so that the stochastic integral is well defined. This is the case for instance when $p$ is a finite polynomial or the exponential function. These two specifications already provide several interesting models used in practice:
\begin{itemize}
\item Stein--Stein \cite{stein1991stock} and Schöbel--Zhu \cite{schobel1999stochastic} model: $p(x)=x$;
\item Bergomi model \cite{bergomi2005smile,dupire1992arbitrage}: $p(x)=\exp(x)$;
\item Quintic OU model \cite{jaber2022quintic}: $p(x)=p_0+p_1x+p_3x^3+p_5x^5$,\quad $p_0,p_1,p_3,p_5\geq0$.
\end{itemize}

\begin{remark}\label{rmk:stock_martingale}
The martingale property of the stock price process $S$ is a key desirable property. In particular, for a constant input curve $g_0$, by \citet*[Remark~3.1.1]{jaber2025martingalepropertymomentexplosions}, if $p$ is a polynomial of degree $N$, then a necessary and sufficient condition for $S$ to be a true martingale is that $\rho$ and $p$ satisfy one of the following: (i) $N=1$; (ii) $\rho=0$; or (iii) $N\geq2$ is odd, $\rho\neq0$, and $\rho p_N\leq0$, where $p_N$ denotes the leading coefficient of $p$. For infinite nonnegative power series such as $p(x)=\exp(x)$, a necessary and sufficient condition for martingality of $S$ is given by $\rho\leq0$; see for instance \citet{lions2007correlations}.
\end{remark}

\begin{definition}\label{2.3}
A power series $p:x\mapsto\sum_{k=0}^{\infty}p_kx^k$ is said to be \textbf{negligible to double factorial} if the power series $\sum_{k=0}^{\infty}(k-1)!!p_kx^k$ has an infinite radius of convergence, i.e.,
\[
\limsup_{k\to\infty}\bigl(|p_k|(k-1)!!\bigr)^{1/k}=0.
\]
The double factorial $n!!$ is defined as
\[
n!!:=\begin{cases}
\displaystyle\prod_{k=0}^{(n-2)/2}(n-2k)&\text{if $n$ is even},\\[1ex]
\displaystyle\prod_{k=0}^{(n-1)/2}(n-2k)&\text{if $n$ is odd},
\end{cases}
\]
with the convention $(-1)!!=0!!=1$.
\end{definition}

\begin{remark}
A power series $p:x\mapsto\sum_{k=0}^{\infty}p_kx^k$, $x\in\mathbb{R}$, with an infinite radius of convergence can be naturally extended to a power series $\tilde p:z\mapsto\sum_{k=0}^{\infty}p_kz^k$, $z\in\mathbb{C}$. The function $\tilde p$ is again an entire function on $\mathbb{C}$. Using the formulas for order and type of an entire function, along with Stirling's formula, $p$ being negligible to double factorial is equivalent to $\tilde p$ being an entire function of order strictly less than $2$, or of order $2$ and of type $0$.
\end{remark}

By Definition~\ref{2.3}, a power series $p$ negligible to double factorial also has an infinite radius of convergence. This allows us to show in Proposition~\ref{def1} that our class of polynomial Ornstein--Uhlenbeck volatility models \eqref{polynomial_model} is well posed. For this, we need a simple lemma about absolute power series, which will be useful later.

\begin{lemma}\label{B.6}
If $p$ is a power series with an infinite radius of convergence, then:
\begin{enumerate}[(i)]
\item For all $x\in\mathbb{R}$, $|p(x)|\leq|p|(|x|)$ and $|p|$ is monotonically increasing on $\mathbb{R}_+$.
\item For all $x,y,z\geq0$, $|p|(x+y+z)\leq\bigl(|p|(3x)+|p|(3y)+|p|(3z)\bigr)$.
\end{enumerate}
\end{lemma}
\begin{proof}
(i) is obvious; for (ii), it suffices to notice that
\[
(x+y+z)^k\leq\frac{1}{3}\bigl((3x)^k+(3y)^k+(3z)^k\bigr)\leq\bigl((3x)^k+(3y)^k+(3z)^k\bigr).
\]
\end{proof}

\begin{proposition}\label{def1}
Let $g_0:[0,T]\to\mathbb{R}$ be a measurable and bounded function. Let $p$ be a power series with an infinite radius of convergence and such that $p^2$ is negligible to double factorial. Then
\begin{equation}\label{eq:intp2}
\int_0^T\mathbb{E}[p^2(X_t)]\,dt<\infty,
\end{equation}
so that the stochastic integral $\int_0^\cdot\sigma_s\,dB_s$, with $\sigma_s=g_0(s)p(X_s)$, is well defined and
\begin{equation}\label{eq:Sexplicit}
S_t=S_0\exp\left(\int_0^t\sigma_s\,dB_s-\frac12\int_0^t\sigma_s^2\,ds\right),\qquad t\geq0,
\end{equation}
is the unique strong solution to \eqref{polynomial_model}.
\end{proposition}
\begin{proof}
Under \eqref{eq:intp2}, it is straightforward to obtain that the unique strong solution $S$ to \eqref{polynomial_model} is given by \eqref{eq:Sexplicit}. It suffices to prove \eqref{eq:intp2}. For this, we note that the explicit solution of $X$ is given by
\begin{align}\label{def}
X_t=e^{bt}X_0+a\int_0^t e^{b(t-s)}\,ds+c\int_0^t e^{b(t-s)}\,dW_s.
\end{align}
We set $w(t):=a\int_0^t e^{b(t-s)}\,ds$, which is a deterministic continuous function of $t$, and $\widetilde W_t:=c\int_0^t e^{b(t-s)}\,dW_s$, which is Gaussian $\sim\mathcal{N}(0,c^2(e^{2bt}-1)/(2b))$ for $t\in[0,T]$. Applying Lemma~\ref{B.6} yields
\begin{align*}
\left|\int_0^T |p^2|\bigl(e^{bs}X_0+w(s)+\widetilde W_s\bigr)\,ds\right|
&\leq\int_0^T |p^2|\bigl(e^{bs}|X_0|+|w(s)|+|\widetilde W_s|\bigr)\,ds\\
&\leq\int_0^T\sum_{k=0}^{\infty}|p^2|_k\bigl((3e^{bs}|X_0|)^k+|3w(s)|^k+|3\widetilde W_s|^k\bigr)\,ds,
\end{align*}
and since $e^{bs}\leq e^{|b|T}$ and $w(s)$ is bounded in $[0,T]$, there exists a constant $C_T\geq0$ such that
\begin{align*}
&\int_0^T\sum_{k=0}^{\infty}|p^2|_k\bigl((3e^{bs}|X_0|)^k+|3w(s)|^k+|3\widetilde W_s|^k\bigr)\,ds\\
\leq&\int_0^T\sum_{k=0}^{\infty}|p^2|_k\bigl((3e^{|b|T}|X_0|)^k+(3C_T)^k+|3\widetilde W_s|^k\bigr)\,ds.
\end{align*}
By assumption, $p^2$ has an infinite radius of convergence; then so does $|p^2|$, and so
\[
\sum_{k=0}^{\infty}|p^2|_k(3e^{|b|T}|X_0|)^k<\infty,\qquad
\sum_{k=0}^{\infty}|p^2|_k(3C_T)^k<\infty,
\]
and thus we only need to prove
\[
\int_0^T\mathbb{E}[|p^2|(|3\widetilde W_t|)]\,dt<\infty.
\]
The random variable $3\widetilde W_t$ is Gaussian with bounded variance in $[0,T]$; so there exists a constant $\widehat C_T>0$ such that $\mathbb{E}[|3\widetilde W_t|^k]\leq\mathbb{E}[|\widehat C_T Z|^k]$ for all $t\in[0,T]$, where $Z$ is a standard normal variable. Because $\mathbb{E}[|Z|^k]\leq(k-1)!!$, we get
\[
\mathbb{E}[|p^2|(\widehat C_T|Z|)]\leq\sum_{k=0}^{\infty}|p^2|_k(k-1)!!\widehat C_T^k<\infty,
\]
since $p^2$ is negligible to double factorial. This completes the proof.
\end{proof}

\begin{remark}
In Proposition~\ref{def1}, $p^2$ needs to be negligible to double factorial. As we shall see later in Lemma~\ref{lemB.12} below, the square of a power series negligible to double factorial is also negligible to double factorial. Hence $p$ negligible to double factorial is sufficient for Proposition~\ref{def1} to hold.
\end{remark}

\begin{example}
Polynomials of finite degree (e.g.~the Stein--Stein and the quintic OU models) and the exponential function (e.g.~the one-factor Bergomi model) are clearly negligible to double factorial. For the exponential function $p(x)=e^{\varepsilon x}$, $\varepsilon>0$, it suffices to observe that $\lim_{k\to\infty}(\varepsilon^k(k-1)!!/k!)^{1/k}=\lim_{k\to\infty}(\varepsilon^k/k!!)^{1/k}=0$. However, the function $p(x)=e^{\varepsilon x^2}$ is not negligible to double factorial. Indeed, a simple calculation leads to $\limsup_{k\to\infty}(|p_k|(k-1)!!)^{1/k}=\sqrt{2\varepsilon}$. In particular, if $b>0$, then the expectation of $p(X_t)$ does not exist for all $t\geq1/(2c^2\varepsilon)$.
\end{example}

\section{The joint Fourier--Laplace transform}\label{S3}
Our goal is to compute the joint Fourier--Laplace transform of the log-price and integrated variance, i.e.
\[
\mathbb{E}\left[\exp\left(\int_t^T g_1(T-s)\,d(\log S_s)+\int_t^T g_2(T-s)\sigma_s^2\,ds\right)\Mid\mathcal{F}_t\right],\qquad t\leq T,
\]
for some measurable complex-valued functions $g_1,g_2:[0,T]\to\mathbb{C}$.

\begin{remark}\label{g1g2}
If $g_1,g_2$ are measurable functions $[0,T]\to\mathbb{C}$ such that $\Re(g_1)=0$, $\Re(g_2)\leq0$ and $p^2$ is negligible to double factorial, then the exponential above exists and has a modulus of at most $1$, so that the conditional expectation is well defined.
\end{remark}

By the Markov property of $X$, the above conditional expectation can be reduced to computing the deterministic measurable function $F:[0,T]\times\mathbb{R}\to\mathbb{C}$ given by
\begin{equation}\label{eq:defF}
\begin{aligned}
F(t,x)&:=\mathbb{E}\left[\exp\left(\int_t^T g_1(T-s)\,d(\log S_s)+\int_t^T g_2(T-s)\sigma_s^2\,ds\right)\Mid X_t=x\right].
\end{aligned}
\end{equation}
We present three main results related to the computation of the characteristic functional $F$ in \eqref{eq:defF}. The first can be seen as a verification result and uncovers an affine structure in infinite dimension in terms of the powers $(1,X,X^2,X^3,\ldots)$ by making a connection with the infinite-dimensional system of Riccati equations (Theorem~\ref{thm1.1} and Corollary~\ref{Cor:Riccati}). The second result is concerned with the existence of a solution to that system (Theorem~\ref{thm1.2}). The last result characterises the joint Fourier--Laplace functional as the limit of a sequence of entire functions (Theorem~\ref{T:Discrete}).

\subsection{A verification result}
Our first main result uncovers an affine structure in infinite dimension expressed in powers $(1,X,X^2,\ldots)$ and makes a connection with the infinite-dimensional system of Riccati equations given by
\begin{align}
\psi_k'(t)&=\left(g_2(t)+\frac{g_1(t)}{2}(g_1(t)-1)\right)g_0^2(T-t)(p*p)_k+bk\psi_k(t)\nonumber\\
&\quad+a(k+1)\psi_{k+1}(t)+\frac{c^2(k+2)(k+1)}{2}\psi_{k+2}(t)+\frac{c^2}{2}(\widetilde\psi(t)*\widetilde\psi(t))_k\nonumber\\
&\quad+c\rho g_1(t)g_0(T-t)(p*\widetilde\psi(t))_k,\qquad\widetilde\psi_k(t):=(k+1)\psi_{k+1}(t),\label{eq:Ric}\\
\psi_k(0)&=0.\label{eq:Ric_init}
\end{align}

\begin{theorem}\label{thm1.1}
Let $g_0:[0,T]\to\mathbb{R}$, $g_1,g_2:[0,T]\to\mathbb{C}$ be measurable and bounded functions. Let $p$ be a power series with an infinite radius of convergence and such that $p^2$ is negligible to double factorial. Assume that there exists a continuously differentiable solution $(\psi_k)_{k\geq0}$ to the infinite-dimensional system \eqref{eq:Ric} and \eqref{eq:Ric_init} of Riccati equations such that the power series $\sum_{k\geq0}\sup_{t\in[0,T]}|\psi_k(t)|x^k$ has an infinite radius of convergence. Define the process
\begin{align}\label{eq:U}
U_t=\sum_{k\geq0}\psi_k(T-t)X_t^k+\int_0^t g_1(T-s)\,d(\log S_s)+\int_0^t g_2(T-s)\sigma_s^2\,ds.
\end{align}
Then the process $M:=\exp(U)$ is a local martingale. If in addition $M$ is a true martingale, then we have for the joint characteristic functional $F$ in \eqref{eq:defF} the expression
\begin{equation}\label{eq:charfun}
F(t,x)=\exp\left(\sum_{k\geq0}\psi_k(T-t)x^k\right),\qquad t\leq T.
\end{equation}
\end{theorem}
\begin{proof}
The proof is given in Section~\ref{section5}.
\end{proof}

\begin{remark}
Theorem~\ref{thm1.1} is in the spirit of \citet*[Theorem~5.3]{cuchiero2023signature}, which provides a similar verification result for the characteristic function of (real-)entire functions of solutions to stochastic differential equations with (real-)entire coefficients. In contrast to \cite[Theorem~5.3]{cuchiero2023signature}, Theorem~\ref{thm1.1} deals with time-integrated quantities.
\end{remark}

A possible strategy for obtaining the representation \eqref{eq:charfun} would involve verifying the assumptions outlined in Theorem~\ref{thm1.1}. These include ensuring the existence of a solution for the system \eqref{eq:Ric} and \eqref{eq:Ric_init} of Riccati equations, such that the power series $(\psi_k)_{k\geq0}$ has an infinite radius of convergence, together with proving that the local martingale $M$ is a true martingale. The latter can typically be achieved by arguing for instance that $\sum_{k\geq0}\Re(\psi_k(t))x^k\leq0$ to obtain that $M$ is uniformly bounded by $1$ whenever $g_1,g_2$ are purely imaginary for instance.

In the specific case of the Stein--Stein model, i.e., when $p$ is an affine function, these assumptions are comparatively easier to confirm, as highlighted in the next example.

\begin{example}\label{ex:stein_stein}
In the classical Stein--Stein model \cite{stein1991stock}, the volatility is defined as
\[
\sigma_t=g_0(t)X_t.
\]
This is equivalent to \eqref{polynomial_model} by setting $p_1=1$ and $p_k=0$ for $k\neq1$. Note that the convolution term is $(p*p)_k=1$ for $k=2$ and zero otherwise. In this case, the infinite-dimensional system \eqref{eq:Ric} and \eqref{eq:Ric_init} of Riccati equations reduces to
\begin{align}
\psi_0'(t)&=a\psi_1(t)+c^2\psi_2(t)+\frac{c^2}{2}\psi_1^2(t),\nonumber\\
\psi_1'(t)&=b\psi_1(t)+2a\psi_2(t)+2c^2\psi_1(t)\psi_2(t)+\rho g_1(t)g_0(T-t)c\psi_1(t),\nonumber\\
\psi_2'(t)&=\left(g_2(t)+\frac{g_1(t)}{2}(g_1(t)-1)\right)g_0^2(T-t)\nonumber\\
&\quad+2b\psi_2(t)+2c^2\psi_2^2(t)+2\rho g_1(t)g_0(T-t)c\psi_2(t),\label{eq:Ric_stein}\\
\psi_k(0)&=0,\qquad k\in\{0,1,2\},\label{eq:Ric_init_stein}\\
\psi_k&\equiv0,\qquad k\geq3,\nonumber
\end{align}
which is a system of standard finite-dimensional Riccati equations, for which existence and differentiability of a solution are well established in the literature (see for example \citet*{cuchiero2011affine,duffie2003affine,spreij2012affine}) under the condition
\begin{equation}\label{eq:Re_g}
\Re\left(g_2(t)+\frac{g_1(t)}{2}(g_1(t)-1)\right)\leq0,\qquad t\leq T,
\end{equation}
which is satisfied whenever $g_1(t),g_2(t)\in\mathbb{C}$, $\Re(g_1)=0$ and $\Re(g_2)\leq0$. Hence the characteristic functional of the classical Stein--Stein model is affine in $(1,X,X^2)$, i.e.,
\begin{equation}\label{eq:charfun_stein}
F(t,x)=\exp\bigl(\psi_0(T-t)+\psi_1(T-t)x+\psi_2(T-t)x^2\bigr),\qquad t\leq T,
\end{equation}
for which $\psi_0,\psi_1$ and $\psi_2$ can even be found explicitly when $g_0,g_1$ and $g_2$ are constants, as shown by \citet{lord2006rotation}.

In addition to the existence of a solution to the finite-dimensional Riccati equations, the other assumptions in Theorem~\ref{thm1.1} can also be verified for the Stein--Stein model. From \eqref{eq:charfun_stein}, it is obvious that the finite sum $\sum_{k=0}^2\sup_{t\in[0,T]}|\psi_k(t)|x^k$ has an infinite radius of convergence in $x$. Furthermore, the process $(\exp(U_t))_{t\leq T}$ is a true martingale in the Stein--Stein model as shown below.
\end{example}

\begin{proposition}\label{prop:stein_martingale}
Fix functions $g_0:[0,T]\to\mathbb{R}$, $g_1,g_2:[0,T]\to\mathbb{C}$ measurable and bounded such that $\Re(g_1)=0$ and $\Re(g_2)\leq0$. If in addition $p(x)=x$, then the process $(\exp(U_t))_{t\leq T}$ in Theorem~\ref{thm1.1} is a true martingale.
\end{proposition}
\begin{proof}
By $\Re(g_1)=0$, $\Re(g_2)\leq0$ and \citet*[Theorem~2.4]{abi2023reconciling}, we deduce that $\Re(\psi_2)\leq0$. Next, $|\psi_0|$ and $|\psi_1|$ are both bounded by positive constants due to continuity in $t\in[0,T]$. Recall that $X_t=e^{bt}X_0+a\int_0^t e^{b(t-s)}\,ds+ce^{bt}Y_t$ with $Y:=\int e^{-bs}\,dW_s$ a martingale. There exist positive constants $C_0,C_1$ such that
\begin{align*}
|\exp(U_t)|&\leq\exp\bigl(|\psi_0(T-t)|+|\psi_1(T-t)||X_t|\bigr)\leq C_0\bigl(\exp(2D_t^+)+\exp(2D_t^-)\bigr),
\end{align*}
where $D_t^\pm:=\pm(C_1/2)Y_t-(C_1^2/8)((1-e^{-2bt})/(2b))$, with the convention $(1-e^{-2bt})/(2b):=t$ if $b=0$. Then $\exp(D^\pm)$ are true martingales by Novikov's condition. Moreover, $\mathbb{E}[\exp(2D_t^\pm)]=\exp\bigl((C_1^2/4)((1-e^{-2bt})/(2b))\bigr)$. Hence by Doob's $L^2$ inequality,
\begin{align*}
\mathbb{E}\left[\sup_{0\leq r\leq T}\exp(2D_r^\pm)\right]
&=\mathbb{E}\left[\sup_{0\leq r\leq T}|\exp(D_r^\pm)|^2\right]\leq4\mathbb{E}[\exp(2D_T^\pm)]=4\exp\bigl((C_1^2/4)((1-e^{-2bT})/(2b))\bigr).
\end{align*}
Together with the preceding bound on $\exp(U)$, we get $\mathbb{E}[\sup_{0\leq t\leq T}|\exp(U_t)|]<\infty$. As $\exp(U)$ is a local martingale, it must be a true martingale on $[0,T]$.
\end{proof}

In summary, obtaining the representation \eqref{eq:charfun_stein} under the Stein--Stein model can be attributed to the finite number of terms $(\psi_0,\psi_1,\psi_2)$ of the involved Riccati equations. However, when dealing with an infinite sum, a notable challenge arises in proving that the log-characteristic functional $\log F$ is entire in the variable $x$ on $\mathbb{R}$. In Section~\ref{S_existence}, we show how to generate a solution for the infinite-dimensional system of Riccati equations modulo a non-vanishing condition on $F$. In Section~\ref{S_arpprox}, we express the Fourier--Laplace transform as the limit of a sequence of entire functions, thanks to the Gaussianity of the process $X$.

In Section~\ref{S:numerical}, we numerically illustrate the validity of Fourier pricing in the context of the Bergomi and quintic OU models by using the representation \eqref{eq:charfun}. Readers interested in the numerical implementation can jump directly to Section~\ref{S:numerical}.

\subsection{Existence of a solution to the Riccati equations}\label{S_existence}
Our second main result generates a solution to the infinite-dimensional system \eqref{eq:Ric} and \eqref{eq:Ric_init} of Riccati equations by differentiating the logarithm of the characteristic functional $F$, i.e.,
\[
\psi_k(t):=\left.\frac{1}{k!}\partial_x^k\log F(T-t,x)\right|_{x=0},\qquad t\geq0.
\]
This requires the logarithm of a complex-valued function; see Appendix~\ref{A:complexlog} for its precise definition.

\begin{theorem}\label{thm1.2}
Fix $g_0:[0,T]\to\mathbb{R}$, $g_1,g_2:[0,T]\to\mathbb{C}$ all continuously differentiable such that $\Re(g_1)=0$ and $\Re(g_2)\leq0$. Let $p$ be a power series negligible to double factorial. Then $F$ in \eqref{eq:defF} is well defined for any $t\leq T$, $x\in\mathbb{R}$, and continuous in $(t,x)$. If in addition $F$ does not vanish on $[0,T]\times\mathbb{R}$, then $\log F$ can be defined as in Definition~\ref{defB.2}. Furthermore, $\log F$ is $C^\infty$ in $x$ and $C^1$ in $t$. In particular, the family of functions
\begin{equation}\label{eq:famriccati}
\psi_k(t)=\left.\frac{1}{k!}\partial_x^k\log F(T-t,x)\right|_{x=0},\qquad t\geq0,\ k\geq0,
\end{equation}
solves the system \eqref{eq:Ric} and \eqref{eq:Ric_init} of Riccati ODEs.
\end{theorem}
\begin{proof}
The proof is given in Section~\ref{section6}.
\end{proof}

Theorem~\ref{thm1.2} establishes the existence of a solution to the Riccati ODEs \eqref{eq:Ric} and \eqref{eq:Ric_init} when the coefficients are real as shown in the following corollary.

\begin{corollary}\label{Cor:Riccati}
Fix $g_0:[0,T]\to\mathbb{R}$, $g_1=0$ and $g_2:[0,T]\to\mathbb{R}_-$ all continuously differentiable. Let $p$ be a power series negligible to double factorial. Then $\log F$ is $C^\infty$ in $x$ and $C^1$ in $t$ and the family $(\psi_k)_{k\geq0}$ of functions given by \eqref{eq:famriccati} solves the system \eqref{eq:Ric} and \eqref{eq:Ric_init} of Riccati ODEs.
\end{corollary}
\begin{proof}
In this case, $F$ reduces to the Laplace transform of the integrated variance, i.e.,
\[
F(t,x)=\mathbb{E}\left[\exp\left(\int_t^T g_2(T-s)\sigma_s^2\,ds\right)\Mid X_t=x\right],
\]
which shows that $F>0$ and that it trivially satisfies the non-vanishing condition on $[0,T]\times\mathbb{R}$ in Theorem~\ref{thm1.2}. So an application of Theorem~\ref{thm1.2} yields the result.
\end{proof}

\begin{remark}
Theorem~\ref{thm1.2} is in the spirit of \citet*[Proposition~5.8]{cuchiero2023signature} which provides the existence of a solution to the infinite-dimensional system of ODEs for the case of the characteristic function of powers of a single Brownian motion, modulo similar assumptions. In contrast, Theorem~\ref{thm1.2} deals with more involved time-integrated quantities and requires a more delicate analysis.
\end{remark}

Although Theorem~\ref{thm1.2} does not establish that $\log F(t,\cdot)$ is entire, its analyticity can be inferred from the properties of solutions to parabolic partial differential equations (PDEs).

\begin{remark}\label{rmk:pde}
To prove Theorem~\ref{thm1.2}, we establish in Theorem~\ref{thmB.1} for the function $f(t,x)=F(t,x)\exp(v(t)q(x))$ the PDE
\[
\begin{cases}
\displaystyle f_t(t,x)+f_x(t,x)(a+bx)+\frac12c^2f_{xx}(t,x)+f(t,x)\sum_{i=1}^3 u_i(t)p_i(x)=0,\\[1ex]
f(T,x)=\exp\bigl(v(T)q(x)\bigr).
\end{cases}
\]
Here $v,u_i$ are continuous and $p_i$ are analytic functions as defined in Lemma~\ref{2.16}, with $a+bx,p_i(x)$ that are bounded, Hölder-continuous and analytic in the variable $x$ within any open bounded set of $\mathbb{R}$. We can thus apply \citet[Theorem~6.2 in Section~II.2.]{eidelman1969parabolic} to show that both $f(t,\cdot)$ and $F(t,\cdot)$ are analytic. Since analyticity is a local property, if $F$ does not vanish, then $\log F(t,\cdot)$ is also analytic. Therefore we can deduce a local representation: for any $t\leq T$, there exists $r_t>0$ such that
\[
F(t,x)=\exp\left(\sum_{k=0}^{\infty}\psi_k(T-t)x^k\right),\qquad x\in(-r_t,r_t),
\]
where the $(\psi_k)_{k\geq0}$ are defined by \eqref{eq:famriccati}. In our case, both $p_i$ and $q$ are not only analytic but also entire. Extending the proof of analyticity using PDE techniques to establish that $F(t,\cdot)$ is entire might be possible, and that would lead to a global representation as in \eqref{eq:charfun} as well as the uniqueness of the solution to the associated infinite-dimensional system of Riccati equations. Such a question holds independent interest for PDEs in its own right, but requires a more delicate analysis at the PDE level that lies beyond our probabilistic approach. Instead, in the next section, we show how $F$ (and $\log F$) can be interpreted as the limit of a sequence of entire functions.
\end{remark}

\subsection{The Fourier--Laplace transform as a limit of entire functions}\label{S_arpprox}
In this section, our objective is to construct a sequence $(F_n)$ of entire functions that converges to $F$. In addition, with a natural non-vanishing condition, we can show that $\log F_n$ is also entire and admits a power series expansion. This allows us to interpret $F$ and $\log F$ as limits of entire functions.

The idea is to approximate the Riemann sum on the sample paths of the Ornstein--Uhlenbeck process $X$ to exploit the underlying Gaussian density in finite dimensions. To achieve this, we first need to get rid of the stochastic integrals and express them in terms of functions of $X_T$ and Lebesgue integrals of $(X_s)_{t\leq s\leq T}$. This is done by combining Itô's lemma and the conditioning trick of \citet{romano1997contingent} in the next result.

\begin{lemma}\label{2.16}
Fix a power series $p$ negligible to double factorial, and continuously differentiable functions $g_0:[0,T]\to\mathbb{R}$, $g_1,g_2:[0,T]\to\mathbb{C}$ with $\Re(g_1)=0$, $\Re(g_2)\leq0$. Then there exist power series $q,p_i$, $i=1,2,3$, with an infinite radius of convergence and functions $v,u_i$, $i=1,2,3$, such that
\[
F(t,x)=\mathbb{E}\left[\exp\left(\int_t^T\sum_{i=1}^3u_i(s)p_i(X_s)\,ds+v(T)q(X_T)-v(t)q(X_t)\right)\Mid X_t=x\right]
\]
with
\begin{equation}\label{eq:change_function}
\begin{aligned}
p_1(x)&=p^2(x),& u_1(s)&=\left(\frac12(1-\rho^2)g_1^2(T-s)-\frac12g_1(T-s)+g_2(T-s)\right)g_0^2(s),\\
p_2(x)&=-\frac{a+bx}{c}p(x)-\frac{c}{2}p'(x),& u_2(s)&=\rho g_1(T-s)g_0(s),\\
p_3(x)&=-\frac1c\int_0^x p(y)\,dy,& u_3(s)&=-\rho g_1'(T-s)g_0(s)+\rho g_1(T-s)g_0'(s),\\
q(x)&=\frac1c\int_0^x p(y)\,dy,& v(s)&=\rho g_1(T-s)g_0(s).
\end{aligned}
\end{equation}
In addition, $u_i,v$ are continuous on $[0,T]$, $\Re(u_1)\leq0$, $\Re(u_2),\Re(u_3),\Re(v)=0$ and $p_i,q$ are all negligible to double factorial.
\end{lemma}
\begin{proof}
The proof is given in Appendix~\ref{sectionA}.
\end{proof}

We now introduce $F_n$, the discretised version of $F$.
\begin{definition}\label{defFn}
For $\mu_n=\frac{T}{n}\sum_{i=0}^{n-1}\delta_{iT/n}$, where $\delta_t$ is the Dirac mass at a point $t$, we define the function
\[
F_n(t,x):=\mathbb{E}\left[\exp\left(\int_t^T\sum_{i=1}^3u_i(s)p_i(X_s)\,\mu_n(ds)+v(T)q(X_T)-v(t)q(X_t)\right)\Mid X_t=x\right].
\]
\end{definition}

\begin{definition}\label{extension}
Given an entire function $g$ on $\mathbb{R}$, we define the extension of $g$ to the complex plane by $g_c(z)=\sum_{k=0}^{\infty}\frac{g^{(k)}(0)}{k!}z^k$, $z\in\mathbb{C}$.
\end{definition}

\begin{remark}
Given an entire function $g$ on $\mathbb{R}$, the power series $\sum_{k=0}^{\infty}\frac{g^{(k)}(0)}{k!}x^k$, $x\in\mathbb{R}$, has an infinite radius of convergence. So does $\sum_{k=0}^{\infty}\frac{g^{(k)}(0)}{k!}z^k$, $z\in\mathbb{C}$, and thus $g_c$ is well defined on $\mathbb{C}$ and is entire on $\mathbb{C}$.
\end{remark}

\begin{theorem}\label{T:Discrete}
Fix a power series $p$ negligible to double factorial, continuously differentiable functions $g_0:[0,T]\to\mathbb{R}$, $g_1,g_2:[0,T]\to\mathbb{C}$ with $\Re(g_1)=0$, $\Re(g_2)\leq0$. Then $F_n(t,\cdot)$ as specified in Definition~\ref{defFn} is well defined and entire in $x\in\mathbb{R}$. If in addition the extension $(F_n(t,\cdot))_c$ of $F_n(t,\cdot)$ to the complex plane does not vanish for all $n\in\mathbb{N}$ and $t\leq T$, then $\log F_n$ is entire in $x\in\mathbb{R}$ and
\[
F_n(t,x)=\exp\left(\sum_{k\geq0}\psi_{n,k}(T-t)x^k\right),\qquad t\leq T,
\]
and
\[
F(t,x)=\lim_{n\to\infty}F_n(t,x),
\]
where for all $n\in\mathbb{N}$, the family of functions $\psi_{n,k}$ is defined by
\begin{align*}
\psi_{n,k}(t)&=\phi_{n,k}(t)-\rho g_1(t)g_0(T-t)\frac{p_{k-1}}{ck},\qquad k\geq1,\\
\psi_{n,0}(t)&=\phi_{n,0}(t),
\end{align*}
where the family $\phi_{n,k}$ solves a (step-wise) `discretised' system of Riccati ODEs
\begin{align*}
\phi_{n,k}'(t)&=bk\phi_{n,k}(t)+a(k+1)\phi_{n,k+1}(t)+\frac{c^2(k+2)(k+1)}{2}\phi_{n,k+2}(t)\\
&\quad+\frac{c^2}{2}(\widetilde\phi_n(t)*\widetilde\phi_n(t))_k,\qquad\widetilde\phi_{n,k}=(k+1)\phi_{n,k+1},\quad\frac{nt}{T}\notin\{0,1,2,\ldots,n\},\\
\phi_{n,k}(t)&=\lim_{s\uparrow t}\phi_{n,k}(s)+\frac{T}{n}\sum_{i=1}^3u_i(T-t)(p_i)_k,\qquad\frac{nt}{T}\in\{1,2,\ldots,n\},\\
\phi_{n,k}(0)&=v(T)q_k.
\end{align*}
\end{theorem}
\begin{proof}
The proof is given in Section~\ref{section7}.
\end{proof}

\begin{remark}
One of the features of Theorem~\ref{T:Discrete} is that the functions $F_n(t,\cdot)$ are entire in $x\in\mathbb{R}$ with no additional assumptions. However, due to the delicate nature of the complex logarithm function, $\log F_n(t,\cdot)$ inherits this property if and only if the extension $(F_n(t,\cdot))_c$ of $F_n(t,\cdot)$ to the complex plane does not vanish; see Lemma~\ref{hololift}. The requirement of $F_n$ being non-vanishing on the real line is not enough; see Remark~\ref{remarkB.4}. For this reason, we could not obtain a corollary of Theorem~\ref{T:Discrete} without the non-vanishing condition to deal with Laplace transforms as we did in Corollary~\ref{Cor:Riccati}.
\end{remark}

\section{Numerical illustration}\label{S:numerical}

\subsection{Numerical scheme for high-dimensional Riccati equations}\label{numerical_schemes}

In this section, we show how to solve the infinite-dimensional system \eqref{eq:Ric} and \eqref{eq:Ric_init} of Riccati equations numerically for derivative pricing and hedging. The idea is to solve a truncated version of the ODE by assuming $\psi_k\approx 0$, $k>M$, for some integer $M$.

Even the truncated version of the infinite-dimensional system of Riccati equations can be challenging to solve. Standard techniques such as the Runge--Kutta methods can fail to solve the truncated version, for example when the parameters $b$ and $c$ in \eqref{polynomial_model} are large; see \citet*[Figs.~2 and 3]{cuchiero2023signature} for other attempts to solve a similar problem.

We present a customised algorithm that combines variation of constants and the implicit Euler scheme. For this section alone, all matrix and vector indices start from $0$. With some abuse of notation, we rewrite the Riccati ODE \eqref{eq:Ric} and \eqref{eq:Ric_init} in matrix form as
\begin{equation}\label{riccati_mat_form}
\begin{aligned}
\psi'(t)&=P(t)+A\psi(t)+Q\psi(t)+\frac{c^2}{2}\widetilde\psi(t)*\widetilde\psi(t)+\ell(t)N\psi(t),\\
\psi(0)&=0,
\end{aligned}
\end{equation}
with $\psi(t)=(\psi_0(t),\psi_1(t),\ldots)^\top$ and $\psi'(t)$ its element-wise derivative. Here $P(t)$ is a vector with $k$th element being $(g_2(t)+\frac{g_1(t)}{2}(g_1(t)-1))g_0^2(T-t)(p*p)_k$ for $k\in\mathbb N_0$, $A$ is a diagonal matrix with diagonal $A_{k,k}=bk$, and $Q$ is a square matrix with $Q_{k,k+1}=a(k+1)$, $Q_{k,k+2}=c^2(k+1)(k+2)/2$ and zero elsewhere. The convolution term $\widetilde\psi(t)*\widetilde\psi(t)$ denotes the discrete convolution of the vectors $\widetilde\psi(t)$, with $k$th element given by $(\widetilde\psi(t)*\widetilde\psi(t))_k$. Finally, $N$ is a matrix with $N_{j,k}=kp_{j+1-k}$, where $p_m$ denotes the power series coefficient such that $p_m=0$ whenever $m<0$, with $\ell(t)$ collecting the term $c\rho g_1(t)g_0(T-t)$.

Given a partition $0=t_0<t_1<\cdots<t_{\widehat N}=T$ of the interval $[0,T]$, variation of constants gives
\begin{equation}\label{var_constants}
\begin{aligned}
\psi(t_{i+1})&=e^{A\Delta_i}\psi(t_i)+\int_{t_i}^{t_{i+1}}e^{A(t_{i+1}-s)}P(s)\,ds+\int_{t_i}^{t_{i+1}}e^{A(t_{i+1}-s)}Q\psi(s)\,ds\\
&\quad+\frac{c^2}{2}\int_{t_i}^{t_{i+1}}e^{A(t_{i+1}-s)}\big(\widetilde\psi(s)*\widetilde\psi(s)\big)\,ds\\
&\quad+\int_{t_i}^{t_{i+1}}e^{A(t_{i+1}-s)}\ell(s)N\psi(s)\,ds
\end{aligned}
\end{equation}
with $\Delta_i=t_{i+1}-t_i$. By truncating the vector $\psi$ up to some level $M$ and assuming $(\psi_k)_{k>M}=0$, we approximate the solution of $\psi(t_{i+1})$ in \eqref{var_constants} by the quasi-implicit scheme
\begin{equation}\label{var_constants_approx}
\begin{aligned}
\psi(t_{i+1})&\approx e^{A\Delta_i}\psi(t_i)+G_iP(t_i)+G_iQ\psi(t_{i+1})+\frac{c^2}{2}G_i\big(\widetilde\psi(t_i)*\widetilde\psi(t_{i+1})\big)+G_i\ell(t_i)N\psi(t_{i+1}),
\end{aligned}
\end{equation}
where $G_i=\int_{t_i}^{t_{i+1}}e^{A(t_{i+1}-s)}\,ds$ is a diagonal matrix whose $k$th diagonal entry is $(e^{b(k-1)\Delta_i}-1)/(b(k-1))$, $1<k\leq M+1$, and first diagonal entry $\Delta_i$. Similarly, $Q$ and $N$ are now $(M+1)\times(M+1)$ matrices as defined above, with $\psi$ and $P(t_i)$ both vectors of dimension $M+1$.

For the quadratic term $\widetilde\psi(t_i)*\widetilde\psi(t_{i+1})$, we define
\[
R^{(i)}\psi(t_{i+1}):=\widetilde\psi(t_i)*\widetilde\psi(t_{i+1}),
\]
where $R^{(i)}$ is an $(M+1)\times(M+1)$ matrix with $R^{(i)}_{j,k}=(j+2-k)k\psi_{j+2-k}(t_i)$, where we complement the definitions by setting $(\psi_z(t_i))_{z<0}=0$. After rearranging terms, we can now rewrite \eqref{var_constants_approx} as
\[
J_i\psi(t_{i+1})\approx\big(e^{A\Delta_i}\psi(t_i)+G_iP(t_i)\big),
\qquad J_i=\left(I-G_i\left(\frac{c^2}{2}R^{(i)}+\ell(t_i)N+Q\right)\right).
\]
We are therefore interested in a system of complex-valued linear equations of the form $J_ix=y$. In practice, this can be efficiently solved using the lower-upper (LU) decomposition, avoiding the need to invert the matrix $J_i$ at each time step which can be computationally expensive. Finally, the matrix $Q$ can contain large coefficients resulting from the term $c^2(k+2)(k+1)/2$ when $k$ is large. This can cause numerical instability; so we capped $(k+2)(k+1)$ to some fixed level $k_{\max}^2=15^2$, which is vital in ensuring that the scheme does not blow up. A detailed implementation of our algorithm can be found at \url{https://colab.research.google.com/drive/1VCVyN1qQmLgOWjOy4fbWftyDqEQdm5n5?usp=sharing} (a Python notebook).

We first test our algorithm to compute $\E[\exp(-W_t^4/4!)]$ as considered in \citet*[Figs.~2 and 3]{cuchiero2023signature}, which can be expressed by the solution of a particular case of the Riccati ODE in \eqref{eq:Ric} and \eqref{eq:Ric_init}, with a slight modification to the initial value, given by
\begin{align*}
\psi'_k(t)&=\frac{(k+2)(k+1)}{2}\psi_{k+2}(t)+\frac12\big(\widetilde\psi(t)*\widetilde\psi(t)\big)_k,\\
\psi_k(0)&=-\frac1{4!}\mathbf{1}_{\{k=4\}},
\end{align*}
where $\mathbf{1}$ is the indicator function. Alternatively, the expectation $\E[\exp(-W_t^4/4!)]$ can also be computed via a numerical integration quadrature with respect to the Gaussian density, allowing us to evaluate our algorithm's performance. Figure~\ref{fig:brownian_example} illustrates numerical convergence in terms of $M$, with step size $n=100$. We observe the instability of the Runge--Kutta scheme with increasing truncation $M$, whereas the scheme we propose remains stable with an increase of $M$.

\begin{figure}[H]
\centering
\includegraphics[width=0.9\textwidth]{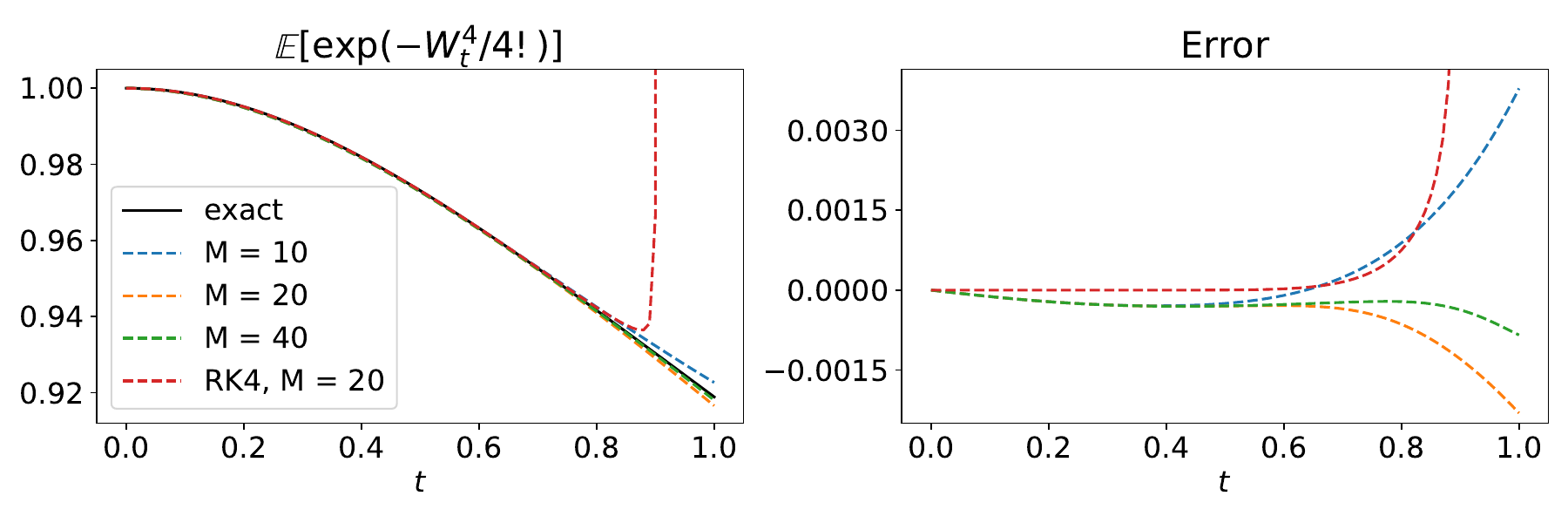}
\caption{Numerical solution of $\E[\exp(-W_t^4/4!)]$ with different values of $M$ using our algorithm vs.~the exact solution. For comparison, the explicit Runge--Kutta 4 solution with $M=20$ is added in red.}
\label{fig:brownian_example}
\end{figure}

Our scheme, in a nutshell, is essentially an implicit scheme of order one. To further improve its efficiency, we can choose at time step $t_i$ an adaptive step size $\Delta_i$ such that the contribution from the second order term $\psi''(t_i)$ does not exceed a certain threshold $\zeta>0$. From \eqref{riccati_mat_form}, assuming that $g_1$, $g_2$ and $g_0$ are continuously differentiable in $t$, we have
\[
\psi(t_{i+1})-\psi(t_i)=\psi'(t_i)\Delta_i+\frac12\psi''(t_i)\Delta_i^2+o(\Delta_i^2),
\]
where $\psi''(t)$ can be computed explicitly for all $t\leq T$. We can then choose
\[
\Delta_i=\min\left(\max\left(\sqrt{2\zeta/|\psi''(t_i)|},\Delta_{\min}\right),\Delta_{\max}\right),
\]
with $\Delta_{\min}$ ($\Delta_{\max}$) some pre-determined minimum (maximum) step size. Intuitively, the smaller $|\psi''(t_i)|$, the larger a time step $\Delta_i$ one can take without sacrificing too much accuracy, provided that $\Delta_{\min}\leq\Delta_i\leq\Delta_{\max}$. For all numerical experiments below, we chose $\zeta=1\times10^{-9}$, $\Delta_{\min}=5\times10^{-5}$ and $\Delta_{\max}=\max(1/252\times10^{-1},T\times10^{-1})$.

\subsection{Pricing SPX derivatives via Fourier}\label{subsec:spx_fourier}

Our numerical scheme gives us direct access to the characteristic function of $\log S$ to price derivatives via Fourier inversion techniques, using the expression \eqref{eq:charfun} with $g_2(t)=0$ and $g_1(t)=iu$, for $u\in\mathbb R$. Specifically, we show how one can price European options via Fourier techniques for the quintic Ornstein--Uhlenbeck model (\citet*{jaber2022quintic}) and the one-factor Bergomi model (\citet{bergomi2005smile,dupire1992arbitrage}). This also serves as a numerical validation of the representation \eqref{eq:charfun}.

There are several Fourier inversion techniques available in the literature to price European-style vanilla call and put options. Here we adapt the pricing formula suggested in \citet{lewis2001simple}, which involves only one integral to evaluate
\begin{equation}\label{call_price_lewis}
\begin{aligned}
C_t(S_t,K,T)&:=\E[(S_T-K)^+\mid\mathcal F_t]=S_t-\frac K\pi\int_0^\infty\Re\left(e^{(iu+\frac12)k_t}\varphi\left(u-\frac i2\right)\right)\frac{du}{u^2+\frac14},
\end{aligned}
\end{equation}
where $C_t(S_t,K,T)$ denotes the European call option price with strike $K$ and time to maturity $T-t$, $k_t:=\log(K/S_t)$ is the log-moneyness and $\varphi(u)$ is the Fourier--Laplace transform of $\log(S_T/S_t)$, conditionally on $X_t=x$. More precisely,
\[
\varphi(u):=\E[e^{iu\log(S_T/S_t)}\mid X_t=x].
\]
Equivalently, $\varphi(u)$ can be obtained from the expression of $F(t,x)$ in \eqref{eq:defF}: for fixed $(t,x)$, $\varphi(u)$ is obtained by setting $g_1\equiv iu$, $g_2\equiv0$ to obtain
\[
\varphi(u):=\left.F(t,x)\right|_{g_1\equiv iu,\,g_2\equiv0}=\E[e^{iu\log(S_T/S_t)}\mid X_t=x].
\]
To compute the improper integral, we make use of the equivalent representation of $F(t,x)$ in \eqref{eq:charfun} via the Gauss--Laguerre quadrature which has been demonstrated to be efficient; see \citet{abi2024signature,hurd2010fourier,samet2022optimal}.

To speed up the computation of $C_t$, we add a control variate to reduce the number of evaluations of $\varphi(u)$ for different $u$ via
\begin{equation}\label{call_price_control}
\begin{aligned}
C_t(S_t,K,T)&=\widehat C_t(S_t,K,T)-\frac K\pi\int_0^\infty\Re\left(e^{(iu+\frac12)k_t}\left(\varphi\left(u-\frac i2\right)-\widehat\varphi\left(u-\frac i2\right)\right)\right)\frac{du}{u^2+\frac14},
\end{aligned}
\end{equation}
where $\widehat C_t(S_t,K,T)$ is the appropriate call price of the control variate and $\widehat\varphi(u)$ is the Fourier--Laplace transform of $\log(S_T/S_t)$ of the control variate.

In this section, we choose as the control variate the Heston model, which admits a closed-form characteristic function that is also affine in its state variables; see \citet{heston1993closed}. The Heston model parameters can be efficiently selected via a standard optimisation algorithm such that the difference between $\varphi(u-i/2)$ and $\widehat\varphi(u-i/2)$ is minimised. Of course, other control variates with explicit characteristic functions, such as the Black--Scholes model, can also be used.

\subsubsection{Quintic Ornstein--Uhlenbeck volatility model}\label{subsubsec:quintic_ou}

The volatility process $\sigma_t$ under the quintic OU model takes the form
\begin{equation}\label{quintic_model}
\begin{aligned}
\sigma_t&=\sqrt{\frac{\xi_0(t)}{\E[p(X_t)^2]}}p(X_t),\qquad p(x)=p_0+p_1x+p_3x^3+p_5x^5,\\
X_t&=\varepsilon^\alpha\int_0^t e^{\alpha\varepsilon^{-1}(t-s)}\,dW_s,
\end{aligned}
\end{equation}
with $\varepsilon>0$ and $\alpha\leq0$. The nonnegative coefficients $p_0,p_1,p_3,p_5\geq0$ ($p_2=p_4=0$) ensure a negative leverage effect as well as the martingality of $S$ whenever $\rho\leq0$; see \citet*{jaber2022quintic}. The particular parametrisation with $X$ means $a=0$, $b=\alpha\varepsilon^{-1}$ and $c=\varepsilon^\alpha$ in \eqref{polynomial_model}. This model has been shown to produce a remarkable joint fit to both SPX and VIX implied volatility surfaces (\citet*{jaber2022quintic,abi2022joint}).

Since $X$ is an OU process that can be simulated exactly, one could be tempted to use Monte Carlo to estimate the SPX derivatives with appropriate control variates. However, the calibrated values of $\varepsilon$ are usually very small, and $\alpha$ is negative, pushing the model effectively into a fast regime with large mean reversion of order $\alpha\varepsilon^{-1}$ and large vol of vol $\varepsilon^\alpha$. It is known that the standard Euler scheme for pricing SPX derivatives can produce a large estimation bias for longer maturities $T$ due to the highly erratic paths of $\sigma$, requiring a finer step size or other asymptotic approximation techniques as in \citet*{fouque2003multiscale,fouque2016second}. Pricing via Fourier methods hence presents an attractive alternative given the increased efficiency and accuracy. To demonstrate, we choose the parameters $\rho=-0.65$, $\alpha=-0.6$, $(p_0,p_1,p_3,p_5)=(0.01,1,0.214,0.227)$, $\xi_0(t)=0.025$, $\varepsilon=1/52$. Figure~\ref{quintic_ou_riccati} shows the convergence of SPX implied volatility for different maturities as the truncation level $M$ increases. We include in Appendix~\ref{cost_comparison} a cost comparison between Fourier vs.~Monte Carlo methods to demonstrate the efficiency of the Fourier methods using our proposed numerical scheme.

\begin{figure}[!htbp]
\centering
\includegraphics[width=0.8\textwidth]{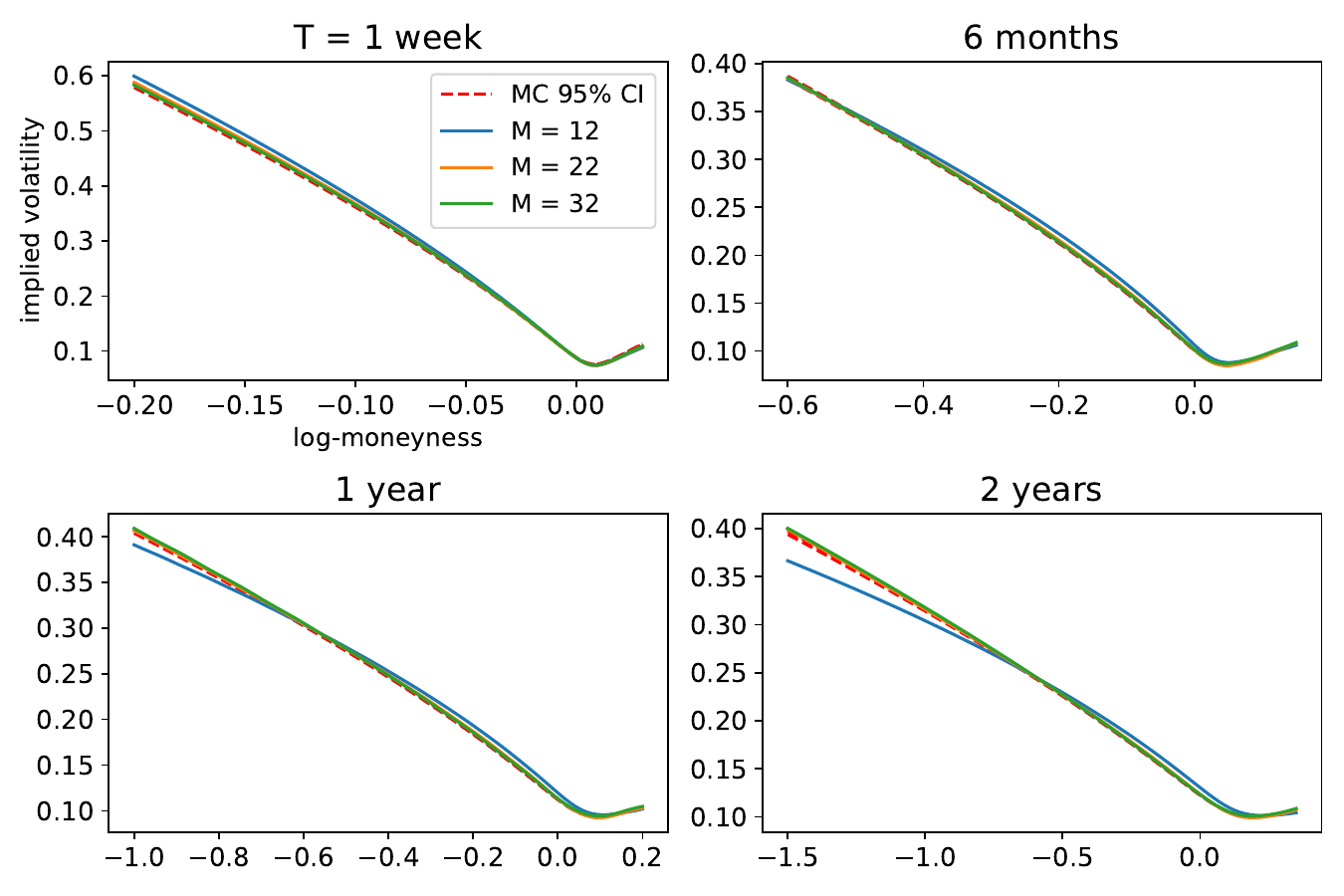}
\caption{SPX implied volatility for different maturities in the quintic OU model, computed with our algorithm with different levels $M$. Dotted red lines are Monte Carlo 95\% confidence interval computed with 500,000 simulations and $n=10,000$ number of steps per maturity slice.}
\label{quintic_ou_riccati}
\end{figure}

\subsubsection{One-factor Bergomi model}\label{subsubsec:bergomi}

The one-factor Bergomi model (\citet{bergomi2005smile,dupire1992arbitrage}) assumes $\sigma_t$ to be log-normal and of the form
\begin{equation}\label{bergomi_one-factor}
\begin{aligned}
\sigma_t&=\sqrt{\xi_0(t)}\exp\left(\frac12\eta X_t-\frac14\eta^2\E[X_t^2]\right),\\
X_t&=\varepsilon^\alpha\int_0^t e^{\alpha\varepsilon^{-1}(t-s)}\,dW_s,
\end{aligned}
\end{equation}
with $\varepsilon>0$ and $\alpha\leq0$. Similarly to the quintic OU model before, we have $a=0$, $b=\alpha\varepsilon^{-1}$ and $c=\varepsilon^\alpha$. We now approximate the exponential as a truncated sum up to level $L$, i.e.,
\[
\widetilde\sigma_t=\sqrt{\frac{\xi_0(t)}{\E[p(X_t)^2]}}p(X_t),\qquad p(X_t)=\sum_{k=0}^L p_kX_t^k,\qquad p_k=\frac{\eta^k}{2^kk!},
\]
where $\widetilde\sigma_t$ converges to $\sigma_t$ in \eqref{bergomi_one-factor} pointwise a.s.~when sending $L\to\infty$. We fix $\rho=-0.7$, $\varepsilon=1/52$, $\alpha=-0.7$, $\eta=1.2$, $\xi_0(t)=0.025$ for the numerical experiment and set $L=8$. Figure~\ref{bergomi_ou_riccati} shows that our numerical scheme quickly converges to Monte Carlo estimates of the original one-factor Bergomi model.

\begin{figure}[!htbp]
\centering
\includegraphics[width=0.8\textwidth]{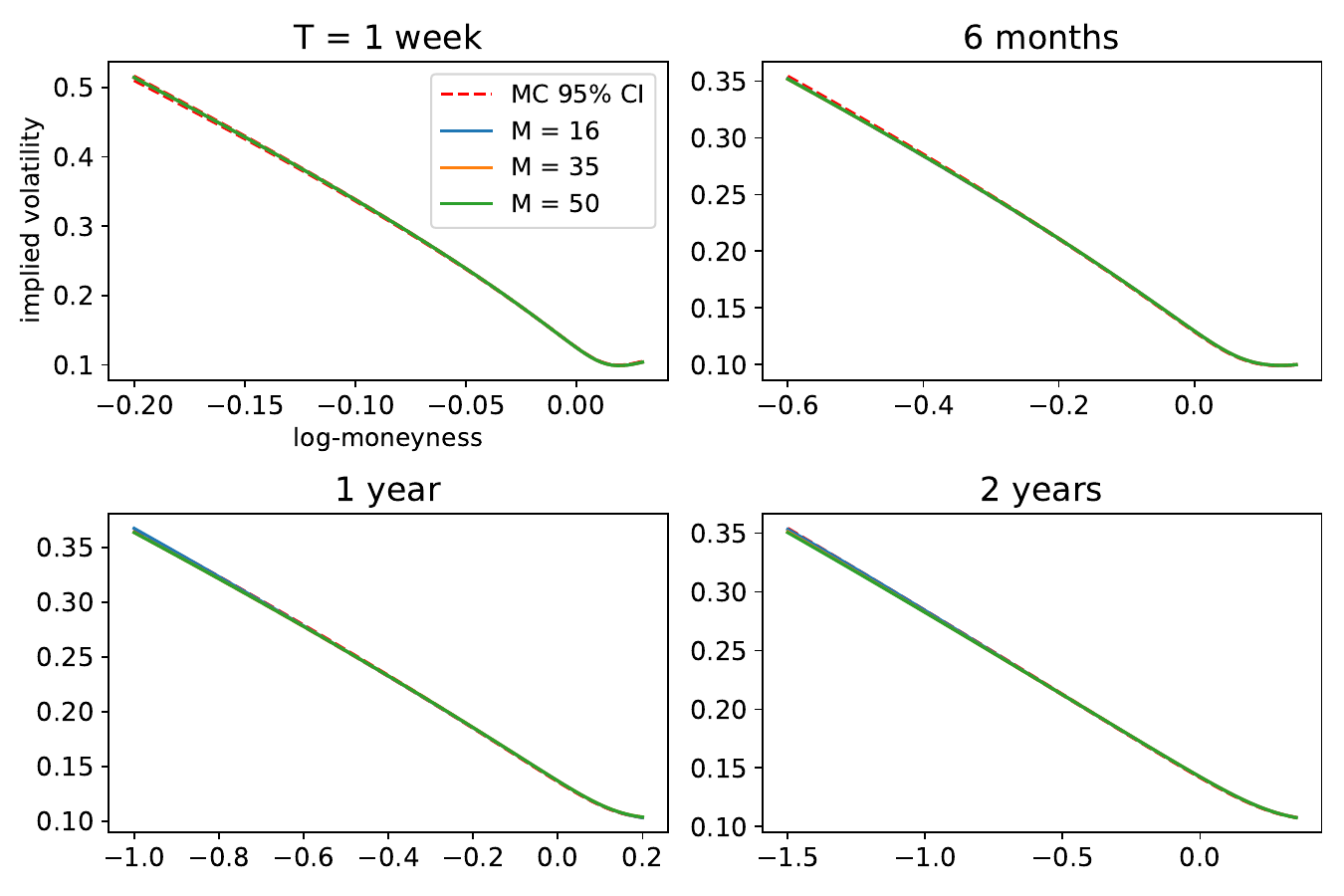}
\caption{SPX implied volatility for different maturities in the one-factor Bergomi model, computed with our algorithm with different levels $M$. Dotted red lines are Monte Carlo 95\% confidence interval computed with 500,000 simulations and $n=10,000$ number of steps per maturity slice.}
\label{bergomi_ou_riccati}
\end{figure}

\subsection{Pricing \texorpdfstring{$q$}{q}-volatility swaps via Laplace inversion}

The $q$-volatility swap rate $K_q(t,T)$ is defined by
\[
K_q(t,T):=\E\left[\left(\frac1{T-t}\int_t^T\sigma_s^2\,ds\right)^q\Bigm|X_t=x\right],\qquad q\in[0,1].
\]
For the case of a standard volatility swap (i.e., $q=1/2$), one can price the swap rate $K_q(t,T)$ via the inverse Laplace transform in \citet{schurger2002laplace} as
\[
K_{\frac12}(t,T):=\E\left[\left(\frac1{T-t}\int_t^T\sigma_s^2\,ds\right)^{1/2}\Bigm|X_t=x\right]
=\frac1{2\sqrt\pi}\int_0^\infty\frac{1-\widetilde F(u)}{u^{3/2}}\,du,
\]
where $\widetilde F(u)$ is the Laplace transform of $\frac1{T-t}\int_t^T\sigma_s^2\,ds$ conditionally on $X_t=x$. More precisely,
\[
\widetilde F(u):=\E\left[e^{-\frac u{T-t}\int_t^T\sigma_s^2\,ds}\mid X_t=x\right].
\]
Equivalently, $\widetilde F(u)$ can be obtained from the expression of $F(t,x)$ in \eqref{eq:defF}: for fixed $(t,x)$, $\widetilde F(u)$ is obtained by setting $g_1\equiv0$, $g_2\equiv-u/(T-t)$ so that
\[
\widetilde F(u):=\left.F(t,x)\right|_{g_1\equiv0,\,g_2\equiv-\frac u{T-t}}=\E\left[e^{-\frac u{T-t}\int_t^T\sigma_s^2\,ds}\mid X_t=x\right].
\]
Again, we can accelerate the computation via a control variate, for example the Black--Scholes control variate
\[
K_{\frac12}(t,T)=\sigma_{\mathrm{BS}}+\frac1{2\sqrt\pi}\int_0^\infty\frac{\widetilde F_{\mathrm{BS}}(u)-\widetilde F(u)}{u^{3/2}}\,du
\]
with $\widetilde F_{\mathrm{BS}}(u)=\exp(-u\sigma_{\mathrm{BS}}^2)$, where $\sigma_{\mathrm{BS}}$ is an arbitrary level of volatility that can be fixed upfront.

Compared to Sect.~\ref{subsec:spx_fourier}, we are even more confident of our numerical scheme thanks to Corollary~\ref{Cor:Riccati}. For demonstration purposes, we use the same model parameters for the quintic Ornstein--Uhlenbeck and the one-factor Bergomi model as in Sect.~\ref{subsec:spx_fourier} and adopt a parametric forward variance curve in the form $\xi_0(t)=V_0e^{-kt}+V_\infty(1-e^{-kt})$ with $V_0=0.025$, $k=5$ and $V_\infty=0.06$. Figure~\ref{fig:vol_swaps} shows the volatility swaps of the two models computed by inverse Laplace transform with truncation level $M=32$ vs.~that computed by Monte Carlo with 400,000 simulations and 10,000 steps.

\begin{figure}[H]
\centering
\includegraphics[width=0.49\textwidth]{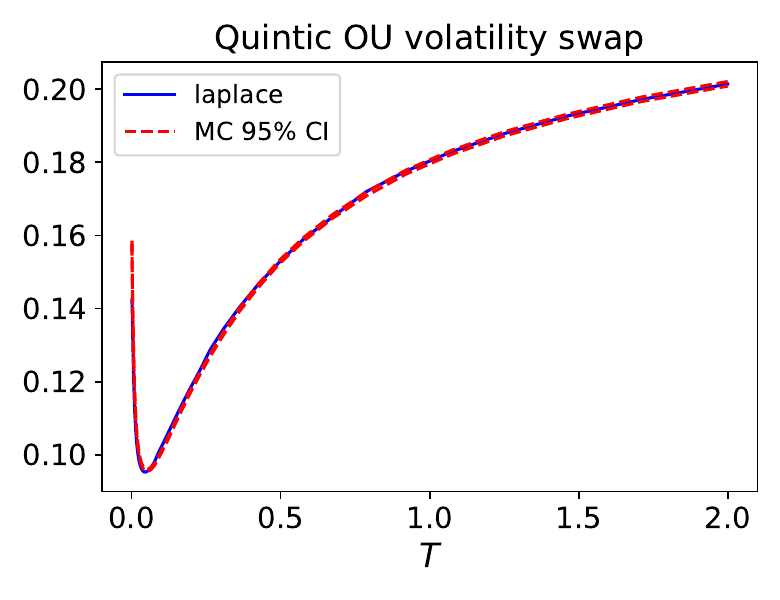}%
\includegraphics[width=0.49\textwidth]{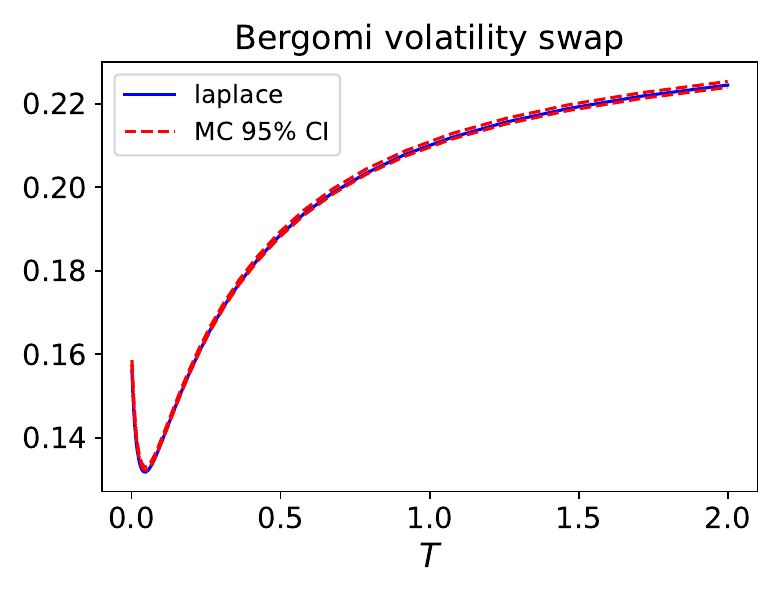}
\caption{Volatility swaps under the quintic OU model (left) and the one-factor Bergomi model (right) for maturities up to 2 years, using the same parameter values as in Sect.~\ref{subsec:spx_fourier}.}
\label{fig:vol_swaps}
\end{figure}

\subsection{Model calibration to market data via Fourier}

The family of polynomial OU volatility models allows fast pricing of VIX futures and VIX options via numerical integration of the payoff with respect to the standard Gaussian density; see \citet*{jaber2022quintic}. Together with fast Fourier pricing, this opens the door to joint calibration to SPX and VIX smiles. This section also serves the purpose of highlighting the stability of our numerical discretisation scheme combined with Fourier inversion in a calibration procedure where a large number of evaluations of the characteristic function are needed for a wide range of realistic model parameters and Fourier variables.

In a nutshell, the calibration of a model involves minimising the mean square error between prices coming from the model vs.~those from market data. Without going too much into the details, we first demonstrate the capability of the quintic OU model \eqref{quintic_model} to jointly calibrate two slices of maturities of SPX smiles together with one slice of maturity of the VIX smile, as shown in Fig.~\ref{quintic_spx_calib}, with the calibrated parameters $\rho=-0.6763$, $\alpha=-0.6821$, $(p_0,p_1,p_3,p_5)=(0.0202,1.3332,0.0578,0.0071)$ and fixed $\varepsilon=1/52$, with $\xi_0(t)$ coming directly from market data.

\begin{figure}[H]
\centering
\includegraphics[width=0.65\textwidth]{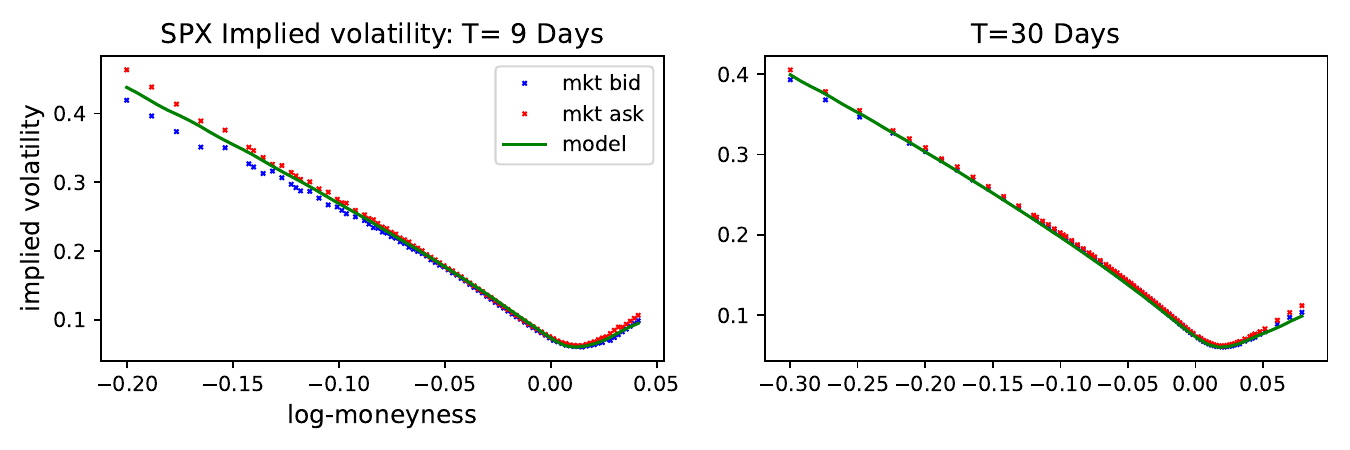}
\includegraphics[width=0.65\textwidth]{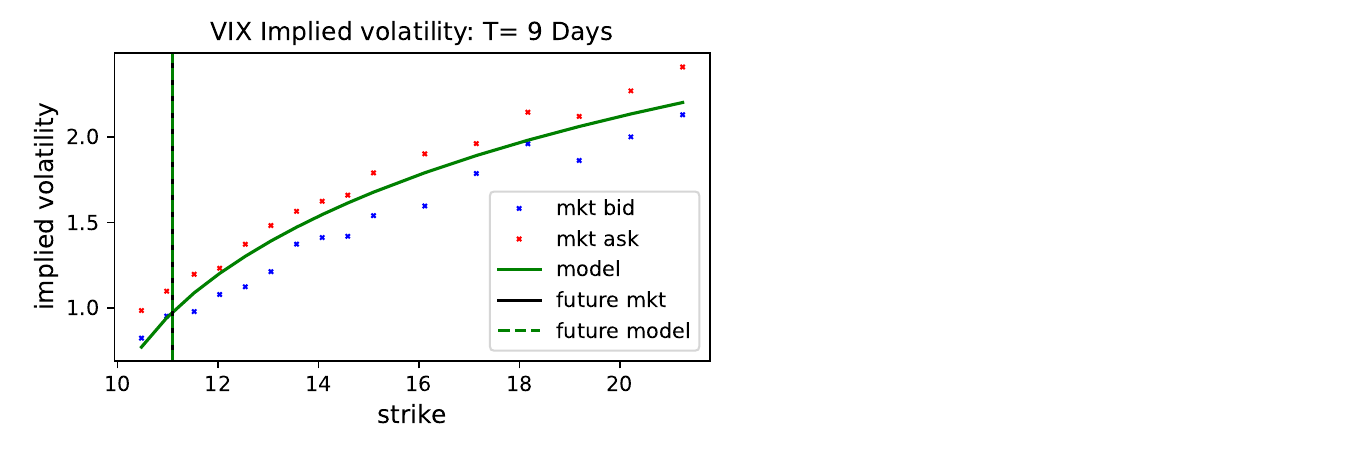}
\caption{Quintic OU model (green lines) jointly calibrated to the SPX and VIX smiles (bid/ask in blue/red) on 23 October 2017 via Fourier using the Nelder--Mead optimisation algorithm. The truncation level of the Riccati solver is set at $M=32$.}
\label{quintic_spx_calib}
\end{figure}

Next we showcase the calibration results of the one-factor Bergomi model \eqref{bergomi_one-factor} on four slices of maturities of SPX smiles between around 1 week to 3 months, with calibrated parameters $\eta=1.1416$, $\rho=-0.6744$, $\alpha=-0.7377$ and fixed $\varepsilon=1/52$ as shown in Fig.~\ref{bergomi_calib}.

\begin{figure}[H]
\centering
\includegraphics[width=0.8\textwidth]{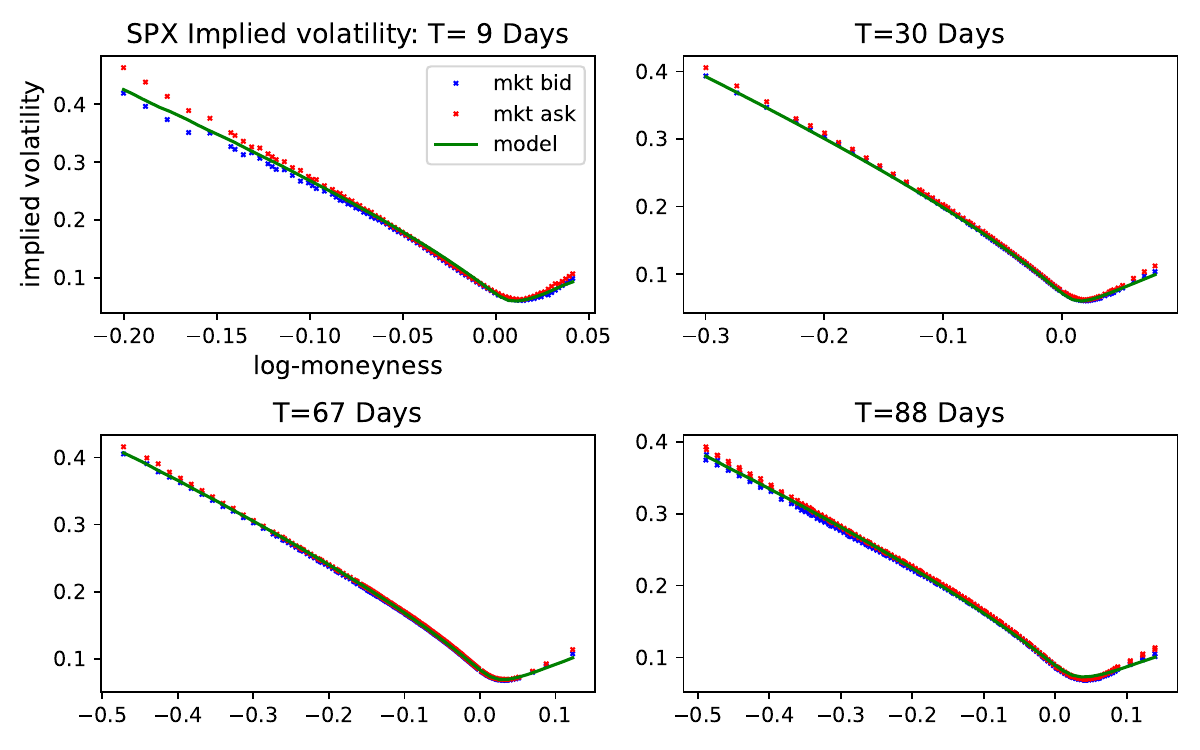}
\caption{One-factor Bergomi model (green lines) calibrated to the SPX smiles (bid/ask in blue/red) on 23 October 2017 via Fourier using the Nelder--Mead optimisation algorithm. The truncation level of the Riccati solver is set at $M=32$.}
\label{bergomi_calib}
\end{figure}

The market data of SPX and VIX volatility surface is purchased from the CBOE website \url{https://datashop.cboe.com/}. For more calibration examples under the quintic OU and the one-factor Bergomi models, we refer to Figs.~\ref{quintic_spx_calib_2} and \ref{bergomi_calib_2} in Appendix~\ref{C:_more_calib_results}.

\section{Proof of Theorem \ref{thm1.1}}
\label{section5}
We first introduce a lemma to justify the application of Itô's formula to the series of functions $\sum_{k\geq 0}\psi_k(T-t)X_t^k$.

\begin{lemma}\label{lemma5.1}
Under the conditions of Theorem \ref{thm1.1}, $h(t,x):=\sum_{k\geq 0}\psi_k(T-t)x^k$ is $C^2$ in $x$ and $C^1$ in $t$, and its partial derivatives are $h_{x}(t,x)=\sum_{k\geq 0}(k+1)\psi_{k+1}(T-t)x^{k}$, $h_{xx}(t,x)=\sum_{k\geq 0}(k+2)(k+1)\psi_{k+2}(T-t)x^{k}$, $h_t(t,x)=\sum_{k\geq 0}-\psi'_k(T-t)x^k$.
\end{lemma}
\begin{proof}
First, we clearly have
\[
|\psi_k(T-t)x^{k}|\leq \sup_{t\in[0,T]}|\psi_k(t)| |x|^k,
\]
and by assumption, $\sum_{k\geq 0}\sup_{t\in[0,T]}|\psi_k(t)| |x|^k<\infty$. So $h(t,x)=\sum_{k\geq 0}\psi_k(T-t)x^{k}$ is well defined. If we restrict $x$ to a bounded set of $\R$, then $h^n(t,x):=\sum_{k=0}^{n}\psi_k(T-t)x^{k}$ is continuous and converges uniformly on bounded subsets of $[0,T]\times\R$ to $h(t,x)$ when $n\rightarrow \infty$, so $h(t,x)$ is continuous. 

Notice that the domain of convergence for $\sum_{k\geq 0}\sup_{t\in[0,T]}|\psi_k(t)| x^k<\infty$ is $\R$, and so the same holds for $\sum_{k\geq 0}\psi_k(T-t)x^{k}$. So $h_x(t,x),h_{xx}(t,x)$ are well-defined, continuous, and have the expression in the statement for the same reason as for $h(t,x)$.

To treat $h_t(t,x)$, we should first prove that $\sum_{k\geq 0} \sup_{t\in[0,T]}|\psi'_k(t)|  x^k$ has also an infinite radius of convergence.
  
Indeed, using the Riccati expression \eqref{eq:Ric} and \eqref{eq:Ric_init}, we only have to check that we have an infinite radius of convergence for 
\begin{equation}\label{eqatk} \sum_{k\geq0}\sup_{t\in[0,T]}|a_{t,k}|x^k,  
\end{equation}

where $a_{t,k}$ is chosen among $(p*p)_k$, $k\psi_k(t)$, $(k+1)\psi_{k+1}(t)$, $\frac{(k+2)(k+1)}{2}\psi_{k+2}(t)$, $(\widetilde\psi(t)*\widetilde\psi(t))_k$ and $(p*\widetilde\psi(t))_k$. 

By the assumption of Lemma~\ref{lemma5.1} and the Cauchy--Hadamard theorem, we know that
\[
\limsup_{k\rightarrow \infty}(\sup_{t\in[0,T]}|\psi_k(t)|)^{\frac{1}{k}}=0.
\]
We claim that when $a_{t,k}=k\psi_k(t),(k+1)\psi_{k+1}(t)=\widetilde \psi (t)_k,\frac{(k+1)(k+2)}{2}\psi_{k+2}(t)$, \eqref{eqatk} has an infinite radius of convergence. We only give the proof for the case when $a_{t,k}=(k+1)\psi_{k+1}(t)=\widetilde{\psi}(t)_k$, the other two cases can be argued in the same way.  As $\lim_{k\rightarrow \infty}\frac{\ln (k+1)}{k}=0$, we have that $\lim_{k\rightarrow \infty}(k+1)^{\frac{1}{k}}=1$. Besides, as $\limsup_{k\rightarrow \infty}(\sup_{t\in[0,T]}|\psi_k(t)|)^{\frac{1}{k}}=0$, we have for all $k$ big enough, $(\sup_{t\in[0,T]}|\psi_k(t)|)^{\frac{1}{k}}\leq 1$, which leads to $\sup_{t\in[0,T]}|\psi_k(t)|\leq 1$. Thus $$ \limsup_{k\rightarrow \infty}(\sup_{t\in [0,T]}|\psi_{k+1}(t)|)^{\frac{1}{k}}\leq \limsup_{k\rightarrow \infty}(\sup_{t\in[0,T]}|\psi_{k+1}(t)|)^{\frac{1}{k+1}}=0.$$ Thus by the Cauchy--Hadamard theorem,
\[
\sum_{k=0}^\infty \sup_{t\in [0,T]}|\widetilde{\psi}(t)|_k x^k=\sum_{k=0}^\infty (k+1)\sup_{t\in [0,T]}|\psi_{k+1}(t)|x^k
\]
has an infinite radius of convergence.

For two power series $\sum_{n=0}^\infty a_nx^n$ and $\sum_{n=0}^\infty b_nx^n$ both with an infinite radius of convergence, their product $\sum_{n=0}^\infty (a*b)_nx^n$  also has an infinite radius of convergence. Notice also that $(p*p)_k, (\widetilde \psi (t)*\widetilde \psi (t))_k, (p*\widetilde \psi(t))_k$ are dominated, respectively, by $(|p|*|p|)_k,(\sup_{t\in[0,T]}|\widetilde \psi (t)|*\sup_{t\in[0,T]}|\widetilde \psi (t)|)_k$, and $(|p|*\sup_{t\in[0,T]}|\widetilde \psi (t)|)_k$. Thus when $a_{t,k}$ is among $(p*p)_k, (\widetilde \psi (t)*\widetilde \psi (t))_k, (p*\widetilde \psi(t))_k$, \eqref{eqatk} has also an infinite radius of convergence. Therefore $\sum_{k\geq 0} \sup_{t\in[0,T]}|\psi'_k(t)|  x^k$ has an infinite radius of convergence, so that $\sum_{k\geq 0}-\psi'_k(T-t)x^k$ also has an infinite radius of convergence and converges uniformly as above. So $\sum_{k\geq 0}-\psi'_k(T-t)x^k$ is well-defined and also continuous.

Next, notice that $\frac{\partial}{\partial t} \sum_{k=0}^{n} \psi_k(T-t) x^k=\sum_{k=0}^{n}-\psi'_k(T-t)x^k$ and given the uniform convergence of both series $\sum_{k\geq 0}-\psi'_k(T-t)x^k$ and $\sum_{k\geq 0}\psi_k(T-t)x^k$, we deduce that $h_t(t,x)=\frac{\partial}{\partial t} \sum_{k\geq 0}\psi_k(T-t)x^k=\sum_{k\geq 0}-\psi'_k(T-t)x^k.$

\end{proof}

\begin{proof}[Proof of Theorem \ref{thm1.1}]

We first notice that since the power series  $\sum_{k} |\psi_k(t)|  x^k $  has an infinite radius of convergence,  $U_t$ is well-defined for all $t \leq T$. 

With Lemma \ref{lemma5.1}, we can now apply Itô's formula to the semimartingale $M$ to get

\begin{equation}
  \begin{aligned}
\frac{dM_t}{M_t} &= dU_t + {\frac{1}{2}d\langle U \rangle_t},\\
dU_t &= g_1(T-t)d\log S_t + g_2(T-t)\sigma_t^2dt - \sum_{k\geq 0}\psi_k'(T-t)X_t^k dt \\
& + \sum_{k\geq 0} (k+1){\psi_{k+1}(T-t)}X_t^{k}dX_t+\sum_{k\geq 0}\frac{1}{2} (k+2)(k+1){\psi_{k+2}(T-t)} X_t^{k}d \langle X \rangle_t.\\
  \end{aligned}\label{ito_form1}
  \end{equation}

Plugging \eqref{polynomial_model} into \eqref{ito_form1}, by noticing that $\sum_{k\geq 0}b(k+1)\psi_{k+1}(T-t)X_t^{k+1}=\sum_{k\geq 0}bk\psi_{k}(T-t)X_t^{k}$, we have:
\begin{equation}
\begin{aligned}
    dU_t &= \Bigg(\bigg(g_2(T-t)-\frac{g_1(T-t)}{2}\bigg)\sigma_t^2 - \sum_{k\geq 0}\psi_k'(T-t)X_t^k\\ &+  \sum_{k\geq 0} a (k+1){\psi_{k+1}(T-t)}X_t^{k} + \sum_{k\geq 0} b k \psi_k(T-t)X_t^k \\ &+ \sum_{k\geq 0} c^2 \frac{(k+2)(k+1)}{2}{\psi_{k+2}(T-t)} X_t^{k} \Bigg)dt\\
    &+ \bigg(\sum_{k\geq 0} c (k+1) {\psi_{k+1}(T-t)}X_t^{k} + \rho g_1(T-t)\sigma_t\bigg) dW_t + \sqrt{1-\rho^2} g_1(T-t)\sigma_t dW_t^{\perp},\\
    d\langle U \rangle_t &= \bigg(\sum_{k\geq 0} c (k+1){\psi_{k+1}(T-t)}X_t^{k} + \rho g_1(T-t)\sigma_t\bigg)^2 dt + (1-\rho^2)g_1^2(T-t)\sigma_t^2dt.\\
  \end{aligned}
  \end{equation}

Applying the Cauchy product to the power series $\sum_{k\geq 0} p_k X_t^k$ leads to

\begin{equation}
  \begin{aligned}
\frac{dM_t}{M_t}&= \Big( \sum_{k\geq 0}\Big[ (g_2(T-t)-\frac{g_1(T-t)}{2})g_0^2(t) (p*p)_k - \psi_k'(T-t) + a(k+1)\psi_{k+1}(T-t)\\
    &+ bk\psi_k(T-t) + \frac{c^2(k+2)(k+1)}{2}\psi_{k+2}(T-t)\\
    &+ \frac{c^2}{2}(\widetilde \psi (T-t)*\widetilde \psi (T-t))_k + \frac{g_1^2(T-t)g_0^2(t)}{2}(p*p)_k\\
    &+ \rho g_1(T-t) g_0(t) c (p*\widetilde \psi(T-t))_k\Big] X_t^k \Big)dt\\
    &+ \bigg(\sum_{k\geq 0} c (k+1){\psi_{k+1}(T-t)}X_t^{k} + \rho g_1(T-t)\sigma_t\bigg) dW_t + \sqrt{1-\rho^2} g_1(T-t)\sigma_t dW_t^{\perp}.
  \end{aligned}
  \end{equation}

Now $M$ is a local martingale if and only if it has zero drift (i.e., the $dt$-part is zero a.s.). This is true for all values of $X$ by assumption from \eqref{eq:Ric}, so that $M$ is a local martingale. Moreover, if $M$ is a true martingale, we have
\begin{equation}
  \begin{aligned}
     &\mathbb{E} \left[ \left. \exp\left(\int_t^T g_1(T-s) d\log S_s + \int_t^T g_2(T-s) \sigma_s^2 ds\right) \right| \mathcal{F}_t \right]\\
     &=\exp\left(  \sum_{k\geq 0}\psi_k(T-t)X_t^k \right), \quad t \leq T.
    \end{aligned}
\end{equation}
which follows from the martingality of $M$ and the assumption of the initial condition of $\psi_k$ from \eqref{eq:Ric_init}. This completes the proof. 
\end{proof}

\section{Proof of Theorem \ref{thm1.2}}
\label{section6}
To prove Theorem~\ref{thm1.2}, we adopt the representation of $F$ in Lemma~\ref{2.16} as
$$F(t,x)=\mathbb{E} \left[ \exp\left(\int_t^T \sum_{i=1}^3u_i(s)p_i(X_s)ds+v(T)q(X_T)-v(t)q(X_t)\right) \Mid  X_t=x \right],$$
where the power series $p_1=p^2,p_2,p_3,q$ are negligible to double factorial and the functions $u_1,u_2,u_3,v$ are continuous such that $\Re(u_2)=\Re(u_3)=\Re(v)=0$ and $\Re(u_1)\leq0$. Consider the function $f$ defined by
\begin{align}
f(t,x):&=\mathbb{E} \left[ \exp\left(\int_t^T \sum_{i=1}^3u_i(s)p_i(X_s)ds+v(T)q(X_T)\right)\Mid  X_t=x \right]\\
&=F(t,x)\exp(v(t)q(x)),  \label{eq:defsmallf}
\end{align} 
for all $t\leq T$ and $x\in \R$.
We recall that $f$ is well-defined, since
$$
\Re\left(\int_t^T \sum_{i=1}^3u_i(s)p_i(X_s)ds+v(T)q(X_T)\right)=\int_t^T \Re \left( u_1(s)p^2(X_s)\right)ds\leq 0.
$$
Our proof is composed of five parts:
\begin{itemize}
\item In Section \ref{section6.0}, we start by deriving some  properties of power series negligible to double factorial which will be used later.
    \item 
In Section \ref{section6.1}, we prove the regularity  in $x$ of $f$ given by \eqref{eq:defsmallf}, i.e.~$f$ is $C^\infty$ in $x$ such that the partial derivatives  $\frac{\partial^kf}{\partial x^k}$ are bounded in a suitable sense in preparation for the Feynman--Kac formula in Section \ref{section6.2}.
\item In Section \ref{section6.2}, we prove that $f$ solves the associated PDE coming from the Feynman--Kac formula.

In particular, we  prove that $\frac{\partial f}{\partial t}$ indeed exists.
\item In Section \ref{section6.3}, we introduce a system of ODEs and obtain a solution by comparing the derivatives of all orders of both sides of the Feynman--Kac formula at $x=0$.
\item In Section \ref{section6.4}, we prove that the system of ODEs in Theorem \ref{thm1.2} has a solution, via a system of ODEs introduced in Section \ref{section6.3}.
\end{itemize}

 \subsection{Properties of power series negligible to double factorial} \label{section6.0}
We collect some properties of power series that are negligible to double factorial as defined in Definition \ref{2.3}.

Given a power series $p(x)=\sum_{k=0}^{\infty}p_kx^k$, we define the set $\mathcal A_D(p)$ as the $\R$-algebra generated by all higher-order derivatives $\{p,p',p'',\ldots\}$ of $p$, i.e.,
\begin{equation}\label{eq:defAD}
\mathcal A_D(p):=\left\{\sum_{s=0}^{\ell}c_s\prod_{i=1}^{m_s}f_{s,i}:\ell,m_s\in\mathbb N_0,\ c_s\in\R,\ f_{s,i}\in\{p,p',p'',\ldots\}\right\}.
\end{equation}
Notice that if $m_s=0$, then by convention the empty product $\prod_{i=1}^{0}f_{s,i}$ equals $1$.

\begin{remark}
This definition comes from the calculations of the successive partial derivatives $F_x,F_{xx},\ldots$.
For example, in the simple case when $g_0=1,g_1=0,g_2=-1$, by setting $Z=\exp \left( -\int_{0}^{T-t} p^2\left(X_s\right) d s\right)$, we have $F(t,x)=\mathbb{E}\left[Z|X_0 =x\right]$ by the Markov property of $X$. Since we can write $X_s=e^{bs}X_0+Y_s$, where $Y_s$ does not depend on $X_0$, this means that formally, $F_x=\mathbb{E}\left[Z_x|X_0=x\right]$, where $Z_x=-Z\int_{0}^{T-t} 2e^{bs}(pp')\left(X_s\right) d s$ and $pp'\in \mathcal{A}_D(p)$ appears. Higher order derivatives like $Z_{xx},Z_{xxx}$ also formally generate elements in $\mathcal{A}_D(p)$. 
\end{remark}

Now we prove some properties of power series that are negligible to double factorial. For the result below, we recall the convention $(-1)!!=(0)!!=1$. Notice that although the proof is a little cumbersome, the essential idea is the observation that $\ell!!$ behaves like
$(\lfloor \frac{\ell}{2} \rfloor )!2^{\frac{\ell}{2}}$ and $\sum_{k=0}^\ell \binom{\ell}{k} =2^\ell$.
\begin{lemma}
\label{l3.6}
There exists a constant $C>0$ such that for all $\ell\geq2$, we have $$\sum_{k=0}^{\ell}\frac{1}{(k-1)!!(\ell-k-1)!!}\leq \frac{C{\ell}^32^{\frac{\ell}{2}}}{(\ell-1)!!}.$$ 
\end{lemma}
\begin{proof}
For $k=0,\ell$, we have $\frac{1}{(k-1)!!(\ell-k-1)!!}=\frac{1}{(\ell-1)!!}$. So we only need to prove that for a constant $C>0$,
\begin{align}
\sum_{k=1}^{\ell-1}\frac{1}{(k-1)!!(\ell-k-1)!!}=\sum_{k=0}^{\ell-2}\frac{1}{k!!(\ell-2-k)!!}\leq \frac{C{\ell}^32^{\frac{\ell}{2}}}{(\ell-1)!!}.
\end{align}
\textit{Case 1:} $\ell=2\ell_1$ is an even number. If $k=2k_1+1$ is an odd number, $k!!(\ell-2-k)!!=(2k_1+1)!!(\ell-3-2k_1)!!\geq \frac{(2k_1+1)!!(\ell-1-2k_1)!!}{\ell-1-2k_1}\geq \frac{(2k_1)!!(\ell-2-2k_1)!!}{\ell}=\frac{2^{\ell_1-1}k_1!(\ell_1-1-k_1)!}{\ell}$. If $k=2k_1$ is an even number, then $k!!(\ell-2-k)!!=2^{\ell_1-1}k_1!(\ell_1-1-k_1)!$. Therefore,
\begin{align}
\sum_{k=0}^{\ell-2}\frac{1}{k!!(\ell-2-k)!!} &\leq \sum_{k_1=0}^{\ell_1-1}\frac{\ell+1}{2^{\ell_1-1}k_1!(\ell_1-1-k_1)!} =\frac{\ell+1}{(\ell_1-1)!}=\frac{(\ell+1)2^{\ell_1-1}}{(\ell-2)!!}\\ &\leq  \frac{(\ell+1)(\ell-1)2^{\ell_1-1}}{(\ell-1)!!} \leq \frac{\ell^2 2^{\frac{\ell}{2}}}{(\ell-1)!!}.
\end{align}

\textit{Case 2:} $\ell=2\ell_1+1$ is an odd number, then
\begin{align}
&\sum_{k=0}^{\ell-2}\frac{1}{k!!(\ell-2-k)!!} \leq \sum_{k=0}^{\ell-3}\frac{1}{k!!(\ell-3-k)!!}+\frac{1}{(\ell-2)!!},
\end{align}
and since $\ell-3$ is even, as in Case 1, we have for a constant $C>0$ that
\begin{align}
\sum_{k=0}^{\ell-3}\frac{1}{k!!(\ell-3-k)!!}+\frac{1}{(\ell-2)!!}\leq\frac{(\ell-1)^22^{\frac{\ell-1}{2}}}{(\ell-2)!!}+\frac{1}{(\ell-2)!!}\leq \frac{C\ell^32^{\frac{\ell}{2}}}{(\ell-1)!!}.
\end{align} 
Therefore there exists a constant $C$ such that for all $\ell$, $\sum_{k=0}^{\ell-2}\frac{1}{k!!(\ell-2-k)!!}\leq \frac{C\ell^32^{\frac{\ell}{2}}}{(\ell-1)!!}$. 
\end{proof}

\begin{lemma}
\label{lemB.12}
The following statements are true.
\begin{enumerate}[1)]
    \item If $p$ is negligible to double factorial, then so is $q(x)=\int_{0}^xp(y)dy$.
   
    \item  The set of power series which are negligible to double factorial is closed under addition, differentiation, scalar multiplication, and multiplication. Thus if $p$ is negligible to double factorial, then for any $q\in \mathcal{A}_D(p)$, $q$ is also negligible to double factorial.
    \item If $p(x)=\sum p_kx^k$ is  negligible to double factorial, then so is $q(x)=\sum p_kc^kx^k$, where $c$ is a constant.
    \item If $p$ is negligible to double factorial, then so is $|p|$.
\end{enumerate}
\end{lemma}

\begin{proof}
\begin{enumerate}[1)]
    \item Note that $q(x)=\sum_{k=1}^{\infty}\frac{p_{k-1}}{k}x^k$. It follows that  
    \begin{align}
\limsup_{k\rightarrow\infty}(|q_k|(k-1)!!)^{\frac{1}{k}}&=\limsup_{k\rightarrow\infty}((|p_{k-1}|\frac{(k-1)!!}{k})^{\frac{1}{k-1}})^{\frac{k-1}{k}}\\ 
&\leq \limsup_{k\rightarrow\infty}((|p_{k-1}|(k-2)!!)^{\frac{1}{k-1}})^{\frac{k-1}{k}}=0,
    \end{align} 
    using that  $\frac{(k-1)!!}{k}\leq \frac{k!!}{k}=(k-2)!!$.
    \item Take $p,q$ negligible to double factorial. Note that      $q'(x)=\sum_{k=0}^{\infty}(k+1)q_{k+1}x^k$. Then, 
    \begin{align}
    \limsup_{k\rightarrow\infty}(|q'_k|(k-1)!!)^{\frac{1}{k}}&=\limsup_{k\rightarrow\infty}((|q_{k+1}|(k+1)!!)^{\frac{1}{k+1}})^{\frac{k+1}{k}},
    \end{align}
then by $(k+1)!!\leq (k+1)k!!$,
    \begin{align}
        \limsup_{k\rightarrow\infty}(|q'_k|(k-1)!!)^{\frac{1}{k}} &\leq\limsup_{k\rightarrow\infty}((|q_{k+1}|k!!)^{\frac{1}{k+1}})^{\frac{k+1}{k}}(k+1)^{\frac{1}{k}}=0.
    \end{align}
    
    By $|a+b|^{\frac{1}{k}}\leq |a|^{\frac{1}{k}}+|b|^{\frac{1}{k}}$, $\limsup_{k\rightarrow\infty}(|q_k+p_k|(k-1)!!)^{\frac{1}{k}}=0$, and $\limsup_{k\rightarrow\infty}(|cq_k|(k-1)!!)^{\frac{1}{k}}=0$ for any constant $c$.

    Next we show that the power series $pq$ is also negligible to double factorial. For any $\epsilon>0$, there exists $N$, such that for all $k>N$, we have that $(|q_k|(k-1)!!)\leq\epsilon^k,(|p_k|(k-1)!!)\leq\epsilon^k$. So there is a constant $C_{\epsilon}>0$ such that $(|q_k|(k-1)!!)\leq C_{\epsilon}\epsilon^k,(|p_k|(k-1)!!)\leq C_{\epsilon}\epsilon^k$ for all $k\in\mathbb N_0$.    
    Thus by Lemma \ref{l3.6}, {for $\ell\geq 2$},
    \begin{align}
    &|(pq)_\ell|=|\sum_{k=0}^\ell p_{k}q_{\ell-k}|\\
    &\leq\sum_{k=0}^\ell|p_{k}||q_{\ell-k}|\leq \sum_{k=0}^\ell \frac{{C_{\epsilon}}^2\epsilon^\ell}{(k-1)!!(\ell-k-1)!!}\leq \frac{{C_{\epsilon}}^2C(\sqrt{2}\epsilon)^\ell \ell^3}{(\ell-1)!!},
    \end{align}
    thus $\limsup_{\ell\rightarrow\infty}(|(pq)_\ell|(\ell-1)!!)^{\frac{1}{\ell}}\leq \sqrt{2}\epsilon$. Since $\epsilon>0$ is arbitrary, we obtain  $$\limsup_{\ell\rightarrow\infty}(|(pq)_\ell|(\ell-1)!!)^{\frac{1}{\ell}}=0.$$
    Therefore, $pq$ is also negligible to double factorial.
    \item Notice that $$\limsup_{k\rightarrow\infty}(|q_k|(k-1)!!)^{\frac{1}{k}}=|c|\limsup_{k\rightarrow\infty}(|p_k|(k-1)!!)^{\frac{1}{k}}=0.$$
    \item This is obvious by noticing that taking the modulus of $p$ has no impact on the behaviour of $\limsup_{k\rightarrow\infty}(|p_k|(k-1)!!)^{\frac{1}{k}}$.
\end{enumerate} 
\end{proof}

\subsection{Regularity in space}
\label{section6.1}

In this subsection, in preparation for the Feynman--Kac formula in Section~\ref{section6.2}, we prove the regularity in $x$ of $f$ given by \eqref{eq:defsmallf}, i.e., that $f$ is $C^\infty$ in $x$ and such that the partial derivatives $\frac{\partial^kf}{\partial x^k}$ are bounded in a suitable sense.

\begin{theorem}\label{B.24}
The function $f$ given by \eqref{eq:defsmallf}  is $C^\infty$ in $x$, i.e. the partial derivatives $\frac{\partial^kf}{\partial x^k}$ exist and are continuous in $(t,x)$, for $k\in \mathbb N$. Furthermore, for each $k\in \mathbb N$, there exists a power series $q_k$ negligible to double factorial and  a constant $C_k$ independent of $(t,x)$, such that 
\begin{align}\label{eq:boundfk}
\left|\frac{\partial^kf}{\partial x^k}(t,x) \right|\leq |q_k|(|x|)+C_k, \quad t\in [0,T],  \;  x \in \R.  
\end{align}
\end{theorem}

Before proving Theorem~\ref{B.24}, we first simplify the expression of $f$ in \eqref{eq:defsmallf} to get rid of the conditioning. Thanks to the Markovianity of $X$, we can write $f$ after a change of variable as

\begin{align}
    f(t,x) &= \mathbb{E}\left[ Z(t,x) \right],\label{eq:smallf2}
\end{align}
with 
\begin{align}
Z(t,x) := \exp \bigg( &\int_{0}^{T-t}\sum_{i=1}^3u_i(t+s)p_i\left(e^{bs}x+w(s) +\widetilde{W}_s\right) d s\\& \quad \quad +v(T)q(e^{b(T-t)}x+w(T-t)+\widetilde{W}_{T-t})\bigg).\label{def_z}
\end{align}

Notice $X_t = e^{bt}x+w(t) +\widetilde{W}_t$ when $X_0 = x$, with $w(t):=a\int_{0}^{t} e^{{b} (t-u)}du$ and ${\widetilde{{W}}_{t}}:=c\int_{0}^{t} e^{{b}(t-u)}  d W_{u}$. From now on, we consider mainly the representation \eqref{eq:smallf2} for $f(t,x)$.

The main idea for proving   Theorem~\ref{B.24} consists in taking  successive derivatives in $x$ inside the expectation above and applying the dominated convergence theorem. 
To illustrate this, we first formally calculate $\frac{\partial f}{\partial x}$.  Notice that

the derivative in $x$ for the term inside the exponential of $Z(t,x)$ is:
\begin{align}\label{h_func}
h(x)&:=\int_{0}^{T-t}\sum_{i=1}^3u_i(t+s) e^{bs}p_i'\left(e^{bs}x+w(s)+\widetilde{W}_s\right) d s\\
&+v(T) e^{b(T-t)}q'\left(e^{b(T-t)}x+w(T-t)+\widetilde W_{T-t}\right),
\end{align}
so that one expects 
\begin{align}\label{eq:partialfx}
    \frac{\partial f}{\partial x}(t,x)=\mathbb{E}\left[h(x) Z(t,x)\right].
    \end{align}

To prove this equality, we aim to use the definition of derivative and the dominated convergence theorem, which requires, roughly speaking, bounding $h(x)Z(t,x)$ in some sense which we clarify. Notice that $|Z(t,x)| \leq 1$ since the terms inside the exponential of $Z(t,x)$ have a nonpositive real part as a result of Lemma \ref{2.16}. Therefore, we only need to bound $h(x)$.
Notice that $u_i,v$ are continuous and thus bounded on $[0,T]$, $e^{bs}$ is also bounded on $[0,T]$, $p_i',q'$ are negligible to double factorial by part 2) of Lemma \ref{lemB.12}. So we only need to bound expressions of the form 
\begin{align}\label{eq:quantitiesq}
    \int_{0}^{t} q\left(e^{bs}x+w(s)+\widetilde W_s\right) d s\quad \text{and} \quad q\left(e^{bt}x+w(t)+\widetilde W_t\right),
\end{align} where $q$ is any power series negligible to double factorial. 

First of all, we introduce a new definition and a lemma to help bound in a certain sense the quantities in \eqref{eq:quantitiesq}. We also introduce the notation $\widetilde{W}^*_t:=\sup_{s\leq t}|\widetilde{W}_s|$, which will be useful for applying Doob's inequality.

\begin{definition}\label{estime}
 We say that a family of processes $(M_t(x))_{t\leq T,x\in \R}$ is \textbf{estimable} if there exists a power series $q$ negligible to double factorial such that for all fixed $t\in[0,T]$ and $x\in \R$,
 $$|M_t(x)|\leq |q|(|x|)+|q|(\widetilde W^*_T),\quad  a.s.$$
\end{definition}

We prove in Proposition~\ref{B.19} below that the expressions in \eqref{eq:quantitiesq} are estimable. 
For now, let us first prove that $|q|(|x|)+|q|(\widetilde W^*_T)$ is integrable via the following result.

\begin{lemma}
\label{B.15}
The set of estimable families of processes is closed under addition and multiplication. In addition, for all power series $q$ negligible to double factorial, $\mathbb{E}\left[|q|(\widetilde W^*_T)\right]<\infty$. Therefore, if $M_t(x)$ is estimable, then there exists a power series $q$ negligible to double factorial and a constant $C$, such that for all $t\in[0,T]$ and $x\in \R$, $\mathbb{E} [|M_t(x)|]\leq |q|(|x|)+C$.
\end{lemma}
\begin{proof}
For the first part, recall that by parts 2) and 4) of Lemma \ref{lemB.12}, the property of being  negligible to double factorial is closed under addition, multiplication, and taking absolute power series.
So the set of estimable families of processes is obviously closed under addition.
For the multiplication, it suffices to use the basic inequality $(a+b)(c+d)\leq a^2+b^2+c^2+d^2$.

To show that 
$\mathbb{E}\left[|q|(\widetilde W^*_T)\right] < \infty$, recall $\widetilde W_t=e^{bt}\int_{0}^{t} e^{-bs} c d W_{s}$, where $\int_{0}^{t} e^{-bs} c d W_{s}$ is a true martingale. There exists a positive constant $C_T$ such that $\frac{1}{C_T}<e^{bt}<C_T$ when $0\leq t\leq T$. We use $\overline{W}_t$ to denote $\int_{0}^{t} e^{-bs} c d W_{s}$ and $\overline{W}_t^*$ to denote its maximal process $\sup_{s\leq t}|\overline{W}_s|$.
By Doob's maximal inequality for $k\geq 2$, 
$$\mathbb{E}\left[(\widetilde W^*_T)^k\right]\leq C_T^k\mathbb{E}\left[(\overline{W}^*_T)^k\right]\leq C_T^k\left(\frac{k}{k-1}\right)^k\mathbb{E}\left[|{\overline{W}_T}|^k\right]\leq 4C_T^k \mathbb{E}\left[|{\overline{W}_T}|^k\right]\leq 4C_T^{2k} \mathbb{E}\left[|{\widetilde W_T}|^k\right].$$
For $k=1$, $$\mathbb{E}\left[\widetilde W^*_T\right]\leq\frac{1}{2}\left(1+\mathbb{E}\left[(\widetilde W^*_T)^2\right]\right)\leq 1+4C_T^4\mathbb{E}\left[({\widetilde W_T})^2\right].$$
Notice that  $\widetilde W_T$ has a centred normal distribution and that the constant  $C_T$ can be taken large enough so that the variance of ${\widetilde W_T}$ is smaller than $C_T^2$.
It follows that 
$\mathbb{E}\left[|{\widetilde W_T}|^k\right]\leq (C_T^2)^{\frac{k}{2}}(k-1)!!$. 
By modifying a little the coefficients $\vert q \vert_0, \vert q \vert_1, \vert q \vert_2$ in the power series $|q|$, which does not influence the conclusion, we obtain that $C=4\sum_{k=0}^{\infty}|q|_k(k-1)!!C_T^{3k}$ satisfies the requirements and is finite, since $q$ is  negligible to double factorial.

\end{proof}

\begin{proposition}
\label{B.19}
If $q$ is  negligible to double factorial, then the two expressions in \eqref{eq:quantitiesq} are estimable, and 
$$
\mathbb{E}\left[\left|\int_{0}^{t} q\left(e^{bs}x+w(s)+\widetilde{W}_s\right)d s\right|\right]+\mathbb{E}\left[\left| q\left(e^{bt}x+w(t)+\widetilde{W}_t\right)\right|\right] \leq |\tilde{q}|(|x|)+C, \quad t\in[0,T], \; x\in \R,
$$ 
{where $\tilde{q}$ is a power series negligible to double factorial, $C$ is a constant, and they both depend only on $T$ and $q$.}
\end{proposition}

\begin{proof} 
Fix $t\in [0,T]$ and $x\in \R$. 
We  first bound $ q\left(e^{bt}x+w(t)+\widetilde W_t\right)$:
\begin{align}
\left| q\left(e^{bt}x+w(t)+\widetilde{W}_t\right)\right|
&\leq |q|\left(e^{bt}|x|+|w(t)|+|\widetilde{W}_t|\right)\leq\sum_{k=0}^{\infty}|q|_k\left((3e^{bt}|x|)^k+|3w(t)|^k+|3\widetilde{W}_t|^k\right)
\end{align}
by Lemma \ref{B.6}. Since $e^{bt}$ and $w(t)$ are both bounded on $[0,T]$,  there exists a constant $ C_T\geq 0$ such that
\begin{align}
\sum_{k=0}^{\infty}|q|_k\left((3e^{bt}|x|)^k+|3w(t)|^k+|3\widetilde{W}_t|^k\right)  \leq& \sum_{k=0}^{\infty}|q|_k\left(|3C_Tx|^k+(3C_T)^k+|3\widetilde{W}_t|^k\right) \\
\leq&|q|(|3C_Tx|)+|q|(|3C_T|)+|q|(3\widetilde W^*_T).
\end{align}
Setting $\tilde{q}(x)=|q|(3(C_T+1)x)+|q|(|3C_T|)$, we get $\left|q\left(e^{bt}x+w(t)+\widetilde{W}_t\right)\right|\leq \tilde{q}(|x|)+\tilde{q}(\widetilde W^*_T)=|\tilde{q}|(|x|)+|\tilde{q}|(\widetilde W^*_T)$. By Lemma \ref{lemB.12}, $|q|$ is also  negligible to double factorial, and $
\tilde{q}$ is still  negligible to double factorial. So  $q\left(e^{bt}x+w(t)+\widetilde W_t\right)$ is estimable.  
An application of Lemma \ref{B.15} yields the bound for  the expectation. Furthermore, the term {$\int_{0}^t q\left(e^{bs}x+w(s)+\widetilde W_s\right)ds$} is clearly bounded by $T|\tilde{q}|(|x|)+T|\tilde{q}|(\widetilde W^*_T)$.
So we can update $\tilde q$ to  $\tilde{q}(x)=(T\vee1)|q|(3(C_T+1)x)+(T\vee1)|q|(|3C_T|)$ to obtain the claimed bound and end the proof. 

\end{proof} 

Now going back to \eqref{eq:partialfx}, it suffices to show that for $x\in \R$,
$$\lim_{\Delta x \rightarrow 0}\frac{f(t,x+\Delta x)-f(t,x)}{\Delta x
}=\lim_{\Delta x\rightarrow 0}\mathbb{E}\left[\frac{Z(t,x+\Delta x)-Z(t,x)}{\Delta x}\right]=\mathbb{E}\left[h(x)Z(t,x)\right].$$

Notice that $\frac{\partial Z(t,x)}{\partial x}=h(x)Z(t,x),$ to apply the dominated convergence theorem to deduce the second equality, it suffices to prove that $\frac{Z(t,x+\Delta x)-Z(t,x)}{\Delta x}$ is dominated by an integrable random variable for $\Delta x\in [-1,1]\setminus\{0\}$. In the following, we suppose $\Delta x\in [-1,1]\setminus\{0\}$.

First of all, as $\frac{\partial Z(t,x)}{\partial x}=h(x)Z(t,x)$, 
$\frac{\partial \Re (Z(t,x))}{\partial x}=\Re h(x)\Re Z(t,x)-\Im h(x)\Im Z(t,x)$, then 
\begin{align}\label{eq:mean_value}
\left|\Re\frac{Z(t,x+\Delta x)-Z(t,x)}{\Delta x}\right|&\leq \sup_{y\in [x-1,x+1]}|\Re h(y)\Re Z(t,y)-\Im h(y)\Im Z(t,y)| \\ &\leq \sup_{y\in[x-1,x+1]}|h(y)|,
\end{align}
where we use Cauchy--Schwarz and $|Z(t,x)|\leq1$. The same bound also works for $\Im \frac{Z(t,x+\Delta x)-Z(t,x)}{\Delta x}$. To bound $|h(y)|$, $y\in[x-1,x+1]$, recall $h(x)$ from \eqref{h_func}, and that $u_i,v$ are bounded, $e^{bs}$ is bounded on $[0,T]$, $p_i',q'$ are negligible to double factorial, we need only to bound $\left|\int_{0}^{t} q\left(e^{bs}y+w(s)+\widetilde W_s\right) d s\right|$ and $\left| q\left(e^{bt}y+w(t)+\widetilde W_t\right)\right|,$ for $q$ an arbitrary power series negligible to double factorial.

By the proof of Proposition~\ref{B.19}, $\left|\int_{0}^{t} q\left(e^{bs}y+w(s)+\widetilde W_s\right) d s\right|\leq |\tilde{q}|(|y|)+|\tilde{q}|(\widetilde W^*_T)\leq |\tilde{q}|(|x|+1)+|\tilde{q}|(\widetilde W^*_T)$, with $|\tilde{q}|(|x|+1)+|\tilde{q}|(\widetilde W^*_T)$ being integrable. Moreover, $\left| q\left(e^{bt}y+w(t)+\widetilde W_t\right)\right|$ can be bounded in the same way. Thus for $\Delta x\in[-1,1]\setminus\{0\}$, $\left|\frac{Z(t,x+\Delta x)-Z(t,x)}{\Delta x}\right|$ is dominated by $\sup_{y\in[x-1,x+1]}|h(y)|$, which justifies the use of the dominated convergence theorem as $\Delta x\to0$.

 Therefore, \eqref{eq:partialfx} holds. Of course, we should like to prove by induction that 
$$\frac{\partial^kf}{\partial x^k}(t,x)=\mathbb{E}\left[H_k(t,x)Z(t,x) \right]$$
holds for all $k \in \mathbb N$ to show that $f$ is $C^{\infty}$ in $x$ , where 
\begin{align}\label{eq:gk}
    H_{k+1 }(t,x)=\frac{\partial H_k(t,x)}{\partial x}+H_k(t,x)h(x), \quad  H_0(t,x)=1, \quad  H_1(t,x)=h(x).
\end{align}

To achieve this, we need to prove first that $H_k$ is well-defined. Then similarly as before, we bound $H_{k+1}(x)$ so that we can by induction apply dominated convergence to get 
\begin{align}
&\lim_{\Delta x \rightarrow 0}
 \frac{\frac{\partial^kf}{\partial x^k}(t,x+\Delta x)-\frac{\partial^kf}{\partial x^k}(t,x)}{\Delta x
}\\
=&\lim_{\Delta x \rightarrow 0}\mathbb{E}\left[\frac{H_{k}(t,x+\Delta x)Z(t,x+\Delta x)-H_{k}(t,x)Z(t,x)}{\Delta x}\right]\\
=&{\mathbb{E}\left[H_{k+1}(t,x)Z(t,x)\right]}.
\end{align}

For $H_2(t,x)$, we need to compute $\frac{\partial h(x)}{\partial x}$, recall that $h(x)$ is just the sum of a Riemann integral of a continuously differentiable function and a differentiable function with respect to $x$ for a fixed $\omega$ outside a nullset. Therefore,
\begin{align}
\frac{\partial h(x)}{\partial x}&=\int_{0}^{T-t}\sum_{i=1}^3u_i(t+s) e^{2bs}p_i''\left(e^{bs}x+w(s)+\widetilde{W}_s\right) d s \\
&+v(T) e^{2b(T-t)}q''\left(e^{b(T-t)}x+w(T-t)+\widetilde W_{T-t}\right).
\end{align}

Thus $H_2$ can be obtained explicitly, and similarly for $H_3,H_4,\ldots, H_k,\ldots$. This procedure will differentiate many times the function $h(x)$; so it is useful to define:
\begin{align}
h_k(x):&=\int_{0}^{T-t}\sum_{i=1}^3u_i(t+s) e^{kbs}p_i^{(k)}\left(e^{bs}x+w(s)+\widetilde{W}_s\right) d s\\
&+v(T) e^{kb(T-t)}q^{(k)}\left(e^{b(T-t)}x+w(T-t)+\widetilde W_{T-t}\right).
\end{align}

\begin{definition}
\label{B.20}
We define the set
\begin{equation}\label{eq:defAx}
\mathcal A_x(h):=\left\{\sum_{s=0}^{\ell}c_s\prod_{i=1}^{m_s}h_{s,i}:\ell,m_s\in\mathbb N_0,\ c_s\in\R,\ h_{s,i}\in\{h_1,h_2,h_3,\ldots\}\right\}
\end{equation}
as the $\R$-algebra generated by higher order derivatives of the function $h(x)$. We call the $h_k$ generating elements of $\mathcal A_x(h)$.
\end{definition}

In the next result, we characterise $H_k$ defined in \eqref{eq:gk}.

\begin{lemma}
\label{B.23}
For every $k\in\mathbb N$, 
the function $H_k$ in \eqref{eq:gk} is well-defined. In addition, $H_k$ is differentiable in $x$, continuous in $(t,x)$ and estimable.
\end{lemma}

\begin{proof}
First, note the following three facts:
\begin{enumerate}[1)]
    \item $h(x)=h_1(x)$ is a generating element of $\mathcal{A}_x(h)$,
    \item $H_0=1\in \mathcal{A}_x(h)$, 
    \item All generating elements $h_k(x)$ of $\mathcal{A}_x(h)$ are differentiable in $x$ and $\frac{\partial h_k(x)}{\partial x}=h_{k+1}(x)$. Thus by Leibniz' rule, if $g \in \mathcal{A}_x(h)$, then $\frac{\partial g}{\partial x}$ exists and $ \in \mathcal{A}_x(h)$.
\end{enumerate}

By induction and noticing that $\mathcal{A}_x(h)$ is closed under linear combinations and multiplication, $H_k$ are well-defined and in $\mathcal{A}_x(h)$. For fixed $\omega \in \Omega$ outside a nullset, the generating elements $h_k$ are continuous in $(t,x)$ and differentiable in $x$. This implies that $H_k$ is also continuous in $(t,x)$ and differentiable in $x$. 

We know $u_i,v$ are continuous (and thus bounded), $e^{kbs}$ is bounded on $[0,T]$, $p_i^{(k)},q^{(k)}$ are negligible to double factorial by part 2) of Lemma \ref{lemB.12}. By Proposition \ref{B.19},  generating elements $h_k$ of $\mathcal{A}_x(h)$ are estimable. Since $H_k\in \mathcal{A}_x(h)$ is a linear combination of finite products of generating elements of $\mathcal{A}_x(h)$,
by Lemma \ref{B.15}, $H_k$ is also estimable. 
\end{proof}

We are now ready to prove Theorem \ref{B.24}.

\begin{proof}[Proof of Theorem \ref{B.24}]
Take $H_k$ as defined in \eqref{eq:gk} which are estimable by Lemma \ref{B.23}. Notice that by Lemma \ref{B.15} and the fact that $|Z(t,x)|\leq1$,
the term $H_k(t,x)Z(t,x)$ is estimable with its expectation bounded by $|q_k|(|x|)+C_k$, where $q_k$ is a power series negligible to double factorial and $C_k$ is a constant. 

It suffices to prove that for $k\in\mathbb N_0$,
\begin{equation}\label{eq:partial_fk}
\frac{\partial^kf}{\partial x^k}(t,x)=\mathbb{E}\left[H_k(t,x)Z(t,x) \right],    
\end{equation}
and it is continuous in $(t,x).$

We use induction to prove~\eqref{eq:partial_fk}. We first notice that for the case $k=0$,  $\frac{\partial^k f}{\partial x^k}(t,x)=f(t,x)= \mathbb E[Z(t,x)]$. Next, suppose that it is true for up to case $k$, recall the definition of $Z$ from \eqref{def_z} and that $Z_x=hZ$, where $h$ is defined in  \eqref{h_func}. Thus $(H_kZ)_x=H_{k+1}Z$. Then we repeat the calculation of~\eqref{eq:mean_value} for real part and imaginary part separately, and noticing that $|Z(t,x)|\leq 1$, for $\Delta x\in [-1,1]\setminus\{0\}$:
\begin{align*}
\left|\frac{1}{\Delta x}\Big( H_k(t,x+\Delta x)Z(t,x+\Delta x)-H_k(t,x)Z(t,x) \Big)\right|
&\leq  \sup_{y\in[x-1,x+1]}|H_{k+1}(t,y)|.  
\end{align*}

Recall that by Lemmas~\ref{B.6} and~\ref{B.23}, we have for fixed $x$ and $y\in[x-1,x+1]$ that $|H_{k+1}(t,y)|\leq|q_{k+1}|(|x|+1)+|q_{k+1}|(\widetilde W^*_T)$ with bounded expectation by Lemma \ref{B.15}, where  $q_{k+1}$ is  negligible to double factorial. 
Applying dominated convergence theorem, for $\Delta x\rightarrow 0$ we have:
\begin{align}
\frac{\partial^{k+1}f}{\partial x^{k+1}}(t,x)&=\lim_{\Delta x\rightarrow 0}\frac{1}{\Delta x}\mathbb{E}\Big[ H_k(t,x+\Delta x)Z(t,x+\Delta x)-H_k(t,x)Z(t,x) \Big]\\&=\mathbb{E}\Big[\lim_{\Delta x\rightarrow 0}\frac{1}{\Delta x} \left(H_k(t,x+\Delta x)Z(t,x+\Delta x)-H_k(t,x)Z(t,x) \right)\Big]\\&=
\mathbb{E}\Bigg[H_{k+1}(t,x)Z(t,x)\Bigg].
\end{align}

Therefore \eqref{eq:partial_fk} and thus \eqref{eq:boundfk} hold for all $k\in\mathbb N_0$. Finally, for the continuity of $\frac{\partial^kf}{\partial x^k}(t,x)$, it suffices to notice that before taking the expectation, the random variable $H_k(t,x)Z(t,x)$ is continuous with respect to $(t,x)$. Moreover, $|Z(t,x)|\leq1$ and $H_k(t,x)$ is bounded, for $t\in[0,T]$ and $x\in[x_0-1,x_0+1]$, by an inequality analogous to the one just above for $H_{k+1}(t,y)$. Then dominated convergence allows us to take the limit $(t,x)\to(t_0,x_0)$ and obtain continuity of $\mathbb E[H_k(t,x)Z(t,x)]$ at $(t_0,x_0)$, for fixed $(t_0,x_0)\in[0,T]\times\R$.
\end{proof}

\subsection{Feynman--Kac}
\label{section6.2}
In this section, we derive the Feynman--Kac formula in Theorem~\ref{thmB.1}. Since we do not have a priori that $f_t$ exists in our setting, we prove its existence and obtain the Feynman--Kac formula at the same time.

First, we introduce a result that will be useful later.

\begin{lemma}\label{expz}
Let $z\in\mathbb C$ such that $\Re(z)\leq0$. Then $|\exp(z)-1|\leq3|z|$.
\end{lemma}
\begin{proof}
We write $z=x+yi$, $x\leq0$ and $y\in\R$. Then
\begin{align*}
|\exp(z)-1|&=|\exp(x+yi)-1|\\
&=\big|\exp(yi)(\exp(x)-1)+\exp(yi)-1\big|\\
&\leq |\exp(x)-1|+|\cos y-1|+|\sin y|.
\end{align*}
If $x=0$, then $|\exp(x)-1|=0$. If $x<0$, $|\frac{\exp(x)-1}{z}|\leq|\frac{\exp(x)-1}{x}|\leq1$. If $y=0$, then $|\cos y-1|=0$, $|\sin y|=0$. If $y\neq0$, $|\frac{\cos y-1}{z}|\leq|\frac{\cos y-1}{y}|\leq1$, $|\frac{\sin y}{z}|\leq|\frac{\sin y}{y}|\leq1$.
\end{proof}

\begin{theorem}\label{thmB.1}
The function $f$ given by \eqref{eq:defsmallf} or equivalently \eqref{eq:smallf2} is $C^\infty$ in $x$, $C^1$ in $t$, and solves the partial differential equation (PDE)
\begin{equation}\label{eq:PDE}
\begin{cases}
f_t(t,x)+f_x(t,x)(a+bx)+\dfrac12c^2f_{xx}(t,x)+f(t,x)\displaystyle\sum_{i=1}^3u_i(t)p_i(x)=0,\\[1ex]
f(T,x)=\exp\big(v(T)q(x)\big).
\end{cases}
\end{equation}
\end{theorem}

\begin{proof}
By Theorem~\ref{B.24}, $f$ is $C^\infty$ in $x$. We compute the infinitesimal generator
\[
\lim_{r\downarrow0}\frac1r\mathbb E^x[f(t_0,X_r)-f(t_0,x)]
\]
in two ways to obtain the PDE and the existence of $f_t$ at the same time. Here $\mathbb E^x$ means the expectation with $X_0=x$.

\emph{First way:} Fix $t_0$; we compute the quantity using the classical definition of the generator. By It\^o's formula,
\[
df(t_0,X_t)=f_x\,dX_t+\frac12f_{xx}\,d\langle X\rangle_t
=cf_x\,dW_t+\left(\frac12c^2f_{xx}+(a+bX_t)f_x\right)dt.
\]
Note that $f_x,f_{xx}$ are evaluated at $(t_0,X_t)$. By Theorem~\ref{B.24} and part~2) of Lemma~\ref{lemB.12}, $f_x(t,x)$, $(a+bx)f_x(t,x)$, $f_{xx}(t,x)$, $|f_x(t,x)|^2$ are all bounded by $|q|(|x|)+C$, where $q$ is negligible to double factorial. By Proposition~\ref{B.19}, $\mathbb E^x[\int_0^r |f_x(t_0,X_t)|^2dt]$ is bounded by $|\widetilde q|(|x|)+C$ for $r\in[0,T]$ with $\widetilde q$ negligible to double factorial. Thus $\int cf_x(t_0,X_t)dW_t$ is a square-integrable true martingale on $[0,T]$. Similarly, the Riemann integral $\int_0^r(\frac12c^2f_{xx}(t_0,X_t)+(a+bX_t)f_x(t_0,X_t))dt$ is estimable and has finite expectation. In addition, we have
\begin{align*}
&\left|\frac1r\int_0^r\left(\frac12c^2f_{xx}(t_0,X_t)+(a+bX_t)f_x(t_0,X_t)\right)dt\right|\\
&\qquad\leq\sup_{s\in[0,r]}\left(\frac12c^2|f_{xx}(t_0,X_s)|+|(a+bX_s)f_x(t_0,X_s)|\right).
\end{align*}
where the right-hand side is still estimable and dominated by a random variable with finite expectation. By the dominated convergence theorem,
\begin{equation}\label{first_way}
\begin{aligned}
\lim_{r\downarrow0}\frac1r\mathbb E^x[f(t_0,X_r)-f(t_0,x)]
&=\lim_{r\downarrow0}\frac1r\mathbb E^x\left[\int_0^r\left(\frac12c^2f_{xx}(t_0,X_t)+(a+bX_t)f_x(t_0,X_t)\right)dt\right]\\
&=\frac12c^2f_{xx}(t_0,x)+(a+bx)f_x(t_0,x).
\end{aligned}
\end{equation}

\emph{Second way:} Compute the generator $\lim_{r\downarrow0}\frac1r\mathbb E^x[f(t_0,X_r)-f(t_0,x)]$ directly by the Markov property of $X$. We define $Z_t=\exp(\int_0^t\sum_{i=1}^3u_i(T-t+s)p_i(X_s)ds)$ and $Y_t=v(T)q(X_t)$, and apply the Markov property of $X$ to the representation of $f(t,x)$ in \eqref{eq:defsmallf}, i.e.,
\[
f(t,x)=\mathbb E\left[\exp\left(\int_0^{T-t}\sum_{i=1}^3u_i(t+s)p_i(X_s)ds+v(T)q(X_{T-t})\right)\Mid X_0=x\right],
\]
and by the Markov property for $r\in[0,t]$,
\[
f(t,X_r)=\mathbb E\left[\exp\left(\int_r^{T-t+r}\sum_{i=1}^3u_i(t-r+s)p_i(X_s)ds+v(T)q(X_{T-t+r})\right)\Mid X_r\right].
\]
By the tower property of conditional expectations, we have
\begin{align*}
\mathbb E^x[f(t,X_r)]
&=\mathbb E^x\left[\mathbb E\left[\exp\left(\int_r^{T-t+r}\sum_{i=1}^3u_i(t-r+s)p_i(X_s)ds+v(T)q(X_{T-t+r})\right)\Mid X_r\right]\right]\\
&=\mathbb E^x\left[\exp\left(\int_r^{T-t+r}\sum_{i=1}^3u_i(t-r+s)p_i(X_s)ds+v(T)q(X_{T-t+r})\right)\right]\\
&=\mathbb E^x\left[Z_{T-t+r}\exp\left(-\int_0^r\sum_{i=1}^3u_i(t-r+s)p_i(X_s)ds\right)\exp(Y_{T-t+r})\right].
\end{align*}
Therefore we get
\begin{align}\label{computation_2}
&\frac1r\mathbb E^x[f(t_0,X_r)-f(t_0,x)]\nonumber\\
&=\frac1r\mathbb E^x\Bigg[Z_{T-t_0+r}\exp\left(-\int_0^r\sum_{i=1}^3u_i(t_0-r+s)p_i(X_s)ds\right)\exp(Y_{T-t_0+r})
-Z_{T-t_0}\exp(Y_{T-t_0})\Bigg]\nonumber\\
&=\frac1r\mathbb E^x\left[\exp(Y_{T-t_0+r})Z_{T-t_0+r}-\exp(Y_{T-t_0})Z_{T-t_0}\right]\nonumber\\
&\quad+\frac1r\mathbb E^x\left[\exp(Y_{T-t_0+r})Z_{T-t_0+r}\left(\exp\left(-\int_0^r\sum_{i=1}^3u_i(t_0-r+s)p_i(X_s)ds\right)-1\right)\right]\nonumber\\
&=\frac1r\big(f(t_0-r,x)-f(t_0,x)\big)\nonumber\\
&\quad+\frac1r\mathbb E^x\left[\exp(Y_{T-t_0+r})Z_{T-t_0+r}\left(\exp\left(-\int_0^r\sum_{i=1}^3u_i(t_0-r+s)p_i(X_s)ds\right)-1\right)\right].
\end{align}
Now we want to apply dominated convergence when $r\downarrow0$ to the term
\[
\frac1r\mathbb E^x\left[\exp(Y_{T-t_0+r})Z_{T-t_0+r}\left(\exp\left(-\int_0^r\sum_{i=1}^3u_i(t_0-r+s)p_i(X_s)ds\right)-1\right)\right].
\]
To dominate the term inside the expectation, we use
\[
\Re(Y_t)=\Re(v(T)q(X_t))=0,
\]
thus $|\exp(Y_{T-t_0+r})|=1$. In addition,
\begin{align*}
&\left|\frac1rZ_{T-t_0+r}\left(\exp\left(-\int_0^r\sum_{i=1}^3u_i(t_0-r+s)p_i(X_s)ds\right)-1\right)\right|\\
&=\left|\frac1r\exp\left(\int_r^{T-t_0+r}\sum_{i=1}^3u_i(t_0-r+s)p_i(X_s)ds\right)
\left(1-\exp\left(\int_0^r\sum_{i=1}^3u_i(t_0-r+s)p_i(X_s)ds\right)\right)\right|.
\end{align*}
Since $\Re(u_1)\leq0$, $\Re(u_2)=\Re(u_3)=0$ and $p_1=p^2$, we therefore get
\begin{align*}
\Re\left(\int_r^{T-t_0+r}\sum_{i=1}^3u_i(t_0-r+s)p_i(X_s)ds\right)
&=\int_r^{T-t_0+r}\Re\big(u_1(t_0-r+s)p_1(X_s)\big)ds\leq0,
\end{align*}
and thus $|\exp(\int_r^{T-t_0+r}\sum_{i=1}^3u_i(t_0-r+s)p_i(X_s)ds)|\leq1$. Therefore we only need to dominate the term
\[
\frac1r\left(\exp\left(\int_0^r\sum_{i=1}^3u_i(t_0-r+s)p_i(X_s)ds\right)-1\right).
\]
Since $\Re(\int_0^r\sum_{i=1}^3u_i(t_0-r+s)p_i(X_s)ds)\leq0$, by Lemma~\ref{expz}, we just need to dominate the term
\[
\left|\frac1r\int_0^r\sum_{i=1}^3u_i(t_0-r+s)p_i(X_s)ds\right|
\leq\sum_{i=1}^3\|u_i\|_\infty\sup_{s\in[0,r]}|p_i(X_s)|.
\]
Notice that the $u_i$ are bounded and the $p_i$ are negligible to double factorial. Thus we obtain by Proposition~\ref{B.19} that $|p_i(X_s)|$, $s\in[0,T]$, are dominated when conditioned on $X_0=x$. Thus by dominated convergence,
\begin{align*}
&\lim_{r\downarrow0}\frac1r\mathbb E^x\left[\exp(Y_{T-t_0+r})Z_{T-t_0+r}\left(\exp\left(-\int_0^r\sum_{i=1}^3u_i(t_0-r+s)p_i(X_s)ds\right)-1\right)\right]\\
&\qquad=-\sum_{i=1}^3u_i(t_0)p_i(x)f(t_0,x),
\end{align*}
which is finite. Equating this with \eqref{first_way}, the term $\lim_{r\downarrow0}\frac1r(f(t_0-r,x)-f(t_0,x))$ is thus well defined. Of course, one can also replace $t_0$ by $t_0+r$ and perform similar computations via the above two methods on the quantity
\[
\lim_{r\downarrow0}\frac1r\mathbb E^x[f(t_0+r,X_r)-f(t_0+r,x)]
\]
to obtain the existence of $f_t$ and also the PDE in \eqref{eq:PDE}. A subtle point to note is that for the first method, we do not apply It\^o's formula to $f(t_0+r,X_r)$ for the variable $r$, which requires a priori the existence of the partial derivative $f_t$. Instead, we obtain the existence of $f_t$ by applying It\^o's formula to $X_r$ in $f(t_0+u,X_r)$ and evaluate at $u=r$ to get
\begin{align*}
&\frac1r\mathbb E^x[f(t_0+r,X_r)-f(t_0+r,x)]=\frac1r\mathbb E^x\left[\int_0^r\left(\frac12c^2f_{xx}(t_0+r,X_t)+(a+bX_t)f_x(t_0+r,X_t)\right)dt\right]
\end{align*}
and take the limit for $r\downarrow0$ as above. The continuity of $f_t$ is directly obtained from the continuity of the other terms in this equation, with the boundary condition at $T$ coming from the continuity of $f$ and its definition.
\end{proof}

\subsection{Infinite-dimensional ODE}
\label{section6.3}
In this section, we obtain a solution to the system \eqref{ode_t613} of ODEs below by comparing the derivatives of both sides of the PDE \eqref{eq:PDE} when $x=0$.

Since $f$ is continuous by Theorem~\ref{B.24}, if $f$ in addition does not vanish, we can define $\log f$ such that $\exp(\log f)=f$ and $\log f(T,0)=v(T)q(0)$; see Lemma~\ref{log}. Note that $\log f$ is also continuous in $(t,x)$, and
\[
\exp(\log f)(T,x)=f(T,x)=\exp\big(v(T)q(x)\big)
\]
with $\log f(T,x)=v(T)q(x)$.

\begin{theorem}\label{thmB.2}
If $f$ does not vanish, then
\[
\phi_k(t):=\left.\frac1{k!}\partial_x^k\log f(T-t,x)\right|_{x=0},\qquad t\geq0,
\]
solves the system of ODEs
\begin{equation}\label{ode_t613}
\begin{aligned}
\phi_k'(t)&=\sum_{i=1}^3u_i(T-t)(p_i)_k\\
&\quad+bk\phi_k(t)+a(k+1)\phi_{k+1}(t)+\frac{c^2(k+2)(k+1)}2\phi_{k+2}(t)\\
&\quad+\frac{c^2}2\big(\widetilde\phi(t)*\widetilde\phi(t)\big)_k,
\qquad\widetilde\phi_k(t)=(k+1)\phi_{k+1}(t),\\
\phi_k(0)&=v(T)q_k.
\end{aligned}
\end{equation}
\end{theorem}
In order to prove Theorem~\ref{thmB.2}, we first link the PDE \eqref{eq:PDE} to the system \eqref{ode_t613} of ODEs. We start with the following result, which follows from \citet[Theorem~9.41]{rudin1976}.

\begin{lemma}\label{B.26}
If a function $h$ is $C^1$ in $(t,x)$ and one of the partial derivatives $\frac{\partial^2h}{\partial t\partial x}$ and $\frac{\partial^2h}{\partial x\partial t}$ exists and is continuous, then they both exist and are equal.
\end{lemma}

Now we can prove Theorem~\ref{thmB.2}.

\begin{proof}[Proof of Theorem~\ref{thmB.2}]
Since $F(t,x)$ does not vanish, we can define $\log F$ as in Definition~\ref{defB.2}, i.e.,
\[
\exp\big(\log F(t,x)\big)=F(t,x),\qquad\log F(T,0)=0.
\]
We define $\log f(t,x)=\log F(t,x)+v(t)q(x)$ such that $\exp(\log f(t,x))=f(t,x)$. By Theorem~\ref{thmB.1}, $f$ is $C^\infty$ in $x$ and $C^1$ in $t$. Then since $\log f$ is continuous, $\log f$ is by Lemma~\ref{C1Cinf} also $C^\infty$ in $x$ and $C^1$ in $t$. In addition, one can deduce the PDE
\begin{align*}
-(\log f)_t(t,x)&=\sum_{i=1}^3u_i(t)p_i(x)+(\log f)_x(t,x)(a+bx)\\
&\quad+\frac12c^2(\log f)_{xx}(t,x)+\frac12c^2(\log f)_x^2(t,x),\\
\log f(T,x)&=v(T)q(x).
\end{align*}
The boundary condition comes from $\log f(T,0)=v(T)q(0)$, the continuity of $\log f$ and the fact that $\exp(\log f)(T,x)=f(T,x)=\exp(v(T)q(x))$.

Now substitute $g(t,x)=\log f(T-t,x)$. We obtain for $g$ the equation
\begin{equation}\label{equag}
\begin{aligned}
g_t(t,x)&=\sum_{i=1}^3u_i(T-t)p_i(x)+g_x(t,x)(a+bx)+\frac12c^2g_{xx}(t,x)+\frac12c^2g_x^2(t,x),\\
g(0,x)&=v(T)q(x).
\end{aligned}
\end{equation}
Since the right-hand side of \eqref{equag} is $C^\infty$ in $x$, so is $g_t$. Recall that by definition, $\frac{\partial^k}{\partial x^k}g(t,x)|_{x=0}=k!\phi_k(t)$. We aim to take the $k$th partial derivative with respect to $x$ on both sides of \eqref{equag} at $x=0$ to deduce \eqref{ode_t613}. For the left-hand side of \eqref{ode_t613}, Lemma~\ref{B.26} allows to interchange $\partial t,\partial x$, and one can prove by induction the equality
\[
\left.\frac{\partial^k}{\partial x^k}\frac{\partial}{\partial t}g(t,x)\right|_{x=0}
=\left.\frac{\partial}{\partial t}\frac{\partial^k}{\partial x^k}g(t,x)\right|_{x=0}=k!\phi_k'(t).
\]
The right-hand side of \eqref{ode_t613} can be deduced by noting that $\frac{\partial^k}{\partial x^k}g(t,x)|_{x=0}=k!\phi_k(t)$ for all $k$. To obtain that
\[
\left.\frac{\partial^k}{\partial x^k}g_x^2(t,x)\right|_{x=0}=k!\big(\widetilde\phi(t)*\widetilde\phi(t)\big)_k,
\qquad\widetilde\phi_k(t)=(k+1)\phi_{k+1}(t),
\]
first notice that $\frac{\partial^k}{\partial x^k}g(t,x)|_{x=0}=k!\phi_k(t)$ implies $\frac{\partial^k}{\partial x^k}g_x(t,x)|_{x=0}=k!\widetilde\phi_k(t)$. The convolution follows naturally by applying the general Leibniz rule. Finally, the initial condition of \eqref{ode_t613} can be easily deduced from the fact that $g(0,x)=v(T)q(x)$.
\end{proof}

\subsection{Putting everything together}
\label{section6.4}
\begin{proof}[Proof of Theorem~\ref{thm1.2}]
Since $p$ is negligible to double factorial by assumption and $g_0:[0,T]\to\R$, $g_1,g_2:[0,T]\to\mathbb C$ are all continuously differentiable functions with $\Re(g_1)=0$ and $\Re(g_2)\leq0$, $F(t,x)$ in \eqref{eq:defF} is well defined. We recall from Lemma~\ref{2.16} that
\[
F(t,x)=\mathbb E\left[\exp\left(\int_t^T\sum_{i=1}^3u_i(s)p_i(X_s)ds+v(T)q(X_T)-v(t)q(X_t)\right)\Mid X_t=x\right]
\]
with $u_i,v$ continuous in $[0,T]$, $\Re(u_1)\leq0$, $\Re(u_2)=\Re(u_3)=\Re(v)=0$, $p_1=p^2$, and $p_i,q$ negligible to double factorial. Note that $F(t,x)=f(t,x)\exp(-v(t)q(x))$. Theorem~\ref{B.24} yields that $F$ is continuous.

If $F$ in addition does not vanish, then neither does $f$. By Lemma~\ref{log}, we can define $\log f,\log F$ such that $\log F=\log f-v(t)q(x)$. By Theorems~\ref{B.24} and~\ref{thmB.1}, $f$ is $C^\infty$ in $x$ and $C^1$ in $t$, and so is $F$. By Lemma~\ref{C1Cinf}, $\log f,\log F$ are also $C^\infty$ in $x$ and $C^1$ in $t$. Then we can take the $k$th partial derivative with respect to $x$ for $\log F=\log f-v(t)q(x)$ around $x=0$. Then the system \eqref{ode_t613} of ODEs in Theorem~\ref{thmB.2} induces the system \eqref{eq:Ric}--\eqref{eq:Ric_init} of ODEs by setting $\psi_k:=\phi_k-v(T-t)q_k$ with $\phi_k$ defined in \eqref{ode_t613}, with the coefficients of the ODE \eqref{eq:Ric} coming from the precise definition of $(u_i,p_i)_{i\in\{1,2,3\}}$ in Appendix~\ref{sectionA}.
\end{proof}

\section{Proof of Theorem \ref{T:Discrete}}\label{section7}

Before proving Theorem \ref{T:Discrete}, we introduce the following definition of the function $f_n$.

\begin{definition}\label{deffn}
For $\mu_n=\frac{T}{n}\sum_{i=0}^{n-1}\delta_{iT/n}$, with $\delta_t$ the Dirac measure at point $t$, we define
\begin{align}\label{eq:defsmallfn}
f_n(t,x)&=\mathbb{E}\left[\exp\left(\int_t^T\sum_{i=1}^3u_i(s)p_i(X_s)\,\mu_n(ds)+v(T)q(X_T)\right)\Mid X_t=x\right]\\
&=F_n(t,x)\exp\bigl(v(t)q(x)\bigr),\qquad t\leq T,\ x\in\mathbb{R},\notag
\end{align}
with $F_n$ as in Definition \ref{defFn}. In the following, we use the notation $\mu_n(\cdot+t)$ for $t\in[0,T]$ to denote the measure $\frac{T}{n}\sum_{i=0}^{n-1}\delta_{iT/n-t}$ obtained by translating $\mu_n$ by $-t$. Thus for any continuous functions $g,h$, we have
\[
\int_0^{T-t}g(s+t)h(s)\,\mu_n(ds+t)
=\frac{T}{n}\sum_{\substack{0\leq i\leq n-1\\t\leq iT/n\leq T}}g\left(\frac{iT}{n}\right)h\left(\frac{iT}{n}-t\right).
\]
\end{definition}

By the Markov property of $X$, we can remove the conditioning and write $f_n(t,x)$ equivalently as
\begin{align}
f_n(t,x)=\mathbb{E}\Bigg[\exp\Bigg(&\int_0^{T-t}\sum_{i=1}^3u_i(s+t)p_i\bigl(e^{bs}x+w(s)+\widetilde W_s\bigr)\,\mu_n(ds+t)\notag\\
&+v(T)q\bigl(e^{b(T-t)}x+w(T-t)+\widetilde W_{T-t}\bigr)\Bigg)\Bigg].\label{def_fn}
\end{align}
Notice that $X_t=e^{bt}x+w(t)+\widetilde W_t$ for $X_0=x$; recall that $w(t):=a\int_0^t e^{b(t-u)}\,du$ and $\widetilde W_t:=c\int_0^t e^{b(t-u)}\,dW_u$. Since
\[
\int_0^{T-t}\Re\bigl(u_1(s+t)p^2(e^{bs}x+w(s)+\widetilde W_s)\bigr)\,\mu_n(ds+t)\leq0
\]
under the conditions on $u_i,p_i,v$ and $q$ from Lemma \ref{2.16}, the function $f_n$ is well defined.

The main advantage of considering the functionals $(f_n,F_n)$ over $(f,F)$ is their property of being entire in $x$ as shown in Theorem \ref{analytic} below. As we shall see, the results for $(f,F)$ in Section \ref{section6} can be easily extended to $(f_n,F_n)$. We now lay out the following steps to prove Theorem \ref{T:Discrete}:
\begin{itemize}
\item In Section \ref{subsection7.1}, we prove that $f_n$ is entire in $x\in\mathbb R$, without any additional conditions on $f_n$ or $F_n$.
\item In Section \ref{subsection7.2}, we prove discretised versions of the results in Section \ref{section6}, in the case when $f_n(t,\cdot)$ (or equivalently $F_n(t,\cdot)$) does not vanish on $\mathbb R$.
\item In Section \ref{subsection7.3}, we prove Theorem \ref{T:Discrete} under the assumption that the extension $(F_n(t,\cdot))_c$ of $(F_n(t,\cdot))$ to the complex plane does not vanish.
\end{itemize}

\begin{lemma}\label{C.5}
$f_n(t,x)\to f(t,x)$ when $n\to\infty$ for all $t\leq T$ and $x\in\mathbb R$.
\end{lemma}
\begin{proof}
This follows directly from dominated convergence as the Riemann sum converges to the Riemann integral and the exponential has a modulus of at most $1$.
\end{proof}

\subsection{Entire property}\label{subsection7.1}

In this section, we do not need any additional condition imposed on $f_n$ or $F_n$ about their zeros.

\begin{theorem}\label{analytic}
For fixed $t$, $f_n(t,\cdot)$ is entire in $x\in\mathbb R$.
\end{theorem}

We first introduce a lemma that will be useful in proving Theorem \ref{analytic}. The idea of the proof is standard and comes from exploiting the fact that the Gaussian density is entire. A similar argument has been used for single marginals in \citet*[Lemma 5.7]{cuchiero2023signature}.

\begin{lemma}\label{C.6}
Suppose that $(W_t)_{t\geq0}$ is a centered Gaussian process with independent increments such that $W_t-W_s$ has non-zero variance when $t>s$. Let the partition be $0=t_0<t_1<\cdots<t_n=T$. Fix $g:\mathbb R^n\to\mathbb C$ bounded measurable. Then the expectation $\mathbb E[g(x+W_{t_1},x+W_{t_2},\ldots,x+W_{t_n})]$ is entire in $x$.
\end{lemma}

\begin{proof}
Denoting the variance of $W_{t_i}-W_{t_{i-1}}$ as $v_i>0$, we have
\begin{align*}
&\mathbb E[g(x+W_{t_1},x+W_{t_2},\ldots,x+W_{t_n})]\\
=&\left(\prod_{i=1}^n\frac{1}{\sqrt{2\pi v_i}}\right)
\int_{\mathbb R^n}g(y_1,y_1+y_2,\ldots,y_1+\cdots+y_n)\times\exp\left(-\frac{(y_1-x)^2}{2v_1}-\sum_{i=2}^n\frac{y_i^2}{2v_i}\right)\,dy_1dy_2\cdots dy_n.
\end{align*}
With $g(y)=g(y_1,y_1+y_2,\ldots,y_1+\cdots+y_n)$, we have
\begin{align*}
&\mathbb E[g(x+W_{t_1},x+W_{t_2},\ldots,x+W_{t_n})]\\
&=\exp\left(-\frac{x^2}{2v_1}\right)\left(\prod_{i=1}^n\frac{1}{\sqrt{2\pi v_i}}\right)\times\int_{\mathbb R^n}g(y)\exp\left(-\sum_{i=1}^n\frac{y_i^2}{2v_i}\right)\exp\left(\frac{xy_1}{v_1}\right)\,dy_1dy_2\cdots dy_n\\
&=\exp\left(-\frac{x^2}{2v_1}\right)\left(\prod_{i=1}^n\frac{1}{\sqrt{2\pi v_i}}\right)\times\int_{\mathbb R^n}g(y)\exp\left(-\sum_{i=1}^n\frac{y_i^2}{2v_i}\right)\sum_{k=0}^{\infty}\frac{1}{k!}\left(\frac{xy_1}{v_1}\right)^k\,dy_1dy_2\cdots dy_n\\
&=\exp\left(-\frac{x^2}{2v_1}\right)\left(\prod_{i=1}^n\frac{1}{\sqrt{2\pi v_i}}\right)\sum_{k=0}^{\infty}\frac{x^k}{k!}\times\int_{\mathbb R^n}g(y)\exp\left(-\sum_{i=1}^n\frac{y_i^2}{2v_i}\right)\left(\frac{y_1}{v_1}\right)^k\,dy_1dy_2\cdots dy_n,
\end{align*}
where we applied Fubini in the last step since
\begin{align*}
&\int_{\mathbb R^n}\sum_{k=0}^{\infty}|g(y)|\exp\left(-\sum_{i=1}^n\frac{y_i^2}{2v_i}\right)\frac{1}{k!}\left|\frac{xy_1}{v_1}\right|^k\,dy_1dy_2\cdots dy_n\\
&\leq\|g\|_\infty\int_{\mathbb R^n}\exp\left(-\sum_{i=1}^n\frac{y_i^2}{2v_i}+\frac{|xy_1|}{v_1}\right)\,dy_1dy_2\cdots dy_n<\infty.
\end{align*}
This implies an infinite radius of convergence for
\begin{align*}
&\exp\left(-\frac{x^2}{2v_1}\right)\left(\prod_{i=1}^n\frac{1}{\sqrt{2\pi v_i}}\right)\times\sum_{k=0}^{\infty}\frac{x^k}{k!}\int_{\mathbb R^n}g(y)\exp\left(-\sum_{i=1}^n\frac{y_i^2}{2v_i}\right)\left(\frac{y_1}{v_1}\right)^k\,dy_1dy_2\cdots dy_n,
\end{align*}
i.e., $\mathbb E[g(x+W_{t_1},x+W_{t_2},\ldots,x+W_{t_n})]$ is equal to a power series with an infinite radius of convergence. So it is entire in $x\in\mathbb R$.
\end{proof}

\begin{proof}[Proof of Theorem \ref{analytic}]
We define $t_j=jT/n$, $j=0,1,\ldots,n-1$. Then
\begin{align*}
f_n(t,x)=\mathbb E\Bigg[\exp\Bigg(&\frac{T}{n}\sum_{i=1}^3\sum_{t\leq t_j\leq T}u_i(t_j)p_i\bigl(e^{b(t_j-t)}x+w(t_j-t)+\widetilde W_{t_j-t}\bigr)\\
&+v(T)q\bigl(e^{b(T-t)}x+w(T-t)+\widetilde W_{T-t}\bigr)\Bigg)\Bigg],
\end{align*}
where the exponential inside the expectation is bounded by $1$ and we notice that $e^{b(t_j-t)}x+\widetilde W_{t_j-t}=e^{b(t_j-t)}(x+e^{-b(t_j-t)}\widetilde W_{t_j-t})$. Recall that $\widetilde W_t=\int_0^t e^{b(t-u)}c\,dW_u$. Therefore $(e^{-bt}\widetilde W_t)=c\int e^{-bu}\,dW_u$ is a Gaussian process satisfying the conditions of Lemma \ref{C.6} since $c\neq0$. Then an application of Lemma \ref{C.6} yields the result.
\end{proof}

\subsection{Discretised versions of the theorems in Section \ref{section6}}\label{subsection7.2}

In this section, we prove discretised versions of the results in Section \ref{section6}, in the case when $f_n(t,\cdot)$ (or equivalently $F_n(t,\cdot)$) does not vanish on $\mathbb R$.

\begin{remark}
Recalling Corollary \ref{Cor:Riccati}, the conditions that $g_1=0$ and $g_2$ is $\mathbb R$-valued with $g_2\leq0$ guarantee that $f_n(t,\cdot)$ and $F_n(t,\cdot)$ do not vanish on $\mathbb R$.
\end{remark}

\begin{theorem}\label{lC.5}
The function $f_n$ defined in \eqref{def_fn} is $C^\infty$ in $x$ with $\partial^k f_n/\partial x^k$ continuous in $(t,x)$ in the region $(jT/n,(j+1)T/n)\times\mathbb R$ for $j=0,1,\ldots,n-1$. In addition, there exist a power series $q_k$ negligible to double factorial and a constant $C_k$ independent of $(t,x)$ such that $|\partial^k f_n/\partial x^k|\leq |q_k|(|x|)+C_k$.
\end{theorem}
\begin{proof}
All the results in Section \ref{section6} dealing with $f$ being $C^\infty$ in $x$ and the bound on $\partial^k f/\partial x^k$ can be easily adapted to $f_n$ by replacing the integral with a finite sum.
\end{proof}

\begin{theorem}\label{C.2}
For $j=0,1,\ldots,n-1$, the function $f_n$
\begin{enumerate}[(i)]
\item is continuous in the region $(jT/n,(j+1)T/n]\times\mathbb R$. In addition, $f_n$ is $C^1$ in $t$ and $C^\infty$ in $x$ in the region $(jT/n,(j+1)T/n)\times\mathbb R$, $j=0,1,\ldots,n-1$;
\item solves the PDE
\begin{align*}
(f_n)_t+(f_n)_x(a+bx)+\frac12c^2(f_n)_{xx}&=0,\qquad t\notin\operatorname{supp}(\mu_n)\cup\{T\},\\
\lim_{s\searrow t}f_n(s,x)&=f_n(t,x)\exp\left(-\frac{T}{n}\sum_{i=1}^3u_i(t)p_i(x)\right),\qquad t\in\operatorname{supp}(\mu_n),\\
f_n(T,x)&=\exp\bigl(v(T)q(x)\bigr),
\end{align*}
with $\operatorname{supp}(\mu_n)=\{0,T/n,2T/n,\ldots,(n-1)T/n\}$.
\end{enumerate}
\end{theorem}

\begin{proof}
The proof follows along the same lines as the proof of Theorem \ref{thmB.1}. By Theorem \ref{lC.5} and Lemma \ref{lemB.12}, $(f_n)_x(t,x)$ and $|(f_n)_x(t,x)|^2$ are bounded by $|q|(|x|)+C$, where $q$ is negligible to double factorial. Therefore by Proposition \ref{B.19}, the stochastic integral $\int (f_n)_x(s,x)\,dW_s$ is a square-integrable true martingale. The Riemann integral is still dominated, allowing the interchange of limit and expectation, and thus we can compute the infinitesimal generator $\lim_{r\searrow0}\frac1r(\mathbb E^x[f_n(t_0,X_r)]-f_n(t_0,x))$ as in the first way in Theorem \ref{thmB.1}.

The second way of computing the generator in Theorem \ref{thmB.1} can also be adapted to compute $\lim_{r\searrow0}\frac1r(\mathbb E^x[f_n(t_0,X_r)]-f_n(t_0,x))$ thanks to the Markov property of $X$. Indeed, define
\[
Z_t^n=\exp\left(\int_0^t\sum_{i=1}^3u_i(T-t+s)p_i(X_s)\,\mu_n(ds+T-t)\right),\qquad Y_t=v(T)q(X_t).
\]
We have
\begin{align*}
&\frac1r\mathbb E^x[f_n(t_0,X_r)-f_n(t_0,x)]\\
=&\frac1r\mathbb E^x\left[Z^n_{T-t_0+r}\exp(Y_{T-t_0+r})
\exp\left(-\int_0^r\sum_{i=1}^3u_i(t_0-r+s)p_i(X_s)\,\mu_n(ds+t_0-r)\right)
-Z^n_{T-t_0}\exp(Y_{T-t_0})\right]\\
=&\frac1r\mathbb E^x\left[\exp(Y_{T-t_0+r})Z^n_{T-t_0+r}
-\exp(Y_{T-t_0})Z^n_{T-t_0}\right]\\
+&\frac1r\mathbb E^x\left[Z^n_{T-t_0+r}\exp(Y_{T-t_0+r})
\left(\exp\left(-\int_0^r\sum_{i=1}^3u_i(t_0-r+s)p_i(X_s)\,\mu_n(ds+t_0-r)\right)-1\right)\right]\\
=&\frac1r\bigl(f_n(t_0-r,x)-f_n(t_0,x)\bigr)\\
+&\frac1r\mathbb E^x\left[Z^n_{T-t_0+r}\exp(Y_{T-t_0+r})
\left(\exp\left(-\int_0^r\sum_{i=1}^3u_i(t_0-r+s)p_i(X_s)\,\mu_n(ds+t_0-r)\right)-1\right)\right].
\end{align*}
Notice that if $t_0\notin\operatorname{supp}(\mu_n)\cup\{T\}$, which is a finite set, there exists $r_0>0$ such that $r_0<|t_0-t|$ for all $t\in\operatorname{supp}(\mu_n)\cup\{T\}$. Thus for $0<r<r_0$,
\[
\exp\left(-\int_0^r\sum_{i=1}^3u_i(t_0-r+s)p_i(X_s)\,\mu_n(ds+t_0-r)\right)=1.
\]
So similarly to the proof of Theorem \ref{thmB.1}, we deduce that $(f_n)_t$ exists and satisfies the equation $0=(f_n)_t+(f_n)_x(a+bx)+\frac12c^2(f_n)_{xx}$. The boundary condition for $t\notin\operatorname{supp}(\mu_n)\cup\{T\}$ can be deduced by applying the definition of $\mu_n$.
\end{proof}

\begin{remark}
Theorem \ref{C.2} shows that for fixed $x$, $f_n(t,x)$ is only left-continuous for $t\in[0,T]$ and discontinuous at $t\in\operatorname{supp}(\mu_n)$.
\end{remark}

We now define $\log f_n$ in the following result.

\begin{lemma}\label{logfn}
If $f_n(t,\cdot)$ does not vanish for $x\in\mathbb R$, then there exists $\log f_n$ such that $\exp(\log f_n)=f_n$, and in the regions $(jT/n,(j+1)T/n]\times\mathbb R$, $j=0,1,\ldots,n-1$, $\log f_n(t,x)$ is continuous in $(t,x)$ and
\begin{align*}
\lim_{s\searrow t}\log f_n(s,x)&=\log f_n(t,x)-\frac{T}{n}\sum_{i=1}^3u_i(t)p_i(x),\qquad t\in\operatorname{supp}(\mu_n),\\
\log f_n(T,x)&=v(T)q(x).
\end{align*}
\end{lemma}

\begin{proof}
We want to use Lemma \ref{log} to define $\log f_n$ from $f_n$. However, Theorem \ref{C.2} shows that $f_n$ is discontinuous when $t\in\operatorname{supp}(\mu_n)$, with
\[
\lim_{s\searrow t}f_n(s,x)=f_n(t,x)\exp\left(-\frac{T}{n}\sum_{i=1}^3u_i(t)p_i(x)\right),\qquad t\in\operatorname{supp}(\mu_n).
\]
We need to first define
\[
\widetilde f_n(t,x)=f_n(t,x)\exp\left(-\frac{T}{n}\sum_{\substack{s\in\operatorname{supp}(\mu_n)\\s\geq t}}\sum_{i=1}^3u_i(s)p_i(x)\right),
\]
where $\widetilde f_n$ is continuous on $[0,T]\times\mathbb R$ and $\widetilde f_n(T,x)=f_n(T,x)=\exp(v(T)q(x))$. Since $f_n$ is non-zero, so is $\widetilde f_n$. We can now apply Lemma \ref{log} to yield the existence of $\widetilde g$ continuous on $[0,T]\times\mathbb R$ such that $\exp(\widetilde g)=\widetilde f_n$, and to specify $\widetilde g$, we choose $\widetilde g(T,0)=v(T)q(0)$. We can now define
\begin{equation}\label{def_log_fn}
\log f_n(t,x)=\widetilde g(t,x)+\frac{T}{n}\sum_{\substack{s\in\operatorname{supp}(\mu_n)\\s\geq t}}\sum_{i=1}^3u_i(s)p_i(x).
\end{equation}
As the double sum on the right-hand side is continuous in the regions $(jT/n,(j+1)T/n]\times\mathbb R$ for $j=0,1,\ldots,n-1$, the definition of $\log f_n$ in \eqref{def_log_fn} yields that
\begin{enumerate}[(i)]
\item $\exp(\log f_n)=f_n$;
\item in the regions $(jT/n,(j+1)T/n]\times\mathbb R$, $j=0,1,\ldots,n-1$, $\log f_n(t,x)$ is continuous in $(t,x)$;
\item $\lim_{s\searrow t}\log f_n(s,x)=\log f_n(t,x)-\frac{T}{n}\sum_{i=1}^3u_i(t)p_i(x)$, $t\in\operatorname{supp}(\mu_n)$.
\end{enumerate}
Since
\[
\exp(\log f_n(T,x))=\exp(\widetilde g(T,x))=\widetilde f_n(T,x)=f_n(T,x)=\exp(v(T)q(x)),
\]
we obtain $\log f_n(T,x)-v(T)q(x)\in\{2k\pi i:k\in\mathbb Z\}$. But $\log f_n(T,x)=\widetilde g(T,x)$ is continuous in $x$, and $\log f_n(T,0)-v(T)q(0)=\widetilde g(T,0)-v(T)q(0)=0$. So we deduce that $\log f_n(T,x)=v(T)q(x)$.
\end{proof}

We now introduce the discretised version of Theorem \ref{thmB.2}.

\begin{theorem}\label{C.3}
Suppose that $f_n(t,\cdot)$ does not vanish on $\mathbb R$. Then
\[
\phi_{n,k}(t)=\left.\frac{1}{k!}\partial_x^k\log f_n(T-t,x)\right|_{x=0},\qquad t\geq0,
\]
solves the system of ODEs
\begin{align}
\phi_{n,k}'(t)&=bk\phi_{n,k}(t)+a(k+1)\phi_{n,k+1}(t)+\frac{c^2(k+2)(k+1)}{2}\phi_{n,k+2}(t)\notag\\
&\quad+\frac{c^2}{2}\bigl(\widetilde\phi_n(t)*\widetilde\phi_n(t)\bigr)_k,
\qquad\widetilde\phi_{n,k}=(k+1)\phi_{n,k+1},\notag\\
&\hspace{3cm}t\in\{T-u:u\notin\operatorname{supp}(\mu_n)\cup\{T\}\},\notag\\
\phi_{n,k}(t)&=\lim_{s\nearrow t}\phi_{n,k}(s)+\frac{T}{n}\sum_{i=1}^3u_i(T-t)(p_i)_k,
\qquad t\in\{T-u:u\in\operatorname{supp}(\mu_n)\},\notag\\
\phi_{n,k}(0)&=v(T)q_k.\label{disc_ode}
\end{align}
\end{theorem}

\begin{proof}
By Theorem \ref{C.2}, in the regions $(jT/n,(j+1)T/n)\times\mathbb R$ for $j=0,1,\ldots,n-1$, $f_n$ is $C^1$ in $t$ and $C^\infty$ in $x$. Since $f_n\neq0$, $\log f_n$ in Lemma \ref{logfn} is by Lemma \ref{C1Cinf} $C^1$ in $t$ and $C^\infty$ in $x$ in these regions. Set $g_n(t,x)=\log f_n(T-t,x)$. Then $g_n$ satisfies the PDE
\begin{align*}
(g_n)_t(t,x)&=(g_n)_x(t,x)(a+bx)+\frac12c^2(g_n)_{xx}(t,x)+\frac12c^2((g_n)_x)^2,\\
&\hspace{2cm}t\in\{T-u:u\notin\operatorname{supp}(\mu_n)\cup\{T\}\},\\
\lim_{s\nearrow t}g_n(s,x)&=g_n(t,x)-\frac{T}{n}\sum_{i=1}^3u_i(T-t)p_i(x),
\qquad t\in\{T-u:u\in\operatorname{supp}(\mu_n)\},\\
g_n(0,x)&=v(T)q(x).
\end{align*}
Next we apply Lemma \ref{B.26} as in the proof of Theorem \ref{thmB.2}; then comparing the derivatives of both sides in $x$ yields the system of ODEs. Finally, the initial condition of \eqref{disc_ode} can be easily deduced from the fact that $g_n(0,x)=v(T)q(x)$.
\end{proof}

\subsection{Putting everything together}\label{subsection7.3}

In this section, we use the condition that the extension $(f_n(t,\cdot))_c$ of $f_n$ to the complex plane as in Definition \ref{extension} does not vanish for $x\in\mathbb C$. This condition is both necessary and sufficient to guarantee that $\log f_n$ is entire in $x\in\mathbb R$; see Lemma \ref{hololift} and Remark \ref{remarkB.4}.

\begin{lemma}\label{C.4}
If $(f_n)_c(t,\cdot)$ does not vanish on $\mathbb C$, then for $x\in\mathbb R$ and $t\in[0,T]$,
\[
f_n(t,x)=\exp\left(\sum_{k\geq0}\phi_{n,k}(T-t)x^k\right),
\]
where the $\phi_{n,k}$ are defined as in Theorem \ref{C.3}.
\end{lemma}

\begin{proof}
By Lemma \ref{hololift}, $\log f_n$ is entire on $\mathbb R$. Hence we deduce by a Taylor expansion of $\log f_n$ that $f_n=\exp(\sum\phi_{n,k}x^k)$ and $\phi_{n,k}(t)=\left.\frac{1}{k!}\partial_x^k\log f_n(T-t,x)\right|_{x=0}$.
\end{proof}

\begin{proof}[Proof of Theorem \ref{T:Discrete}]
By Lemma \ref{C.4}, $f_n(t,x)=\exp(\sum_{k\geq0}\phi_{n,k}(T-t)x^k)$, where the $\phi_{n,k}$ are defined in Theorem \ref{C.3}. We set $\psi_{n,k}(t)=\phi_{n,k}(t)-v(T-t)q_k$. By Lemma \ref{2.16} and Lemma \ref{lem3.1}, this is equivalent to
\begin{align*}
\psi_{n,k}(t)&=\phi_{n,k}(t)-\rho g_1(t)g_0(T-t)\frac{p_{k-1}}{ck},\qquad k\geq1,\\
\psi_{n,0}(t)&=\phi_{n,0}(t).
\end{align*}
Therefore by Lemma \ref{C.5},
\[
F(t,x)=f(t,x)\exp\bigl(-v(t)q(x)\bigr)=\lim_{n\to\infty}f_n(t,x)\exp\bigl(-v(t)q(x)\bigr),
\]
and by Lemma \ref{C.4},
\begin{align*}
f_n(t,x)\exp\bigl(-v(t)q(x)\bigr)
&=\exp\left(\sum_{k\geq0}\phi_{n,k}(T-t)x^k\right)\exp\bigl(-v(t)q(x)\bigr)\\
&=\exp\left(\sum_{k\geq0}\bigl(\phi_{n,k}(T-t)-v(t)q_k\bigr)x^k\right)\\
&=\exp\left(\sum_{k\geq0}\psi_{n,k}(T-t)x^k\right).
\end{align*}
Recalling $F_n(t,x)=f_n(t,x)\exp(-v(t)q(x))$ by Definition \ref{deffn} ends the proof.
\end{proof}

\appendix

\section{Proof of Lemma \ref{2.16}}\label{sectionA}

\begin{lemma}\label{lem3.1}
Let $p$ be a power series with an infinite radius of convergence. Define the power series $r_1,r_2$ by
\[
r_1(x)=-\frac1c\left((a+bx)p(x)+\frac12c^2p'(x)\right),\qquad
r_2(x)=\frac1c\int_0^x p(y)\,dy,\qquad x\in\mathbb R,
\]
which again have an infinite radius of convergence. Then for any continuously differentiable function $h:[0,T]\to\mathbb C$, we have
\begin{align*}
\int_t^T h(s)p(X_s)\,dW_s
&=\int_t^T\bigl(h(s)r_1(X_s)-h'(s)r_2(X_s)\bigr)\,ds\\
&\quad+h(T)r_2(X_T)-h(t)r_2(X_t).
\end{align*}
\end{lemma}

\begin{proof}
\begin{sloppypar}
This follows from a straightforward application of It\^o's lemma to the process $(h(s)r_2(X_s))_{s\leq T}$ between $t$ and $T$ if we recall the dynamics of $X$ in \eqref{polynomial_model}.\qedhere
\end{sloppypar}
\end{proof}

\begin{proof}[Proof of Lemma \ref{2.16}]
Using $d(\log S_t)=-\frac{\sigma_t^2}{2}\,dt+\sigma_t\,dB_t$, we have
\begin{align*}
\int_t^T g_1(T-s)\,d(\log S_s)
&=\int_t^T-\frac12g_1(T-s)g_0^2(s)p^2(X_s)\,ds+\rho g_1(T-s)g_0(s)p(X_s)\,dW_s\\
&\quad+\sqrt{1-\rho^2}\,g_1(T-s)g_0(s)p(X_s)\,dW_s^\perp.
\end{align*}
As $\sqrt{1-\rho^2}\,g_1(T-s)g_0(s)p(X_s)$ is measurable with respect to $\mathcal F_t\vee\mathcal F_T^W$ and continuous for $s\in[t,T]$ and $(W_{t+s}^\perp-W_t^\perp)_{s\geq0}$ is a Brownian motion independent of $\mathcal F_t\vee\mathcal F_T^W$, the integral $\int_t^T\sqrt{1-\rho^2}\,g_1(T-s)g_0(s)p(X_s)\,dW_s^\perp$ is conditionally on $\mathcal F_t\vee\mathcal F_T^W$ Gaussian with conditional variance $\int_t^T(1-\rho^2)g_1^2(T-s)g_0^2(s)p^2(X_s)\,ds$. So we get
\begin{align*}
&\mathbb E\left[\exp\left(\int_t^T\sqrt{1-\rho^2}\,g_1(T-s)g_0(s)p(X_s)\,dW_s^\perp\right)\Mid\mathcal F_t\vee\mathcal F_T^W\right]\\
&=\exp\left(\frac12\int_t^T(1-\rho^2)g_1^2(T-s)g_0^2(s)p^2(X_s)\,ds\right).
\end{align*}
Since the terms $\int_t^T-\frac12g_1(T-s)g_0^2(s)p^2(X_s)\,ds+\rho g_1(T-s)g_0(s)p(X_s)\,dW_s$ and $\int_t^Tg_2(T-s)\sigma_s^2\,ds=\int_t^Tg_2(T-s)g_0^2(s)p^2(X_s)\,ds$ are measurable with respect to $\mathcal F_T^W$, we have
\begin{align*}
&\mathbb E\left[\exp\left(\int_t^Tg_1(T-s)\,d(\log S_s)+\int_t^Tg_2(T-s)\sigma_s^2\,ds\right)\Mid\mathcal F_t\vee\mathcal F_T^W\right]\\
&=\exp\Bigg(\int_t^T\left(-\frac12g_1(T-s)+g_2(T-s)\right)g_0^2(s)p^2(X_s)\,ds+\int_t^T\rho g_1(T-s)g_0(s)p(X_s)\,dW_s\Bigg)\\
&\quad\times\mathbb E\left[\exp\left(\int_t^T\sqrt{1-\rho^2}\,g_1(T-s)g_0(s)p(X_s)\,dW_s^\perp\right)\Mid\mathcal F_t\vee\mathcal F_T^W\right]\\
&=\exp\Bigg(\int_t^T\left(\frac12(1-\rho^2)g_1^2(T-s)-\frac12g_1(T-s)+g_2(T-s)\right)\\
&\hspace{2.5cm}\times g_0^2(s)p^2(X_s)\,ds+\int_t^T\rho g_1(T-s)g_0(s)p(X_s)\,dW_s\Bigg).
\end{align*}
Thus by the tower property of conditional expectations,
\begin{align*}
F(t,X_t)
&=\mathbb E\left[\exp\left(\int_t^Tg_1(T-s)\,d(\log S_s)+\int_t^Tg_2(T-s)\sigma_s^2\,ds\right)\Mid\mathcal F_t\right]\\
&=\mathbb E\Bigg[\exp\Bigg(\int_t^T\left(\frac12(1-\rho^2)g_1^2(T-s)-\frac12g_1(T-s)+g_2(T-s)\right)\\
&\hspace{2.5cm}\times g_0^2(s)p^2(X_s)\,ds+\int_t^T\rho g_1(T-s)g_0(s)p(X_s)\,dW_s\Bigg)\Mid\mathcal F_t\Bigg].
\end{align*}
Applying Lemma \ref{lem3.1} and choosing $h(t)=\rho g_1(T-t)g_0(t)$, we can rewrite
\begin{align*}
\int_t^T\rho g_1(T-s)g_0(s)p(X_s)\,dW_s
&=\int_t^T\bigl(h(s)r_1(X_s)-h'(s)r_2(X_s)\bigr)\,ds\\
&\quad+h(T)r_2(X_T)-h(t)r_2(X_t).
\end{align*}
Now we define $u_1(s)=(\frac12(1-\rho^2)g_1^2(T-s)-\frac12g_1(T-s)+g_2(T-s))g_0^2(s)$, $p_1(x)=p^2(x)$, $u_2(s)=h(s)$, $p_2(x)=r_1(x)$, $u_3(s)=h'(s)$, $p_3(x)=-r_2(x)$, $v(s)=h(s)$, $q(x)=r_2(x)$. Then we deduce the representations given in \eqref{eq:change_function}. Since $p$ is negligible to double factorial, $g_0,g_1,g_2$ are continuously differentiable and $\Re(g_1)=\Im(g_0)=0$, $\Re(g_2)\leq0$, we deduce that $u_i,v$ are continuous on $[0,T]$, $\Re(u_1)\leq0$, $\Re(u_2)=\Re(u_3)=\Re(v)=0$, $p_i,q$ are negligible to double factorial by Lemma \ref{lemB.12}, and
\[
F(t,x)=\mathbb E\left[\exp\left(\int_t^T\sum_{i=1}^3u_i(s)p_i(X_s)\,ds+v(T)q(X_T)-v(t)q(X_t)\right)\Mid X_t=x\right].
\]
\end{proof}

\section{A small remark on the complex logarithm}\label{A:complexlog}

For complex-valued functions such as $F$ in \eqref{eq:defF}, the definition of $\log F$ is not trivial, especially if the range of $F$ is not simply connected. We use the lifting property in algebraic topology.

\begin{lemma}\label{log}
Suppose $f:I\times\mathbb R\to\mathbb C\setminus\{0\}$ is a continuous function, where $I$ is an interval of $\mathbb R$. Then there exists a continuous function $g:I\times\mathbb R\to\mathbb C$ such that $\exp(g)=f$. If the value of $g$ at one point is specified, then $g$ is fully specified, i.e., there exists only one $g$ that satisfies these conditions in this case.
\end{lemma}

\begin{proof}
Notice that $p=\exp(z)$ is a covering projection from $\mathbb C$ to $\mathbb C\setminus\{0\}$, i.e., a local homeomorphism, and $I\times\mathbb R$ is both path-connected and locally path-connected and has a trivial fundamental group. Then by \citet[Proposition 1.33]{hatcher2002algebraic}, for the continuous function $f:I\times\mathbb R\to\mathbb C\setminus\{0\}$, there exists a continuous function $g:I\times\mathbb R\to\mathbb C$ such that $f=p\circ g=\exp(g)$. Furthermore, the domain $I\times\mathbb R$ is connected. Therefore by the uniqueness property \cite[Proposition 1.34]{hatcher2002algebraic} $g$ is uniquely specified once its value at a single point is given.
\end{proof}

\begin{lemma}\label{C1Cinf}
Let $f:I\times\mathbb R\to\mathbb C\setminus\{0\}$ be a continuous function, where $I$ is an interval of $\mathbb R$. Assume that $f$ is $C^1$ in the first variable $t$ and $C^\infty$ in the second variable $x$. Then the function $g$ defined in Lemma \ref{log} is also $C^1$ in $t$ and $C^\infty$ in $x$.
\end{lemma}

\begin{proof}
Suppose we define $g$ as in Lemma \ref{log}. We only need to prove that for each fixed $(t_0,x_0)\in I\times\mathbb R$, $g$ is $C^1$ in $t$ and $C^\infty$ in $x$ in a neighborhood of $(t_0,x_0)$. Since $f(t_0,x_0)\neq0$, there exists a $\delta>0$ such that for all $(t,x)\in I\times\mathbb R$ with $|t-t_0|+|x-x_0|<\delta$, we have $|f(t,x)-f(t_0,x_0)|<|f(t_0,x_0)|$. So for all such $(t,x)$, $f(t,x)$ is contained in an open ball $B$ with center $f(t_0,x_0)$ and radius $|f(t_0,x_0)|$. In particular, $B$ is simply connected and does not contain $0$. We define $h:B\to\mathbb C$ as
\[
h(z)=\int_{\ell_{f(t_0,x_0),z}}\frac1w\,dw+g(t_0,x_0),
\]
where $\ell_{f(t_0,x_0),z}$ is the line segment between $f(t_0,x_0)$ and $z$. We aim to prove that $g(t,x)=h(f(t,x))$ for all $(t,x)\in I\times\mathbb R$ with $|t-t_0|+|x-x_0|<\delta$.

By the definition of $h$, $h'(z)=1/z$; thus $(\exp(h(z))/z)'=0$ which yields that $\exp(h(z))/z$ is constant on $B$. In addition,
\[
\frac{\exp(h(f(t_0,x_0)))}{f(t_0,x_0)}=\frac{\exp(g(t_0,x_0))}{f(t_0,x_0)}=1.
\]
So $\exp(h(z))=z$ and thus $\exp(h(f(t,x)))=f(t,x)=\exp(g(t,x))$ for all $(t,x)\in I\times\mathbb R$ with $|t-t_0|+|x-x_0|<\delta$. Thus $h(f(t,x))-g(t,x)\in\{2k\pi i:k\in\mathbb Z\}$. Notice again that $h(f(t_0,x_0))-g(t_0,x_0)=0$, and by continuity, we have $g(t,x)=h(f(t,x))$ for all $(t,x)\in I\times\mathbb R$ with $|t-t_0|+|x-x_0|<\delta$. Notice that $h'(z)=1/z$; so $h$ is holomorphic on $B$. Therefore $h$ is $C^\infty$ on $B$ when regarding $B$ as a subset of $\mathbb R^2$. Then as a composition of $f$ and $h$, $g$ is also $C^1$ in $t$ and $C^\infty$ in $x$.
\end{proof}

\begin{lemma}\label{hololift}
Let $f:\mathbb R\to\mathbb C\setminus\{0\}$ be an entire function on $\mathbb R$ and $g:\mathbb R\to\mathbb C$ a continuous function with $\exp(g(x))=f(x)$. Then $g$ is entire on $\mathbb R$ if and only if the extension $f_c$ of $f$ to the complex plane as in Definition \ref{extension} does not vanish on $\mathbb C$.
\end{lemma}

\begin{proof}
``$\Leftarrow$'': Assume that $f_c:\mathbb C\to\mathbb C$ does not vanish. Then for $z\in\mathbb C$, we can define
\[
h(z)=\int_{\ell_{0,z}}\frac{f_c'(w)}{f_c(w)}\,dw+g(0),
\]
where $\ell_{0,z}$ is the line segment between $0$ and $z$. It follows that $h'(z)=f_c'(z)/f_c(z)$, showing that $h$ is holomorphic on $\mathbb C$. Moreover, for $F(z)=f_c(z)\exp(-h(z))$, we know that
\[
F(0)=f_c(0)\exp\bigl(-h(0)\bigr)=f(0)\exp\bigl(-g(0)\bigr)=1,\qquad F'(z)=0.
\]
Thus $F(z)=1$ for all $z\in\mathbb C$. Thus for all $x\in\mathbb R$,
\[
\exp\bigl(g(x)-h(x)\bigr)=\exp\bigl(g(x)\bigr)\exp\bigl(-h(x)\bigr)=f(x)\exp\bigl(-h(x)\bigr)=F(x)=1.
\]
Thus there exists a $k\in\mathbb Z$ such that $g(x)-h(x)=2k\pi i$. Because $h(0)=g(0)$, we obtain $g(x)=h(x)$ for $x\in\mathbb R$. In particular, $h$ is holomorphic on $\mathbb C$; so it is entire on $\mathbb C$ and thus on $\mathbb R$, and then so is $g$.

``$\Rightarrow$'': Assume $g:\mathbb R\to\mathbb C$ is entire on $\mathbb R$. Then it can be written in the form of a power series $g(x)=\sum_{n=0}^{\infty}a_nx^n$ with an infinite radius of convergence. Then $g_c(z)=\sum_{n=0}^{\infty}a_nz^n$ is entire on $\mathbb C$, and so is $\exp(g_c(z))$. Notice that $\exp(g_c(z))$ and $f_c(z)$ are both entire on $\mathbb C$, and $\exp(g_c(x))=\exp(g(x))=f(x)=f_c(x)$ for all $x\in\mathbb R$. Since the zeros of an entire function are isolated, except for the zero function, we have that $f_c(z)=\exp(g_c(z))$, and hence $f_c(z)\neq0$ for all $z\in\mathbb C$.
\end{proof}

\begin{remark}\label{remarkB.4}
We point out that for a real entire function $f:\mathbb R\to\mathbb C$ that does not vanish on $\mathbb R$, there exists a continuous function $\log f:\mathbb R\to\mathbb C$ such that $\exp(\log f)=f$, (choosing $I=[0,0]=\{0\}$ in Lemma \ref{log}). However, $\log f$ is not necessarily entire on $\mathbb R$. For instance, the function $f(x)=1+ix$ is entire on $\mathbb R$ and does not vanish on $\mathbb R$. However, $\log f=\frac12\log(1+x^2)+\arctan(x)i+2k\pi i$, $k\in\mathbb Z$, is no longer entire on $\mathbb R$. In fact, the Taylor series at $x=0$ for both $\log(1+x^2)$ and $\arctan x$ have a finite convergence radius of $1$ (can be checked easily) and not $\infty$.
\end{remark}

\begin{definition}\label{defB.2}
If the joint characteristic functional $F(t,x)$ in \eqref{eq:defF} is continuous on $[0,T]\times\mathbb R$ and $F(t,x)\neq0$ for all $(t,x)\in[0,T]\times\mathbb R$, we define $\log F$ via Lemma \ref{log} as the function $g$ such that $\exp(g)=F$ and $g(T,0)=0$.
\end{definition}

\begin{remark}
Notice that by the definition of $F$ in \eqref{eq:defF}, $F(T,0)=1$. So we should have $\exp(g(T,0))=1$. In particular, we can choose $g(T,0)=0$, and in this case, by Lemma \ref{log}, $g$ is fully specified.
\end{remark}

\section{Cost comparison: Fourier vs.~Monte Carlo methods}\label{cost_comparison}

We compare here the computation times for the SPX implied volatility surface between the Fourier and Monte Carlo methods on an i7 6-core computer.

\textit{Quintic OU model:} We fix the model parameters used in the numerical experiment in Sect.~\ref{subsubsec:quintic_ou} and take the maturity $T=2$. Table~\ref{time_quintic_fourier} shows the computation times using Fourier for different truncation levels $M$. Table~\ref{time_quintic_mc} shows the computation times using Monte Carlo for various numbers of simulations and time steps.

\textit{Bergomi model:} We fix the model parameters used in the numerical experiment in Sect.~\ref{subsubsec:bergomi} and take the maturity $T=2$. Table~\ref{time_bergomi_fourier} shows the computation times using Fourier for different truncation levels $M$. Table~\ref{time_bergomi_mc} shows the computation times using Monte Carlo for various numbers of simulations and time steps.

\begin{table}[H]
\centering
\begin{tabular}{lrrr}
\toprule
Truncation $M$ & 12 & 22 & 32\\
\midrule
Time (in seconds) & 0.043 & 0.047 & 0.111\\
\bottomrule
\end{tabular}
\caption{Computation times (in seconds) for the quintic OU model for maturity $T=2$ using Fourier, with the same model parameters as in Fig.~\ref{quintic_ou_riccati}.}
\label{time_quintic_fourier}
\end{table}

\begin{table}[H]
\centering
\begin{tabular}{lrrrr}
\toprule
Nr.~simulations & \multicolumn{4}{c}{Nr.~time steps}\\
\cmidrule(lr){2-5}
& 10 & 100 & 1,000 & 10,000\\
\midrule
100 & 0.008 & 0.018 & 0.014 & 0.08\\
1,000 & 0.027 & 0.025 & 0.082 & 0.508\\
10,000 & 0.125 & 0.167 & 0.541 & 5.042\\
100,000 & 1.360 & 1.733 & 5.583 & 57.734\\
\bottomrule
\end{tabular}
\caption{Computation times (in seconds) for the quintic OU model for maturity $T=2$ using Monte Carlo, with the same model parameters as in Fig.~\ref{quintic_ou_riccati}.}
\label{time_quintic_mc}
\end{table}

\begin{table}[H]
\centering
\begin{tabular}{lrrr}
\toprule
Truncation $M$ & 16 & 35 & 50\\
\midrule
Time (in seconds) & 0.059 & 0.142 & 0.219\\
\bottomrule
\end{tabular}
\caption{Computation times (in seconds) for the Bergomi model for maturity $T=2$ using Fourier, with the same model parameters as in Fig.~\ref{bergomi_ou_riccati}.}
\label{time_bergomi_fourier}
\end{table}

\begin{table}[H]
\centering
\begin{tabular}{lrrrr}
\toprule
Nr.~simulations & \multicolumn{4}{c}{Nr.~time steps}\\
\cmidrule(lr){2-5}
& 10 & 100 & 1,000 & 10,000\\
\midrule
100 & 0.015 & 0.008 & 0.016 & 0.076\\
1,000 & 0.031 & 0.025 & 0.052 & 0.399\\
10,000 & 0.143 & 0.148 & 0.457 & 3.902\\
100,000 & 1.384 & 1.650 & 4.370 & 42.624\\
\bottomrule
\end{tabular}
\caption{Computation times (in seconds) for the Bergomi model for maturity $T=2$ using Monte Carlo, with the same model parameters as in Fig.~\ref{bergomi_ou_riccati}.}
\label{time_bergomi_mc}
\end{table}

\section{Another example of model calibration via Fourier}\label{C:_more_calib_results}

Here, we provide more calibration examples under the quintic OU and the one-factor Bergomi models on the SPX and VIX volatility surface on 9 November 2021.

\begin{figure}[H]
\centering
\includegraphics[width=0.65\textwidth]{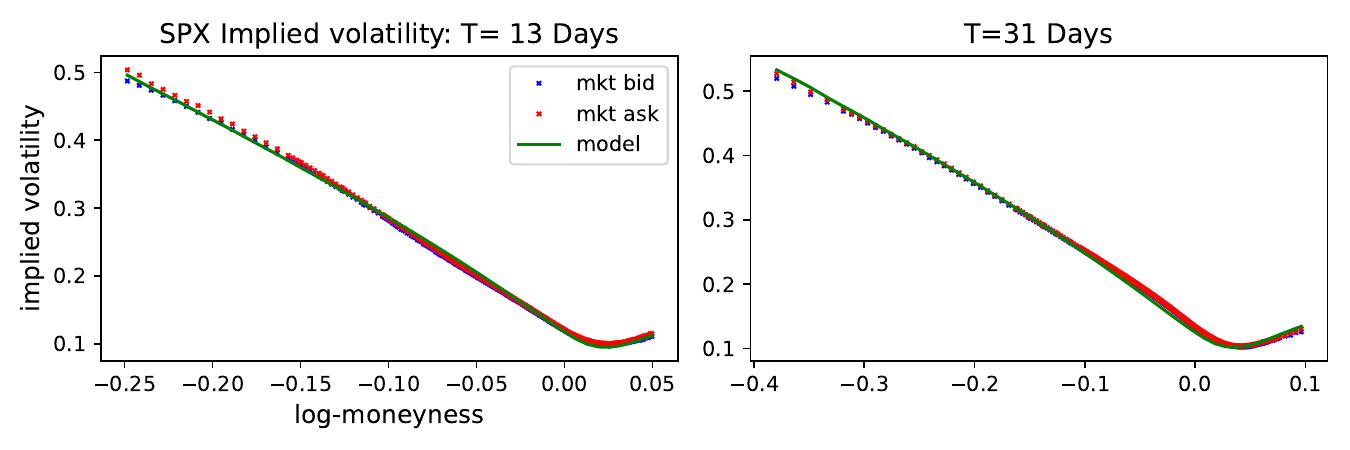}
\includegraphics[width=0.65\textwidth]{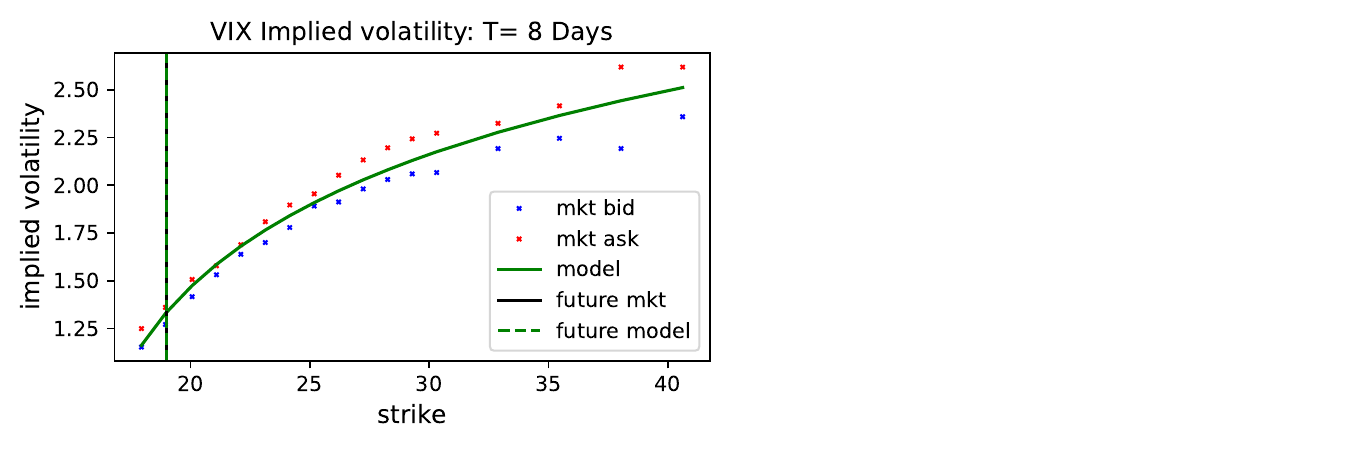}
\caption{Quintic OU model (green lines) jointly calibrated to the SPX and VIX smiles (bid/ask in blue/red) on 9 November 2021 via Fourier using the Nelder--Mead optimisation algorithm. The truncation level of the Riccati solver is set at $M=32$, with calibrated parameters $\rho=-0.6838$, $\alpha=-0.3914$, $(p_0,p_1,p_3,p_5)=(0.0062,0.4964,0.0939,0.0654)$. $\varepsilon$ is fixed upfront at 1/52 without calibration.}
\label{quintic_spx_calib_2}
\end{figure}

\begin{figure}[H]
\centering
\includegraphics[width=0.8\textwidth]{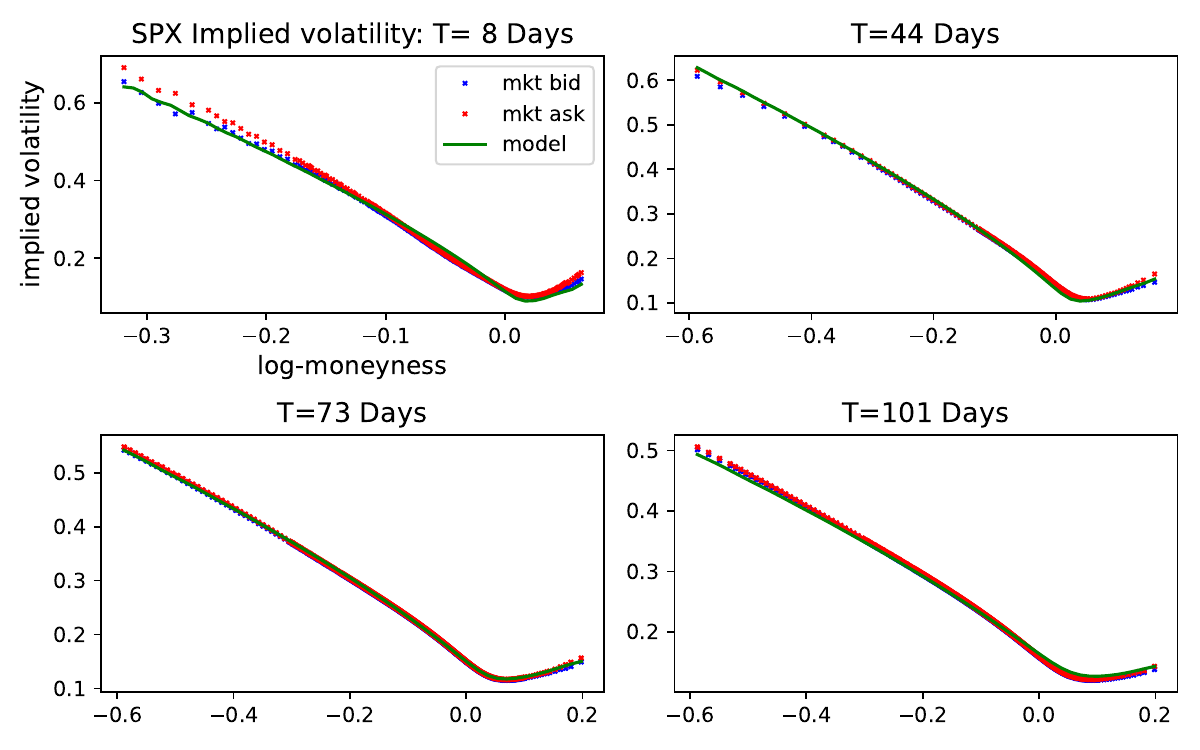}
\caption{One-factor Bergomi model (green lines) calibrated to the SPX smiles (bid/ask in blue/red) on 9 November 2021 via Fourier using the Nelder--Mead optimisation algorithm. The truncation level of the Riccati solver is set at $M=32$, with calibrated parameters $\eta=1.6002$, $\rho=-0.7214$, $\alpha=-0.5992$. $\varepsilon$ is fixed upfront at 1/52 without calibration.}
\label{bergomi_calib_2}
\end{figure}

\section*{Acknowledgements}
Eduardo Abi Jaber is grateful for the financial support from the Chaires FiME-FDD, Financial Risks, Deep Finance \& Statistics at Ecole Polytechnique. Xiaoyuan (Shaun) Li acknowledges the financial support from AXA Investment Managers (now part of BNP Paribas Asset Management) during his PhD at Universit\'e Paris 1 Panth\'eon-Sorbonne. Together, we should like to thank Alessandro Bondi, Louis-Amand G\'erard, Camille Illand, Sergio Pulido and Ning Tang for fruitful discussions, and two anonymous referees for their valuable feedback.

\section*{Declarations}
\textbf{Competing interests.} The authors declare no competing interests.

\bibliographystyle{plainnat}
\bibliography{bibl}
\end{document}